\documentclass{article}

\usepackage{graphicx}
\usepackage{fullpage}
\usepackage{subcaption}

\usepackage[
backend=biber,
style=alphabetic,
sorting=nyt,
maxnames=8,
maxalphanames=7,
backref=true,
url=false
]{biblatex}
\AtBeginBibliography{\emergencystretch=1em}

\usepackage{mathrsfs}
\usepackage{tabularx}
\usepackage{wrapfig}
\usepackage{booktabs}
\usepackage{quiver}
\usepackage{multirow}
\usepackage{float}
\usepackage{tikz}
\usetikzlibrary{positioning, arrows.meta, calc, decorations.pathreplacing}

\usepackage{amsthm}
\usepackage{thmtools}
\usepackage[disable]{todonotes}
\usepackage{suffix}
\newcommand{\nomi}[1]{\textcolor{orange}{[\textbf{Nomi:} #1]}}
\WithSuffix\newcommand\nomi*[1]{\todo[inline]{Nomi: #1}}
\newcommand{\nick}[1]{\textcolor{blue}{[\textbf{Nick:} #1]}}
\WithSuffix\newcommand\nick*[1]{\todo[inline]{Nick: #1}}
\newcommand{\eshan}[1]{\textcolor{red}{[\textbf{Eshan:} #1]}}
\WithSuffix\newcommand\eshan*[1]{\todo[inline]{Eshan: #1}}

\usepackage{nomimacros}
\usepackage[hidelinks,colorlinks=true,citecolor=PineGreen,linkcolor=blue,hypertexnames=false]{hyperref}
\usepackage[plain]{tcolorrestate}
\usepackage{cleveref}

\usepackage{subfiles}
\newcommand{\dobib}{
}

\title{Two-Sided Product Expanding Codes via  Rademacher Matrices}

\author{  Eshan Chattopadhyay\thanks{Supported by NSF Award CCF-2514586.}\\ Cornell University\\ \texttt{eshan@cs.cornell.edu}  \and  Noam Ringach \thanks{Supported by NSF GRFP grant DGE – 2139899  and NSF Award CCF-2514586.}   \\ Cornell University\\ \texttt{nomir@cs.cornell.edu} \and Nicholas Spooner \\ Cornell University \\ \texttt{nspooner@cornell.edu}}
\hypersetup{pdftitle={Two-Sided Product Expanding Codes via Rademacher Matrices},
  pdfauthor={Eshan Chattopadhyay, Noam Ringach, and Nicholas Spooner}}

\begin{document}

\date{}
\maketitle
\begin{abstract}
    Beginning with the work of Dinur, Lin, and Vidick (FOCS, 2024), tensor codes with constant product expansion have been foundational to recent advances in quantum locally testable codes (qLTCs) based on cubical complexes. Informally, product expansion says that any low-weight parity check of the tensor code can be written as the sum of few parity checks of each component code in each dimension.

   We give a proof of the existence of two-sided product expanding codes which, unlike the earlier result of Kalachev and Panteleev (FOCS, 2025), does not rely on explicit constructions of asymptotically optimal locally testable codes ($c^3$-LTCs). For every fixed number of component codes and dimensions whose rates are bounded away from zero and one, we show that independent random linear codes over a sufficiently large prime field have constant two-sided product expansion with probability tending to one. The tradeoff is that our proof requires the characteristic to grow with the block length, while [KP25] takes extensions of $\F_2$.

    To replace the use of $c^3$-LTCs, we develop several new techniques that we view as interesting in their own right. Instead of working directly over finite fields, we work over the reals and use random Rademacher matrices for the generator matrices of the component codes. 
    From here, we we show that the extendability of \ec sets can be reduced to controlling the operator norm of sparse restrictions of carefully chosen Gram matrices. Applying the trace power method to bound this norm reduces to bounding the number of possible labelings of certain closed walks on bipartite graphs, which we bound using the sparsity of the operators and bounds on the number of equivalence classes of weak Wigner words from Anderson and Zeitouni (Probab. Theory Relat. Fields., 2006), which were originally applied to band matrices. Transferring our result from the reals to finite fields is what causes our result to require a large field characteristic. 

    Due to our techniques not relying on explicit $c^3$-LTCs, we believe that they form a promising starting point towards showing the existence of product expanding tensor codes where the component codes have non-trivial automorphism groups and coboundary expanding codes that could be used as the local codes of non-cubical complexes, such as simplicial complexes. Overall, we offer a different random model and novel techniques for studying the product expansion of random linear codes. 
\end{abstract}
\newpage
{\small
\tableofcontents
}
\newpage

\section{Introduction}\label{sec:introduction}

A \emph{linear error-correcting code} is a subspace $\cC\subseteq\F_q^n$ in which any pair of distinct elements $c \neq c' \in \cC$ satisfy $\frac{1}{n}\abs{c-c'} \geq \delta$, where $\delta$ is known as the \emph{relative distance} of the code. One desirable property of a linear error-correcting code is for it to be \emph{locally testable}: a randomized function $T:\F_q^n\to\{\textsf{Accept}, \textsf{Reject}\}$ such that  $T(c)=\textsf{Accept}$ for all $c\in \cC$ and $\Pr[T(w)=\textsf{Reject}]\geq \frac{s}{n} \cdot \operatorname{dist}(w,\cC)$, where $\operatorname{dist}(w,\cC)$ is the Hamming distance of $w$ to $\cC$. The \emph{locality} of $T$ is how many coordinates of $w$ it reads and $s$ is its \emph{soundness}. A family of locally testable codes $(\cC_n)_{n \in \N}$ is known as a $c^3$-LTC if its parameters are asymptotically optimal, i.e., its relative distance $\delta$, rate $\rho = \dim(\cC)/n$, soundness $s$ and locality are all constant. One can also define an analogous notion of local testability for \emph{quantum} codes; codes satisfying this property are referred to as qLTCs.
Classical LTCs and qLTCs have found important applications throughout theoretical computer science \cite{goldreich_locally_2006,barak_making_2015,eldar_local_2017,kumar_relaxed_2025,nguyen_quantum_2025}; in particular, LTCs are a key component in the original proof of the PCP theorem.

Recent constructions of classical LTCs \cite{dinur_locally_2022,panteleev_asymptotically_2022}, qLTCs \cite{dinur_expansion_2024}, and quantum low-density parity-check codes (qLDPCs) \cite{panteleev_asymptotically_2022,dinur_good_2022}, all use variants of the same strategy whereby a high-dimensional expander (HDX) is used to lift a small, constant-sized code with good product expansion, as in \Cref{def:intro-product-expansion}, to create a large code with good parameters. These constructions fall into the framework of sheaf codes, a generalization of Tanner codes, first formally defined by \cite{panteleev_maximally_2024}. Just like a Tanner code, a sheaf code only needs two pieces of information: a large, explicit (possibly high-dimensional) graph with structured local neighborhoods and a collection of constant-sized``local'' codes that can be placed at these neighborhoods on the graph.

Sheaf codes have already been implicitly used extensively for constructing linear codes \cite{dinur_high_2017,breuckmann_balanced_2021,meshulam_graph_2018,panteleev_asymptotically_2022,dinur_expansion_2024,dinur_new_2023,first_good_2024,kalachev_two-sided_2023}. The main benefit of the sheaf code framework is that it is flexible enough to encompass many different constructions of linear codes while being endowed with a formal cohomology theory that allows for the study of the cosystolic distance and cocycle expansion of the resulting codes, which are directly related to the relative distance and soundness of the resulting locally testable code.

The HDXs used in the current qLDPC \cite{panteleev_asymptotically_2022} and qLTC \cite{dinur_expansion_2024,chen_cubical_2026,gay_asymptotically_2026,bafna_good_2026} constructions is a cubical complex which requires that the local codes they utilize be tensor codes; consequently, there has been much recent work on the coboundary expansion of tensor codes. The equivalent notion of coboundary expansion used in these works, and that we will focus on here, is that of \emph{product expansion}. For two component codes of linear distance, product expansion is equivalent, up to a change in parameters, to robust and agreement testability. For more than two codes, it is a strictly stronger requirement~\cite{kalachev_two-sided_2023,kalachev_high-dimensional_2023}.

\paragraph{Product expansion}
 Let $\cC_1,\dots,\cC_D\subseteq\F_q^n$ be linear codes and view words of length $n^D$ as arrays indexed by the grid $[n]^D$. For $i\in[D]$, let
\[
    \cC^{(i)}=(\F_q^n)^{\otimes(i-1)}\otimes\cC_i\otimes(\F_q^n)^{\otimes(D-i)}.
\]
Thus, a word belongs to $\cC^{(i)}$ precisely when its restriction to each line parallel to the $i$-th axis belongs to $\cC_i$. We define the \emph{tensor code} and the \emph{check-product code}, respectively, as
\[
    \cC_1\otimes\cdots\otimes\cC_D=\bigcap_{i=1}^D\cC^{(i)},
    \qquad
    \cC_1\boxplus\cdots\boxplus\cC_D=\sum_{i=1}^D\cC^{(i)}.
\]
An easy way to visualize the tensor of two codes $\cC_1\otimes\cC_2$ is by labeling the grid $[n]^2$ by field elements of $\F_q$ such that each row is a codeword of $\cC_2$ and each column is a codeword of $\cC_1$. 
The check-product code is related to the tensor code by duality:
\[
    (\cC_1\otimes\cdots\otimes\cC_D)^\perp
    =\cC_1^\perp\boxplus\cdots\boxplus\cC_D^\perp.
\]

A \emph{line-codeword} is a word supported on a single axis-parallel line whose restriction to that line belongs to the corresponding component code. Every word in the check-product code is created by a sum of such words. Product expansion then asks whether every check-product codeword can be written using few nonzero line-codewords relative to its Hamming weight.

For $a\in\F_q^{[n]^D}$, let $\abs{a}_H$ denote its Hamming weight and set $\norm{a}_H=\abs{a}/n^D$. Let $\abs{a}_i$ be the number of lines parallel to the $i$-th axis on which $a$ is nonzero, and set $\norm{a}_i=\abs{a}_i/n^{D-1}$.
\begin{definition}[Product expansion]\label{def:intro-product-expansion}
    A collection $(\cC_1,\dots,\cC_D)$ is \emph{$\nu$-product expanding} if every $c\in\cC_1\boxplus\cdots\boxplus~\cC_D$ has a representation $c=\sum_{i=1}^D a_i$, where $a_i\in\cC^{(i)}$, such that
    \[
        \nu\sum_{i=1}^D\norm{a_i}_i\leq\norm{c}_H.
    \]
    The collection is \emph{two-sided $\nu$-product expanding} if both $(\cC_1,\dots,\cC_D)$ and $(\cC_1^\perp,\dots,\cC_D^\perp)$ are $\nu$-product expanding.
\end{definition}
Equivalently, a word of Hamming weight $w$ in the check-product code has a decomposition using at most $w/(\nu n)$ nonzero line-codewords. 
Conceptually, product expansion controls how the component codes interact and forces them to act independently of each other. In applications, it is necessary for $\nu$ to remain a positive constant as the component length $n$ grows, with $D$ fixed.

Kalachev and Panteleev~\cite{kalachev_two-sided_2023} proved that pairs of random linear codes have constant two-sided product expansion, while \cite{dinur_good_2022} showed that pairs of random linear codes are robust. However, the almost-good qLTC construction of~\cite{dinur_expansion_2024} and subsequent work on quantum codes~\cite{golowich_quantum_2025,nguyen_quantum_2025} require product expansion for collections of more than two codes. The breakthrough work of Kalachev and Panteleev~\cite{kalachev_maximally_2025} showed that random linear codes over sufficiently large fields have constant two-sided product expansion for any fixed number of component codes. In particular, they obtained the following existence theorem.
\begin{theorem}[Corollary 3 of~\cite{kalachev_maximally_2025}]\label{thm:KP25 main thm}
    Fix $D\geq1$ and intervals $I_1,\dots,I_D\subseteq(0,1)$ with nonempty interior. There exists $\nu>0$ such that, for all sufficiently large $n$, there are codes $\cC_1,\dots,\cC_D\subseteq\F_{2^{t_n}}^n$, where $t_n=(n+1)^D$, with $\dim(\cC_i)/n\in I_i$ such that the collection is two-sided $\nu$-product expanding.
\end{theorem}
Curiously, Kalachev and Panteleev's proof of this result relies on the existence of $c^3$-LTCs, as established by the explicit, deeply structured constructions of \cite{dinur_locally_2022,panteleev_asymptotically_2022}. For the reader who has seen many probabilistic existential proofs of combinatorial objects, it should be quite surprising (and interesting!) that the proof of \cite{kalachev_maximally_2025} relies on an explicit construction to obtain an existential result. The main contribution of this work is a ``purely'' probabilistic proof of (a variant of) \Cref{thm:KP25 main thm}, thereby providing new insight into the phenomenon of product expansion for random codes, and new analytical tools that may allow us to ``go beyond'' random codes. We expand upon these motivations next.

\subsection{Motivation}
Besides pure curiosity, there are two main motivations for this work.
First, to obtain good two-sided product expanding codes where the component codes have non-trivial automorphism groups, which are relevant in the constructions of qLTCs \cite{panteleev_maximally_2024} and applicable whenever tensor codes are being used as local codes on cubical complexes. Second, to prove the existence of good two-sided coboundary expanders that can be used as the local codes of other complexes, like simplicial HDXs. We elaborate on both of these goals below, but, before doing so, we recall the strategy of \cite{kalachev_maximally_2025} to explain the barriers in adapting their argument.

\paragraph{The technique of \cite{kalachev_maximally_2025}}

We begin by briefly recalling the role of \emph{extendability}. Given a tensor code $\cC=\cC_1\otimes\cdots\otimes\cC_D$ and a set $M\subseteq[n]^D$, an assignment $x\in\F_q^M$ is locally valid if for every axis-parallel line $\ell$ entirely contained in $M$, $x|_\ell\in\cC_i$ for the corresponding component code $\cC_i$. We call $M$ \emph{extendable} in $\cC$ if \emph{every} such assignment extends (not necessarily uniquely) to a codeword of $\cC$ on the whole grid. A dual view to this is that the parity checks contained in $M$ impose all the conditions needed for a global extension. In particular, if $M$ contains no entire line, all assignments are locally valid, so extendability means that any arbitrary assignment of elements from $\F_q$ on $M$ extends.

The sets relevant to product expansion are the \emph{$\varepsilon$-closed} sets: sets for which every axis-parallel line is either contained in $M$ or meets it in at most $\varepsilon n$ points. By duality, extendability of $M$ in $\cC$ is equivalent to the following property: every word in $\cC^\perp=\cC_1^\perp\boxplus\cdots\boxplus\cC_D^\perp$ supported on $M$ can be generated by dual line-codewords on lines entirely contained in $M$. Such sets are called \emph{inner-generated}. Kalachev and Panteleev show that if every $\varepsilon$-closed set has this property, then the collection $(\cC_1^\perp,\dots,\cC_D^\perp)$ has constant product expansion, with a parameter depending only on $\varepsilon$ and $D$~\cite{kalachev_maximally_2025}. Applying the statement to both a collection and its dual gives two-sided product expansion.

Their proof obtains the necessary extendability from \emph{maximally extendable} (ME) codes. Within a specified family, an ME code has the property that every set extendable in any member of the family is also extendable in the ME code. The family used in~\cite{kalachev_maximally_2025} is parametrized by parity-check matrices $H_i\in\F_2[\bX]^{(n-k_i)\times n}$, with a separate variable for every entry. We call a setting of the variables in $\bX$ a \emph{specialization}, which in turn gives concrete codes $\cC_1,\dots,\cC_D$ as defined by these parity check matrices. Specializations of full row rank give component codes of dimensions $k_i$, and random specializations over a sufficiently large extension field give an ME tensor code with high probability.
While this setup is equivalent to simply taking a random parity-check matrix, we introduce this notion of parity-check matrices parametrized as polynomials from \cite{panteleev_maximally_2024} since they allow us to write the determinants of compositions of these matrices as polynomials in $\bX$ themselves, and are more flexible, as we elaborate below.

More generally, \cite{panteleev_maximally_2024} proves an analogous statement for polynomially parametrized families of sheaf codes, where the polynomials used to define the parity check matrices can be more elaborate than simply having a distinct variable for each entry. Importantly, an ME tensor code can inherit, with a constant loss in parameters, the product expansion of any code in its parametrized family. 

Nevertheless, one still needs a member of the family from which the ME code can inherit product expansion, which we call a \emph{witness}. Kalachev and Panteleev obtain this witness from the explicit $c^3$-LTCs of~\cite{dinur_locally_2022}, namely LTCs with constant rate, relative distance, and query complexity. After adapting lengths and rates, they show that suitable collections derived from these LTCs have constant product expansion from which ME codes can inherit their own product expansion. Lastly, the observation that the dual of a uniformly random linear code is itself uniformly random lets them union bound the two failure probabilities to obtain two-sided product expansion.

\paragraph{Towards product expanding codes with symmetries}
Even with such two-sided product expanding codes, the qLTCs of~\cite{dinur_expansion_2024} have constant rate but inverse-polylogarithmic relative distance and soundness. In~\cite{panteleev_maximally_2024}, the predecessor to~\cite{kalachev_maximally_2025}, the authors suggest using the Ramanujan cubical complexes of Rungtanapirom, Stix, and Vdovina~\cite{rungtanapirom_infinite_2019} as a possible route to asymptotically good qLTCs. The proposed approach calls for local codes that are equivariant with respect to the symmetries of the cubical complex. That is, the collection of component codes must be two-sided product expanding, and the automorphism group of each component must also contain the corresponding local permutation group of the complex. The random model of~\cite{kalachev_maximally_2025} cannot enforce this symmetry requirement since random linear codes have trivial automorphism groups with high probability \cite{li_automorphism_2026}. 

We encounter two immediate barriers to using this technique to show the existence of two-sided product expanding codes with prescribed nontrivial symmetries. First, the parametrization of the parity-check matrices allows arbitrary component codes $\cC_i$ instead of just ones with the desired automorphism groups. One can try to restrict the parametrization so that the resulting codes have the required symmetries, and, when such a polynomial parametrization is available,~\cite{panteleev_maximally_2024} guarantees an ME code within that family. However, this brings us to the second barrier: the $c^3$-LTC construction of~\cite{dinur_locally_2022} does not provide the prescribed automorphism groups. Consequently, the extension witness obtained from these LTCs is not known to lie in the restricted family, and maximal extendability in that family has no established witness from which to inherit the required properties.

In our work, we remove the dependence on $c^3$-LTCs that gives rise to this second barrier. We prove the necessary extension properties directly, using random matrices over the reals and then transferring the result to finite fields. This gives a different random model and novel techniques for studying the product expansion of random linear codes. 

\paragraph{Towards the coboundary expansion of general sheaf codes}
As we elaborate more in \Cref{sec:open questions}, our techniques are flexible enough that they provide a new starting point to analyze the coboundary expansion of more general sheaf codes, as defined by \cite{panteleev_maximally_2024}, instead of just tensor codes. In particular, recent constructions of classical \cite{dinur_new_2023} and quantum \cite{gulshen_quantum_2025} codes on simplicial high-dimensional expanders require local codes that do not have a tensor product structure and currently must use Reed-Solomon or Reed-Müller codes. We believe our techniques could lend themselves to analyzing the sheaf codes for the required local codes of these constructions. We also anticipate that our new framework of analyzing over the reals the desired properties of a code generated by random Rademacher matrices will be useful more broadly.

\subsection{Our results}
We obtain the following theorem over a prime field. In the statement, a uniformly random code of dimension $k$ means a uniformly random $k$-dimensional subspace of the ambient space.
\begin{theorem}[Simplified version of~\Cref{thm:main two sided prod expanding}]\label{thm:intro-main}
    Fix $D\geq1$ and $0<\rho<1/2$. There are constants $\nu>0$ and $n_0$, depending only on $D$ and $\rho$, such that the following holds. For every $n\geq n_0$ and integers $k_1,\dots,k_D$ satisfying $\rho n\leq k_i\leq(1-\rho)n$, there is a prime
    \[
        p=O\!\left(2^{n^D}n^D\log n\right)
    \]
    such that independent uniformly random codes $\cC_i\subseteq\F_p^n$ of dimensions $k_i$ are two-sided $\nu$-product expanding with probability at least $1-1/\log(n+1)$. The constant in the field-size bound depends only on $D$.
\end{theorem}
The full statement, including an explicit expansion parameter and a sufficient threshold for $n$, appears in~\Cref{thm:main two sided prod expanding}. In particular, we obtain codes whose rates lie in any  intervals with nonempty interior in $(0,1)$. Thus, we recover the existence of two-sided product expanding codes without the use of $c^3$-LTCs. Unlike~\Cref{thm:KP25 main thm}, our proof requires the characteristic to grow with $n$ and does not establish the result in characteristic two. In applications, such as \cite{dinur_expansion_2024}, this is not a serious issue: it simply results in a code over some constant-sized alphabet instead of the binary alphabet.

Our main task is to prove extendability for every $\varepsilon$-closed set. We first work over the reals, using independent Rademacher entries for the generator matrices. By dualizing and showing that the given \ec set is inner-generated, we reduce the problem to showing a certain map has norm less than one.\footnote{Specifically, this is the plane-interaction matrix as defined in \Cref{eq:overview-plane interaction}.}

To bound its norm, we apply the trace power method. Leveraging the unique properties of Rademacher random variables turns the norm estimate into a problem of counting labelings of closed walks on bipartite graphs. The work of Anderson and Zeitouni \cite{anderson_clt_2006} bounds the number of equivalence classes of such walks, known as weak Wigner words. We then bound the size of each equivalence class using the fact the original set is \ec. This argument shows that each \ec set has a collection of real generator matrices for which it is inner-generated for the corresponding collection of codes. We transfer this result from the reals back to finite fields and to a single collection of codes, instead of a different collection of codes for each \ec set, by a Schwartz-Zippel argument similar to that of \cite{kalachev_maximally_2025}. We explain these arguments in~\Cref{sec:proof-overview}.

\paragraph{Concurrent work}
In the last stages of preparing our manuscript, a number of preprints have been posted that give constructions of asymptotically optimal qLTCs. The works \cite{gay_asymptotically_2026,bafna_good_2026,chen_cubical_2026} all use cubical complexes with local Reed-Solomon codes to construct asymptotically optimal qLTCs. \cite{cohen_good_2026} constructs asymptotically optimal qLTCs under a conjecture on the existence of certain lossless cubical complexes while not employing local codes at all.

\paragraph{Organization.}
After the proof overview, \Cref{sec:preliminaries} introduces necessary preliminaries and our framework of planes and anchor sets. In \Cref{sec:duality-atoms} we derive the restriction and corestriction maps of tensor and check-product codes, respectively, and define the atom decomposition. \Cref{sec:two-sided-product-expansion} proves our main result based on the injectivity of a universal corestriction map. Then, \Cref{sec:injectivity of cores via K} reduces this injectivity to an operator norm bound over the reals and gives the transfer to finite fields. Finally, \Cref{sec:expected op norm of K} reduces this to a norm bound on sparse tensor matrices, which we prove in~\Cref{sec:bounding max opnorm with wigner} using weak Wigner words. We finish by proposing some future directions in \Cref{sec:open questions}.

\dobib

\section*{Acknowledgements}
N.R. would like to thank Robert Kleinberg for conversations that inspired him to consider conducting the analysis in this paper over the reals instead of over a finite field. We thank the Simons Institute for the Theory of Computing at UC Berkeley for providing a stimulating research environment. Part of this work was done during visits by the authors to the Institute.

\section*{AI Methodology}
The ideas, proofs and exposition in this work are almost entirely the product of human effort. ChatGPT Pro 5.6 and 6 were used to help check correctness and for copy editing, and suggested the connection to the works of \cite{wigner_characteristic_1955,furedi_eigenvalues_1981,anderson_clt_2006}.

The authors take responsibility for all results, content, and the presentation of this paper. 

\dobib

\section{Proof Overview}\label{sec:proof-overview}
We now explain the main ideas behind~\Cref{thm:intro-main}. Fix the number of codes $D$ and a constant $\rho>0$, and assume that each code has dimension between $\rho n$ and $(1-\rho)n$. Our goal is to prove $\nu$-product expansion with $\nu=\Omega_{D,\rho}(1)$. 

Our proof has four main parts:
\begin{enumerate}
    \item Using Lemma 19 and the existence of maximally extendable codes from \cite{kalachev_maximally_2025}, we reduce product expansion to showing that, for a fixed \ec set $M\subseteq[n]^D$, there is a collection of codes (over the reals) for which $M$ is inner-generated.
    
    \item We switch to working over $\R$ and observe that, for a fixed \ec $M$, it is inner-generated if the \emph{plane-interaction} matrix $K_M$, a centered Gram matrix built using the structure of $M$, has norm less than 1.
    
    \item By cleverly decomposing $K_M$ into a constant number of blocks with easily describable norms and taking the generator matrices $G_i\in\R^{k_i\times n}$ of the codes $\cC_1,\dots,\cC_D$ to have Rademacher entries, we can bound the norm of $K_M$ by the expected value of Rademacher chaos polynomials.\footnote{A \emph{Rademacher chaos} is a homogeneous multivariate polynomial with domain $\{-1,1\}$.}
    
    \item Finally, we bound the expected value of these polynomials by using the fact that $M$ is \ec and tools from \cite{anderson_clt_2006} for weak Wigner words, ultimately giving us that $\opnorm{K_M}<1$ so $M$ is inner-generated.

\end{enumerate}
We now elaborate on these steps in detail.

\subsection{Reducing to showing a fixed \texorpdfstring{\ec}{eps-closed} set is inner-generated}

As mentioned in the introduction, two important results from \cite{kalachev_maximally_2025} are that the product expansion of a collection of codes is implied by every \ec set being inner-generated (Lemma 19) and that ME codes exist. Kalachev and Panteleev used maximal extendability by showing the fact that a collection of $c^3$-LTCs from \cite{dinur_locally_2022} has every \ec set inner-generated, so an ME code would have each of these sets inner-generated as well.\footnote{While the definition of maximally extendable codes is with regards to the extendability of a set, since extendable sets and inner-generated sets are dual notions to each other, we can equivalently consider their result on the existence of maximally extendable codes as one for maximally inner-generated codes.} Instead of showing that there exists a single collection of codes with all \ec sets inner-generated, we show that for each fixed \ec set $M$ there is a collection of codes for which $M$ is inner-generated. Then an ME code can inherit the inner-generatedness of each \ec set from different collections of codes instead of a single one. The difference in these strategies is summarized in \Cref{fig:witness-strategies}

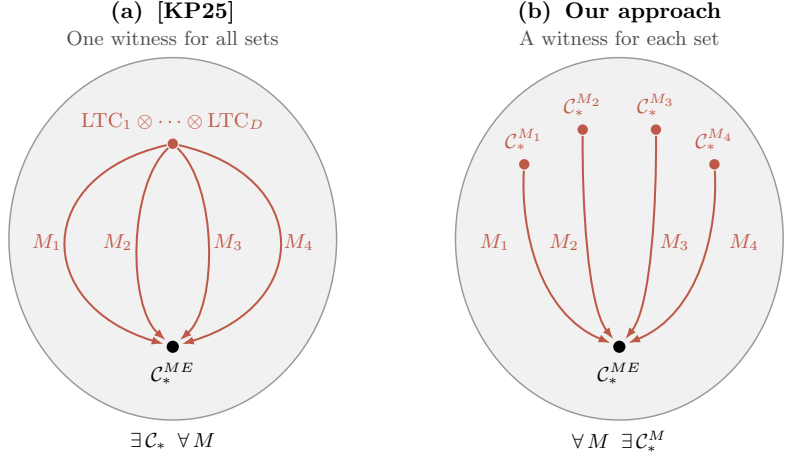
\begin{figure}[ht]
    \centering
    \scalebox{0.85}{\begingroup
\providecommand{\StrategyPriorLabel}{[KP25]}
\providecommand{\cC}{\mathcal{C}}
\providecommand{\StrategyOurTarget}{\cC_*^{ME}}

\definecolor{strategyRed}{RGB}{190,89,72}
\definecolor{strategyGray}{gray}{0.945}

\begin{tikzpicture}[
    x=1cm,y=1cm,
    font=\small,
    family/.style={
        draw=black!40,fill=strategyGray,line width=0.65pt
    },
    inherit/.style={
        draw=strategyRed,line width=0.9pt,
        -{Latex[length=1.9mm,width=1.3mm]},
        shorten <=2.8pt,shorten >=4.1pt
    },
    witness/.style={
        circle,fill=strategyRed,draw=none,
        inner sep=0pt,minimum size=4.6pt,outer sep=0pt
    },
    common/.style={
        circle,fill=black,draw=none,
        inner sep=0pt,minimum size=5.4pt,outer sep=0pt
    },
    set label/.style={text=strategyRed,inner sep=1pt},
    code label/.style={text=strategyRed,inner sep=1pt},
    panel title/.style={font=\normalsize\bfseries,inner sep=0pt},
    subtitle/.style={font=\small,text=black!70,inner sep=0pt}
]

\begin{scope}
    \node[panel title] at (0,3.52) {(a)\enspace\StrategyPriorLabel};
    \node[subtitle] at (0,3.12) {One witness for all sets};
    \path[family] (0,0) ellipse [x radius=2.55,y radius=2.82];

    \coordinate (ltc) at (0,1.48);
    \coordinate (me-left) at (0,-1.68);
    \node[code label] at (0,1.83)
        {$\mathrm{LTC}_1\otimes\cdots\otimes\mathrm{LTC}_D$};

    \draw[inherit] (ltc)
        .. controls (-2.18,0.98) and (-2.18,-1.02) .. (me-left);
    \draw[inherit] (ltc)
        .. controls (-0.70,0.82) and (-0.70,-1.04) .. (me-left);
    \draw[inherit] (ltc)
        .. controls ( 0.70,0.82) and ( 0.70,-1.04) .. (me-left);
    \draw[inherit] (ltc)
        .. controls ( 2.18,0.98) and ( 2.18,-1.02) .. (me-left);

    \node[set label] at (-1.96,-0.05) {$M_1$};
    \node[set label] at (-0.86,-0.05) {$M_2$};
    \node[set label] at ( 0.86,-0.05) {$M_3$};
    \node[set label] at ( 1.96,-0.05) {$M_4$};

    \node[witness] at (ltc) {};
    \node[common] at (me-left) {};
    \node at (0,-2.10) {$\cC_*^{ME}$};
    \node at (0,-3.15) {$\exists\, \cC_*\;\;\forall\, M$};
\end{scope}

\begin{scope}[xshift=6.95cm]
    \node[panel title] at (0,3.52) {(b)\enspace Our approach};
    \node[subtitle] at (0,3.12) {A witness for each set};
    \path[family] (0,0) ellipse [x radius=2.55,y radius=2.82];

    \coordinate (w1) at (-1.48,1.16);
    \coordinate (w2) at (-0.57,1.70);
    \coordinate (w3) at ( 0.57,1.70);
    \coordinate (w4) at ( 1.48,1.16);
    \coordinate (me-right) at (0,-1.68);

    \draw[inherit] (w1)
        .. controls (-1.54,0.15) and (-1.25,-1.16) .. (me-right);
    \draw[inherit] (w2)
        .. controls (-0.58,0.61) and (-0.50,-1.04) .. (me-right);
    \draw[inherit] (w3)
        .. controls ( 0.58,0.61) and ( 0.50,-1.04) .. (me-right);
    \draw[inherit] (w4)
        .. controls ( 1.54,0.15) and ( 1.25,-1.16) .. (me-right);

    \node[set label] at (-1.94,-0.05) {$M_1$};
    \node[set label] at (-0.86,-0.05) {$M_2$};
    \node[set label] at ( 0.86,-0.05) {$M_3$};
    \node[set label] at ( 1.94,-0.05) {$M_4$};

    \foreach \i in {1,...,4}{
        \node[witness] at (w\i) {};
        \node[code label,above=5pt] at (w\i) {$\cC_*^{M_\i}$};
    }
    \node[common] at (me-right) {};
    \node at (0,-2.10) {$\StrategyOurTarget$};
    \node at (0,-3.15) {$\forall\, M\;\;\exists\, \cC_*^M$};
\end{scope}

\end{tikzpicture}
\endgroup}
    \caption{Comparison of strategies between \cite{kalachev_maximally_2025} in (a) and our work in (b). Each gray ellipse represents the family of polynomially parametrized codes with $\cC_*^{ME}$ as the maximally extendable code in the family. In (a), each \ec set $M_i$ is extendable in the tensor product of LTCs, so $\cC_*^{ME}$ can use this product as its witness of extendability for all \ec sets. In (b), each \ec $M_i$ has a different collection of codes $\cC_*^{M_i}$ in which it is extendable. Here, $\cC_*^{ME}$ still inherits the extendability of all \ec sets, but from different witnesses.}
    \label{fig:witness-strategies}
\end{figure}

In reality, our proof technique has one major difference that does not allow us to apply the aforementioned results from \cite{kalachev_maximally_2025} in a black-box manner. For a fixed \ec set $M\subseteq[n]^D$, the collection of codes that we show witness the fact that $M$ is inner-generated is defined by real generator matrices $G_i\in\R^{k_i\times n}$ instead of over $\F_{2^t}$ as in \cite{kalachev_maximally_2025}. Loosely speaking, to transfer this result from the reals to a finite field, we use the fact (shown in the next step) that $M$ is inner-generated if and only if a certain matrix depending on the generator matrices $(G_i)_{i\in[D]}$ is full rank. By taking the determinant polynomial $f_M\in\Z[\bX]$ of a non-singular minor in this matrix---which we know is not identically the zero polynomial since we have real generator matrices that are not a zero of $f_M$---we can condition $M$ being inner-generated on choosing generator matrices $G_i\in\F_p^{k_i\times n}$ such that they are not a zero of $f_M$. Finally, bounding $\deg f_M$ as $O(n^D)$ and noticing that there are at most $2^{n^D}$ sets $M\subseteq[n]^D$, we use a Schwartz-Zippel argument to show that taking $p=O(2^{n^D}n^D)$ will guarantee some setting of the generator matrices for which $\prod_{\varepsilon\text{-closed }M} f_M((G_i)_{i\in[D]})$ is non-zero, meaning that all \ec sets are inner-generated.

We elide the details of this argument since the techniques and, in particular, the application of Schwartz-Zippel, are very similar to those of \cite{kalachev_maximally_2025}. Our new observations and techniques are contained in the subsequent steps of our proof. The takeaway from this reduction is that, to prove \Cref{thm:intro-main}, we just need to show the following:
\begin{lemma}\label{lem:overview-existence of witness}
    Let $D\geq 1$, $\rho\in(0,1/2)$, $\varepsilon>0$ be small enough, and $n\in\N$ be large enough. Take integers $k_1,\dots,k_D$ with $\rho n\leq k_i\leq (1-\rho)n$ and fix an \ec set $M\subseteq[n]^D$. Then there exist full-rank generator matrices $G_1,\dots,G_D$ with $G_i\in\R^{k_i\times n}$ such that $M$ is extendable for $\cC_1\otimes\cdots\otimes\cC_D$ or, equivalently, inner-generated for $\cC_1^\perp\boxplus\cdots\boxplus\cC_D^\perp$.
\end{lemma}

\subsection{Proving an \texorpdfstring{\ec}{eps-closed} set is inner-generated by the norm of its plane-interaction matrix}
\paragraph{Random model}
We begin by defining the random Rademacher model we will use in the proof of \Cref{lem:overview-existence of witness}. To sample $G_i\in\R^{k_i\times n}$, we take it to have columns $g_i(j)\in\R^{k_i}$ indexed by $j\in[n]$ where $g_i(j)=k_i^{-1/2}(\sigma^{(i)}_{1,j},\dots,\sigma^{(i)}_{k_i,j})^\transpose$ and each $\sigma^{(i)}_{r,j}\sim\{-1,1\}$.  By construction, each column has unit norm. While this setup allows for the generator matrices to not necessarily be full-rank, it is easy to show that they will all be full-rank with high probability, so we will just concern ourselves with showing that $M$ is inner-generated. 

Our goal in the rest of this step is to show that $M$ is inner-generated for the codes from $G_1,\dots,G_D$ if \emph{plane-interaction matrix} $K_M$, constructed later in this step, has $\opnorm{K_M}<1$. 

\paragraph{Reducing to the injectivity of the corestriction map}
Take dual codes $\cB_i=\cC_i^\perp$ and define the check-product code $\cB=\cB_1\boxplus\cdots\boxplus\cB_D$. Recall that $M$ is inner-generated when every codeword of $\cB$ supported only on $M$ can be expressed as a sum of line-codewords on lines contained within $M$. Using the generator matrices $(G_i)_i$, we can define the natural map
\begin{align*}
    \cores:\gfield^{[n]^D}\to\bigotimes_{i=1}^D \gfield^{k_i}=:U_{[D]}, \qquad e_x\mapsto\bigotimes_{i=1}^D g_i(x_i) 
\end{align*}
where $\gfield$ is any field\footnote{While we will focus on the setting of $\gfield=\R$, these corestriction maps are well-defined over any field and, in particular, our ultimate goal of $\gfield=\F_p$.} and $e_x$ is the standard basis indexed by $x\in[n]^D$. By construction, the kernel of $\tau$ is exactly $\cB$. This is close to the map we want to construct indicating whether $M$ is inner-generated, but no information about $M$ is present in $\cores$, so it should not be surprising that we have to modify it to get our desired map. Observe that if we restrict $\cores$ to words only supported on $M$ and $M$ contains a line of $[n]^D$, say in the $i$-th direction, then $\cores$ is not injective because all the codewords from $\cB_i$ on that line are supported on $M$ and in the kernel of $\cores$. However, for $M$ to be inner-generated, we exactly want to allow words supported on $M$ generated by the sum of line-codewords. Formally, this is the set $R_M^{\text{line}}=\sum_{\ell\in\lines(M)}\cB_\ell$ where $\lines(M)$ is the set of axis-parallel lines contained in $M$ and, for an $i$-directional line $\ell\in\lines(M)$, we let $\cB_\ell$ be $\cB_i$ included into $\gfield^{[n]^D}$ by placing it on $\ell$ and padding with 0's elsewhere. We can then consider the induced quotient map
\begin{align*}
    \cores_M:\cosheaf(M):=\gfield^M/R_M^{\text{line}}\to U_{[D]},
\end{align*}
where we are using $\gfield^M$ as shorthand for words in $\gfield^{[n]^D}$ supported on $M$. Recalling that $\ker(\cores)=\cB$, the kernel of $\cores_M$ is exactly
\begin{align*}
    \ker(\cores_M)\cong\frac{\gfield^M\cap\ker(\cores)}{R_M^{\text{line}}}=\frac{\gfield^M\cap \cB}{\sum_{\ell\in\lines(M)}\cB_\ell}.
\end{align*}
Therefore, when $\cores_M$ is injective we get that $(\gfield^M\cap\cB)/(\sum_{\ell\in\lines(M)}\cB_\ell)=\ker(\cores_M)=0$, meaning that $\gfield^M\cap\cB=\sum_{\ell\in\lines(M)}\cB_\ell$, which is exactly the definition of $M$ being inner-generated.

An instructive special case to consider is when $M$ is \ec but contains no full lines, a property that we call \es. Consequently, $R_M^{\text{line}}=\{0\}$ and our goal becomes showing that the columns of $G_1\otimes\cdots\otimes G_D$ indexed by $M$, that is $\otimes_{i\in[D]} g_i(x_i)$ for $x\in M$, are linearly independent. However, when $M$ does contain full lines, we must handle the induced relations between them.

\paragraph{The local contributions of planes}
Instead of focusing on every single axis-parallel line in $M$ individually, we will bundle these into planes when possible. Formally, an $I$-plane at direction $I\subseteq[D]$ anchored at $a\in[n]^{\setcomp{I}}$ is defined as $P=\plane(I,a)=\{x\in[n]^D:\forall j\in\setcomp{I},x_j=a_j\}$. We assign to it a local space $U_P=U_I=\otimes_{i\in I}\R^{k_i}$. For a plane $Q=\plane(J,b)$ containing $P$, so $J\supseteq I$, we define the corestriction map
\begin{align*}
    \cores_{P\to Q}:U_P\to U_Q\qquad u\mapsto u\otimes\bigotimes_{j\in J\setminus I}g_j(a_j),
\end{align*}
with the tensor factors placed in the order of elements in $[D]$. In particular, we can always take $Q=[n]^D$ and consider the global corestriction map $\cores_P:=\cores_{P\to[n]^D}$ which takes a local contribution from $P$ and maps it to $U_{[n]^D}$. 

The \ec set $M$ can contain many planes---every $I$-plane in $M$ is made up of $n$ planes of dimension $\abs{I}-1$ also in $M$. We can simplify the structure of $M$ by focusing on the \emph{maximal} (by inclusion) planes, denoted as $\maxplanes(M)$, and the planes resulting from the intersection of maximal planes 
$\maxplanesclosed(M)=\left\{
        \bigcap_{P\in\mathcal S}P:
        \varnothing\ne\mathcal S\subseteq\maxplanes(M),\
        \bigcap_{P\in\mathcal S}P\ne\emptyset
        \right\}$.
A useful combinatorial consequence of this viewpoint is that, for a fixed direction $I\subseteq[D]$, the set of anchors $\maxAI(M)\subseteq[n]^{\setcomp{I}}$ of the maximal $I$-planes in $\maxplanes(M)$ is \es because $M$ is \ec. Indeed, if there were more than $\varepsilon n$ anchors from $\maxAI(M)$ on a line in some direction $j\notin I$, then the $\varepsilon$-closure of $M$ would force there to be a larger $(I\cup\{j\})$-plane containing the anchored planes, contradicting their maximality. If we introduce intersections of planes to consider the $I$-planes in $\maxplanesclosed(M)$, then the resulting anchor set $\maxAIclosed(M)$ is $(2^D\varepsilon)$-sparse via a similar argument. In either case, we use the $\varepsilon$-closure of $M$ to obtain sparseness guarantees for its anchor sets.

In particular, this viewpoint is helpful in dealing with the injectivity of $\cores_M$ because we can equivalently describe its domain $\cosheaf(M)$ as the direct sum of all $U_P$ for $P\in\maxplanesclosed(M)$ where every $u\in U_P$ is identified with its corestriction $\cores_{P\to Q}(u)\in U_Q$ whenever $P\subseteq Q$. Conceptually, this says that a contribution from a smaller plane contained in a larger one is equivalent to the same contribution from the larger one. This is shown formally in \Cref{lem:cosheaf(M) as quotient of direct sum}.

\paragraph{Removing interactions between planes via atoms}

In this equivalent definition of $\cosheaf(M)$, we see that a plane $Q$ contained in another plane $P$ contributes to the local space $U_P$ by the image $\cores_{Q\to P}(U_Q)\subseteq U_P$. To tease apart how much of the local space $U_P$ is controlled solely by $P$, we define the \emph{atom} $A_P$ of $P$ as
\begin{align*}
    U_{\subsetneq P}=\sum_{\substack{Q\in\maxplanesclosed(M)\\ Q\subsetneq P}} \cores_{Q\to P}(U_Q) \qquad \text{and} \qquad A_P=U_{\subsetneq P}^\perp\subseteq U_P.
\end{align*}
Importantly, since we repeat this procedure for every plane, the direct sum of all atoms $\bigoplus_{P\in\maxplanesclosed(M)}A_P$ spans $\cosheaf(M)$, which we formally show in \Cref{lem:direct sum of atoms spans cosheaf}.

The main utility of this atom decomposition is that it makes clear the fact that non-intersecting planes are the only obstructions to $M$ being inner-generated. In particular, we observe that the atoms of two distinct intersecting planes have orthogonal images in $U_{[D]}$, formally proved in \Cref{lem:cores of elts of intersecting atoms are orthogonal}. To see why, consider distinct intersecting planes $P$ and $Q$ with $R=P\cap Q$. For elements in their atoms $a_P\in A_P$ and $a_Q\in A_Q$, we can compute their inner product by using the adjoint of the corestriction map as
\begin{align*}
    \langle \cores_{P\to[n]^D}(a_P),\cores_{Q\to [n]^D}(a_Q)\rangle=\langle\cores_{R\to P}^\adjoint(a_P),\cores_{R\to Q}^\adjoint(a_Q)\rangle.
\end{align*}

Say, without loss of generality, that $R\subsetneq P$. Then, by the definition of atoms, $\cores_{R\to P}^\adjoint(a_P)=0$ making the entire inner product zero. In the case that $P$ and $Q$ are not distinct, so $P=Q$, it is not too hard to see that we get an isometry $\langle \cores_{P\to[n]^D}(a_P),\cores_{P\to [n]^D}(a_P')\rangle=\langle a_P,a_P'\rangle$, as shown in \Cref{cor:cores is isometry on atoms}.

\paragraph{The plane-interaction matrix}

We arrive at our main analytic tool, the plane-interaction matrix, by pairwise comparing the corestriction maps of every plane in $\maxplanesclosed(M)$. In particular, we define $\Theta:\bigoplus_{P\in\maxplanesclosed(M)} A_P\to U_{[D]}$ by $\Theta((a_P)_P)=\sum_P \cores_{P\to[n]^D}(a_P)$. Using our two observations about the inner product of intersecting planes above, we can decompose the Gram matrix of $\Theta$ as $\Theta^\adjoint\Theta=I+K$ where the block matrix $K$ is defined in its $(P,Q)$-th block by

\begin{wrapfigure}[22]{r}{0.50\textwidth}
    \vspace{-0.7\baselineskip}
    \centering
    \begingroup
\providecommand{\PlaneFigureUnit}{0.72cm}
\providecommand{\PlaneMatrixSymbol}{K}
\providecommand{\DirectionBlockLabel}{K_{I,J}}
\providecommand{\PlaneBlockLabel}{K_{P,Q}}
\providecommand{\PlanePairLabelP}{P=\mathcal{P}(I,a)}
\providecommand{\PlanePairLabelQ}{Q=\mathcal{P}(J,b)}

\begin{tikzpicture}[
    x=\PlaneFigureUnit,y=\PlaneFigureUnit,
    font=\footnotesize,
    outer matrix/.style={draw=black,line width=0.85pt},
    dir boundary/.style={draw=black,dashed,line width=0.85pt},
    plane boundary/.style={draw=black,dotted,line width=0.75pt},
    dir fill/.style={fill=black!8,draw=none},
    plane fill/.style={fill=black!16,draw=none}
]

\coordinate (SW) at (0,0);
\coordinate (NE) at (8.6,8.4);
\draw[outer matrix] (SW) rectangle (NE);

\def\xL{4.75}
\def\xR{7.95}
\def\yB{3.10}
\def\yT{5.95}

\def\xpL{5.35}
\def\xpR{6.55}
\def\ypB{4.05}
\def\ypT{5.15}

\fill[dir fill] (\xL,\yB) rectangle (\xR,\yT);
\fill[plane fill] (\xpL,\ypB) rectangle (\xpR,\ypT);

\draw[dir boundary] (\xL,0) -- (\xL,8.4);
\draw[dir boundary] (\xR,0) -- (\xR,8.4);
\draw[dir boundary] (0,\yB) -- (8.6,\yB);
\draw[dir boundary] (0,\yT) -- (8.6,\yT);

\draw[plane boundary] (\xpL,\yB) -- (\xpL,\yT);
\draw[plane boundary] (\xpR,\yB) -- (\xpR,\yT);
\draw[plane boundary] (\xL,\ypB) -- (\xR,\ypB);
\draw[plane boundary] (\xL,\ypT) -- (\xR,\ypT);

\node[font=\small] at ({(\xpR+\xR)/2},{(\yB+\ypB)/2}) {$\DirectionBlockLabel$};
\node[font=\small] at ({(\xpL+\xpR)/2},{(\ypB+\ypT)/2}) {$\PlaneBlockLabel$};

\node[font=\Huge,anchor=south east] at (-0.14,8.48) {$\PlaneMatrixSymbol$};

\draw[decorate,decoration={brace,amplitude=4.5pt}] (-1.08,\yB) -- (-1.08,\yT)
    node[midway,left=7pt,font=\normalsize] {$I$};
\draw[decorate,decoration={brace,amplitude=3.8pt}] (-0.20,\ypB) -- (-0.20,\ypT)
    node[midway,left=5pt,font=\normalsize] {$a$};

\draw[decorate,decoration={brace,amplitude=4.5pt}] (\xL,9.42) -- (\xR,9.42)
    node[midway,above=7pt,font=\normalsize] {$J$};
\draw[decorate,decoration={brace,amplitude=3.8pt}] (\xpL,8.50) -- (\xpR,8.50)
    node[midway,above=5pt,font=\normalsize] {$b$};

\node[font=\small,anchor=east] at (4.62,{(\ypB+\ypT)/2}) {$\PlanePairLabelP$};
\node[font=\small,anchor=south] at ({(\xpL+\xpR)/2},6.08) {$\PlanePairLabelQ$};

\end{tikzpicture}
\endgroup
    \caption{the plane-interaction matrix. The lighter gray box represents the directional block $K_{I,J}$, while the darker gray box is a plane block $K_{P,Q}$ with planes $P=\plane(I,a)$ and $Q=\plane(J,b)$.}
    \label{fig:plane-comparison}
    \vspace{-0.8\baselineskip}
\end{wrapfigure}
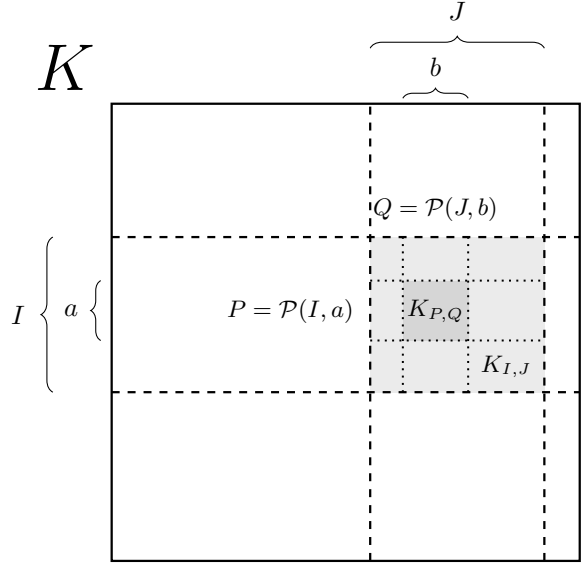

\leavevmode
\begin{align}
    K_{P,Q}=
    \begin{cases}
        \cores_{P\to[n]^D}^\adjoint\cores_{Q\to[n]^D} & P\cap Q=\varnothing\\
        0 & P\cap Q\neq\varnothing.
    \end{cases}\label{eq:overview-plane interaction}
\end{align}

This matrix $K$ is exactly our plane-interaction matrix and visualized in \Cref{fig:plane-comparison}.

Technically, the domain of $K$ is $\bigoplus_P A_P$, but expanding it to $\bigoplus_P U_P$ is convenient for later calculations and can only increase its operator norm, which we are trying to upper bound. From this construction, we get
\begin{align*}
    \norm{\Theta((a_P)_P)}^2\geq(1-\opnorm{K})\sum_P \norm{a_P}^2,
\end{align*}
meaning that $\Theta^\adjoint\Theta$ is positive definite, and thus $\Theta$ injective, if $\opnorm{K}<1$.  Finally, to conclude that $\cores_M$ is injective, meaning that $M$ is inner-generated, we recall that the direct sum of all atoms spans $\cosheaf(M)$. With the observation that, after identifying $u\in U_P$ with $\cores_{P\to Q}(u)\in U_Q$ for $P\subseteq Q$, the resulting $\Theta$ on the quotient is equivalent to $\cores_M$, we have that $\cores_M$ is injective if $\Theta$ is.

Consequently, our goal in the next couple steps is to bound the expected operator norm of $K$ where our randomness comes from the Rademacher generator matrices $G_1,\dots,G_D$. We will eventually conclude that $\E[\opnorm{K}]\leq O_{D,\rho}(1)\cdot\sqrt{2^D\varepsilon}$, so taking $\varepsilon$ small enough with respect to the constants $D$ and $\rho$ can make this expectation less than $1/2$. Markov's inequality then gives us that there must exist some real setting of $G_1,\dots,G_D$ such that $M$ is inner-generated for the resulting collection of codes.

\subsection{Bounding the norm of the plane-interaction matrix by the expected value of Rademacher chaos polynomials}

From the last step, we have reduced \Cref{lem:overview-existence of witness} to bounding $\E[\opnorm{K}]$. A tempting first attempt at bounding $\opnorm{K}$ is to use the convenient fact that it is defined via the blocks $K_{P,Q}$, so one could plausibly bound $\opnorm{K_{P,Q}}$ for all pairs $P,Q\in\maxplanesclosed(M)$ and then use the triangle inequality to bound $\opnorm{K}$ by their sum. However, we would only be able to achieve a constant bound $\opnorm{K_{P,Q}}=O_{D,\rho,\varepsilon}(1)$ for each block while it is possible to have $\abs{\maxplanesclosed(M)}=O_{D,\rho,\varepsilon}(n^D)$, meaning the triangle inequality would fail to provide the necessary bound.

\paragraph{Directional block decomposition}
To overcome this barrier, for two directions $I,J\subseteq[D]$, we chunk together all blocks $K_{P,Q}$ where $P$ is an $I$-plane and $Q$ is a $J$-plane to form the directional block $K_{I,J}$. There are at most $4^D$ many of these blocks, so bounding $\opnorm{K_{I,J}}$ by a small enough constant and then using the triangle inequality now succeeds.  

Focusing on a single directional block $K_{I,J}$, in \Cref{lem:hatKIJ decomposition} we again decompose it as the sum of matrices as $K_{I,J}=\sum_{\varnothing\neq E\subseteq F} K_{I,J}^E$ where $F=\setcomp{(I\cup J)}$, the coordinates that are fixed by an anchor in both $I$-planes and $J$-planes. Notice that there are at most $2^{\abs{F}}\leq 2^D$ terms in this sum, so bounding each term by a small enough constant and applying the triangle inequality is again a viable strategy. Using the product structure of the tensor code and some elementary, but tedious, algebraic manipulations that we omit here, in \Cref{lem:UB on op norm of KIJE} we bound the operator norm of $K_{I,J}^E$ by
\begin{align*}
    \opnorm{K_{I,J}^E}\leq\opnorm{\bigotimes_{i\in I\setminus J} G_i}\opnorm{\bigotimes_{j\in J\setminus I} G_j}\max_{\lambda\in\Lambda}\opnorm{\Pi_{X_\lambda} \left(\bigotimes_{j\in E} D_j\right)\iota_{Y_\lambda}},
\end{align*}
where $D_j=G_j^\transpose G_j-\Id_{\R^n}$, the indexing set $\Lambda$ has cardinality at most $n^\nu$ where $\nu+\abs{E}\leq D$, the sets $X_\lambda,Y_\lambda$ are subsets of $[n]^E$, and $\Pi_{X_\lambda}$ and $\iota_{Y_\lambda}$ are projection and inclusion (by 0's). 

The first two operator norms are easy to bound by $D$ and $\rho$ using standard techniques from Gaussian matrices. The third factor is where we focus our attention on for the rest of this step and the next. The two sets $X_\lambda,Y_\lambda$ are derived by taking an $E$-plane, indexed by $\lambda$, from $\maxAIclosed(M)$ and $\maxAclosed{J}(M)$, respectively. Consequently, they are both $(2^D\varepsilon)$-sparse in $[n]^E$. Taking $T_\lambda=\Pi_{X_\lambda} \left(\bigotimes_{j\in E} D_j\right)\iota_{Y_\lambda}$, we will show, as stated in \Cref{prop:UB on op norm of max T_lambda}, that
\begin{align}
    \E\left[\max_{\lambda\in\Lambda}\opnorm{T_\lambda}\right]\leq 2\left(\frac{4}{\rho}\right)^{\abs{E}}\sqrt{2^D\varepsilon},\label{eq:overview-bound on exp max opnorm T}
\end{align}
which we note can be made a small constant by taking $\varepsilon$ small enough. 

\paragraph{Deriving the Rademacher chaos polynomials} 
To focus on a single $\lambda\in\Lambda$, we fix a pair of $\eta$-sparse subsets $X,Y\subseteq[n]^q$, where $\eta=2^D\varepsilon$, and let $T=\Pi_X\left(\bigotimes_{j=1}^q D_j\right)\iota_Y$. We apply the trace-power method for an integer $m\in\N$, that we will set later, to get $\opnorm{T}^{2m}\leq\trace((TT^\transpose)^m)$. Expanding out this trace yields the expression
\begin{align*}
    \trace((TT^\transpose)^m)
    =\sum_{\substack{x_1,\ldots,x_m\in X\\y_1,\ldots,y_m\in Y}}
       \prod_{j=1}^q\prod_{t=1}^m
           D_j(x_{t,j},y_{t,j})D_j(y_{t,j},x_{t+1,j}),
    \qquad x_{m+1}=x_1.
\end{align*}
Recalling the definition of $D_j$, we can write a single entry of $D_j$ as $D_j(a,b)=\Ind\{a\neq b\}k_j^{-1}\sum_{r=1}^{k_j}\sigma_{r,a}^{(j)}\sigma_{r,b}^{(j)}$. Plugging this in, as we do in \Cref{cor:tr(TT^transpose) as sum of prod of monomial signs} gives us
\begin{align}
    \trace\left((TT^\transpose)^p\right)=\sum_{\substack{x_1,\dots,x_m\in X\\y_1,\dots,y_m\in Y}}\sum_{\substack{r_j\in[k_j]^{2m}\\j\in[q]}}\left(\prod_{j=1}^q k_j^{-2m}\right)\left(\prod_{j=1}^q\Phi_j(x,y,r_j)\right),\label{eq:overview-trace power monomial expansion}
\end{align}
where $x=(x_1,\dots,x_m)$, $y=(y_1,\dots,y_m)$, and we define the monomials of signs
\begin{align}
    \Phi_j(x,y,r_j)=\prod_{t=1}^m\left(\sigma^{(j)}_{r_{j,2t-1},x_{t,j}}\cdot\sigma^{(j)}_{r_{j,2t-1},y_{t,j}}\cdot\sigma^{(j)}_{r_{j,2t},y_{t,j}}\cdot\sigma^{(j)}_{r_{j,2t},x_{t+1,j}}\cdot\Ind\{x_{t,j}\neq y_{t,j}\text{ and } y_{t,j}\neq x_{t+1,j}\}\right).\label{eq:overview-sign monomial}
\end{align}

This expression for the trace is a sum of homogeneous multivariate polynomials with the domain of each variable being $\{-1,1\}$, alternatively known as a Rademacher chaos.\footnote{The only slight deviation from the formal definition of a Rademacher chaos that we have here is the indicator functions in each monomial, but we do not use any off-the-shelf results on Rademacher chaoses in which this would be an issue; in fact, we significantly leverage them later on.} Now, we finally leverage the unique nature of Rademacher random variables to easily compute the expectation $\E_\sigma[\Phi_j(x,y,r_j)]$ as 1 whenever each sign appears an even number of times\footnote{While satisfying each indicator function $\Ind\{x_{t,j}\neq y_{t,j}\text{ and } y_{t,j}\neq x_{t+1,j}\}$, which we will see is equivalent to the corresponding bipartite walk being non-left-backtracking.} or 0 otherwise, since just one random sign is enough to randomize the entire monomial to zero mean. 

Recall that our goal is to upper bound the expectation of \Cref{eq:overview-trace power monomial expansion}, so we would like to bound the number of $\E_\sigma[\Phi_j(x,y,r_j)]$ that are deterministically 1, which we do in the next step.

\subsection{Bounding the expected value of Rademacher chaos polynomials using weak Wigner words}

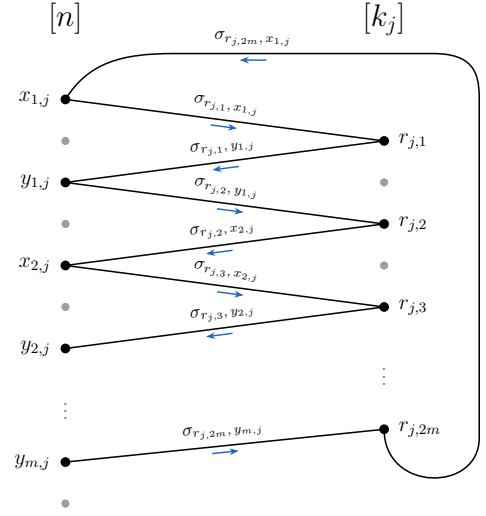
\begin{wrapfigure}[22]{r}{0pt}
    \centering
    
    \scalebox{.68}{\begingroup
\providecommand{\MonomialWalkSigma}{\sigma}
\definecolor{walkBlue}{RGB}{32,102,195}
\newcommand{\MonomialEdgeSign}[2]{%
    \MonomialWalkSigma_{r_{j,#1},\,#2}%
}
\newcommand{\MonomialDirArrow}[5]{%
    \coordinate (monomialArrowPos) at ($(#1)!#3!(#2)$);
    \pgfmathanglebetweenpoints{\pgfpointanchor{#1}{center}}%
        {\pgfpointanchor{#2}{center}}%
    \let\monomialAngle\pgfmathresult
    \begin{scope}[shift={(monomialArrowPos)},rotate=\monomialAngle]
        \draw[monomial traversal]
            ({-0.26*(#5)},{0.15*(#4)}) -- ({0.26*(#5)},{0.15*(#4)});
    \end{scope}
}

\begin{tikzpicture}[
    x=1cm,y=0.85cm,font=\normalsize,
    monomial edge/.style={draw=black,line width=0.85pt,line cap=round},
    monomial vertex/.style={circle,fill=black,draw=none,
        inner sep=0pt,minimum size=5.4pt,outer sep=0pt},
    monomial unused/.style={circle,fill=black!38,draw=none,
        inner sep=0pt,minimum size=4.5pt,outer sep=0pt},
    monomial left label/.style={font=\large,anchor=east,inner sep=1pt},
    monomial right label/.style={font=\large,anchor=west,inner sep=1pt},
    monomial sign label/.style={font=\normalsize,inner sep=1pt},
    monomial traversal/.style={draw=walkBlue,line width=0.8pt,
        -{Stealth[length=1.8mm,width=1.3mm]}},
    monomial header/.style={font=\LARGE,inner sep=1pt},
    monomial omitted/.style={font=\large,text=black!55,inner sep=0pt}
]

\node[monomial header] at (0,8.65) {$[n]$};
\node[monomial header] at (6.20,8.65) {$[k_j]$};

\coordinate (msw-x1) at (0, 6.80);
\coordinate (msw-y1) at (0, 4.90);
\coordinate (msw-x2) at (0, 3.00);
\coordinate (msw-y2) at (0, 1.10);
\coordinate (msw-ym) at (0,-1.50);

\coordinate (msw-r1) at (6.20, 5.85);
\coordinate (msw-r2) at (6.20, 3.95);
\coordinate (msw-r3) at (6.20, 2.05);
\coordinate (msw-rm) at (6.20,-0.75);

\coordinate (msw-close-top) at (2.10,7.85);

\draw[monomial edge]
    (msw-x1)
    .. controls (0.45,7.65) and (1.10,7.85) .. (2.10,7.85)
    -- (7.15,7.85)
    .. controls (7.90,7.85) and (8.05,7.45) .. (8.05,6.85)
    -- (8.05,-0.95)
    .. controls (8.05,-2.20) and (6.20,-2.20) .. (msw-rm);
\node[monomial sign label,anchor=south] at (3.65,7.99)
    {$\MonomialEdgeSign{2m}{x_{1,j}}$};
\draw[monomial traversal] (3.91,7.70) -- (3.39,7.70);

\draw[monomial edge] (msw-x1) -- (msw-r1)
    node[monomial sign label,pos=0.50,sloped,above=3pt]
    {$\MonomialEdgeSign{1}{x_{1,j}}$};
\MonomialDirArrow{msw-x1}{msw-r1}{0.50}{-1}{1}
\draw[monomial edge] (msw-y1) -- (msw-r1)
    node[monomial sign label,pos=0.50,sloped,above=3pt]
    {$\MonomialEdgeSign{1}{y_{1,j}}$};
\MonomialDirArrow{msw-y1}{msw-r1}{0.50}{-1}{-1}
\draw[monomial edge] (msw-y1) -- (msw-r2)
    node[monomial sign label,pos=0.50,sloped,above=3pt]
    {$\MonomialEdgeSign{2}{y_{1,j}}$};
\MonomialDirArrow{msw-y1}{msw-r2}{0.52}{-1}{1}
\draw[monomial edge] (msw-x2) -- (msw-r2)
    node[monomial sign label,pos=0.50,sloped,above=3pt]
    {$\MonomialEdgeSign{2}{x_{2,j}}$};
\MonomialDirArrow{msw-x2}{msw-r2}{0.48}{-1}{-1}
\draw[monomial edge] (msw-x2) -- (msw-r3)
    node[monomial sign label,pos=0.50,sloped,above=3pt]
    {$\MonomialEdgeSign{3}{x_{2,j}}$};
\MonomialDirArrow{msw-x2}{msw-r3}{0.52}{-1}{1}
\draw[monomial edge] (msw-y2) -- (msw-r3)
    node[monomial sign label,pos=0.50,sloped,above=3pt]
    {$\MonomialEdgeSign{3}{y_{2,j}}$};
\MonomialDirArrow{msw-y2}{msw-r3}{0.48}{-1}{-1}
\draw[monomial edge] (msw-ym) -- (msw-rm)
    node[monomial sign label,pos=0.50,sloped,above=3pt]
    {$\MonomialEdgeSign{2m}{y_{m,j}}$};
\MonomialDirArrow{msw-ym}{msw-rm}{0.50}{-1}{1}

\foreach \yy in {5.85,3.95,2.05,-2.45}{%
    \node[monomial unused] at (0,\yy) {};
}
\foreach \yy in {4.90,3.00}{%
    \node[monomial unused] at (6.20,\yy) {};
}

\node[monomial omitted] at (0,-0.20) {$\vdots$};
\node[monomial omitted] at (6.20,0.58) {$\vdots$};

\foreach \v in {x1,y1,x2,y2,ym,r1,r2,r3,rm}{%
    \node[monomial vertex] at (msw-\v) {};
}

\node[monomial left label,left=7pt] at (msw-x1) {$x_{1,j}$};
\node[monomial left label,left=7pt] at (msw-y1) {$y_{1,j}$};
\node[monomial left label,left=7pt] at (msw-x2) {$x_{2,j}$};
\node[monomial left label,left=7pt] at (msw-y2) {$y_{2,j}$};
\node[monomial left label,left=7pt] at (msw-ym) {$y_{m,j}$};
\node[monomial right label,right=7pt] at (msw-r1) {$r_{j,1}$};
\node[monomial right label,right=7pt] at (msw-r2) {$r_{j,2}$};
\node[monomial right label,right=7pt] at (msw-r3) {$r_{j,3}$};
\node[monomial right label,right=7pt] at (msw-rm) {$r_{j,2m}$};

\end{tikzpicture}
\endgroup}
    \caption{The closed walk on the bipartite graph for the sign monomial $\Phi_j(x,y,r_j)$. We omit the superscript $(j)$ for legibility.}
    \label{fig:monomial-walk}
\end{wrapfigure}
\paragraph{Sign monomials as bipartite walks}
As we've already observed, our sign monomials $\Phi_j(x,y,r_j)$ have mean 1 exactly when each sign appears an even number of times in the product of \Cref{eq:overview-sign monomial}. Thus, all we must do is count how many choices of $x\in X^m$, $y\in Y^m$, and $r_j\in[k_j]^{2m}$ for $j\in[q]$ there are that yield a sign monomial with each sign appearing an even number of times. A useful way to visualize this is by writing out the relevant indices in the order they are read for a fixed monomial $\Phi_j(x,y,r_j)$ as
\begin{align}
    x_{1,j},\ r_{j,1},\ y_{1,j},\ r_{j,2},\ x_{2,j},\ldots,
       y_{m,j},\ r_{j,2m},\ x_{1,j}.\label{eq:monomial sign walk}
\end{align}
Each sequential pair in this list defines one of the $4m$ signs in \Cref{eq:overview-sign monomial}. Equivalently, we could view each $x_{a,j}$, $y_{a,j}$, and $r_{j,a}$ index as a vertex and this list as giving us a walk on the resulting graph, as seen in \Cref{fig:monomial-walk}. Since the list from \Cref{eq:monomial sign walk} always alternates between $x_{a,j}$ or $y_{a,j}$ and a $r_{j,a}$, it is in fact a walk on a bipartite graph: the bipartite graph with $[n]$ left vertices and $[k_j]$ right vertices where each edge represents a single sign. 

We now make three observations regarding this view of the signs in $\Phi_j(x,y,r_j)$ as a walk on a bipartite graph. First, it is a closed walk by definition since $x_{1,j}=x_{m+1,j}$. Second, it is non-left backtracking, meaning that each step along the walk must have different left neighbors when entering and exiting a right vertex. This is due to the indicator variables in \Cref{eq:overview-sign monomial}. Third, each sign appears an even number of times if each edge used in the walk is traversed an even number of times. Counting the number of such closed walks on this bipartite graph thus gives us an upper bound on the number of sign monomials with expectation 1.

There are two pieces of information that define such a bipartite walk: its shape and the labels on its vertices. For example, the shape of the bipartite walk, if it had length 8, could be tracing the bipartite clique $K_{2,2}$ twice. The labeling of the walk determines where exactly we place this clique on the left $[n]$ and right $[k_j]$ vertices. 

Our strategy for bounding the number of non-left-backtracking closed walks that traverse each edge an even number of times then is to use the series of works \cite{wigner_characteristic_1955,furedi_eigenvalues_1981,anderson_clt_2006}  to bound the number of shapes such a walk could take and then to use the $\eta$-sparsity of $X$ and $Y$ to bound the number of labelings.

\paragraph{Bounding the number of shapes with weak Wigner words} 
The walks on bipartite graphs that we have described here, which can alternatively be described by a string as in \Cref{eq:monomial sign walk}, are a special case of weak Wigner words, which we treat formally in \Cref{sec:bounding max opnorm with wigner}. Anderson and Zeitouni in \cite{anderson_clt_2006} give a combinatorial bound on the number of shapes such a walk can take. More concretely, they bound the number of equivalence classes of these bipartite walks, where two walks are equivalent if there is a bijection between the strings representing the walks.

To apply their bound, we define $\wt_\sL(j)$ and $\wt_\sR(j)$ to be the number of distinct left and right vertices in the walk of $\Phi_j(x,y,r_j)$, and $\wt(j)=\wt_\sL(j)+\wt_\sR(j)$. Since there are $4m$ edge traversals, and each edge is used at least twice, we have that $\wt(j)\leq 2m+1$. We define the \emph{deficiency} of a walk to be $\Delta_j=2m+1-\wt(j)$, which represents how far the walk is from using as many vertices as it possibly could. If we let $\fW_m(\Delta_j)$ be the set of equivalence classes of even, non-left-backtracking, closed bipartite walks that traverse $4m$ edges with deficiency $\Delta_j$, then the bound from \cite{anderson_clt_2006}, that we state in a form useful to us in \Cref{cor: prop 7.9 for bipartite equiv}, says that $\abs{\fW_m(\Delta_j)}\leq16^m(4m)^{6\Delta_j}$.  For example, in \Cref{fig:fW_2(1) example} we explicitly write out all equivalence classes in $\fW_2(1)$.

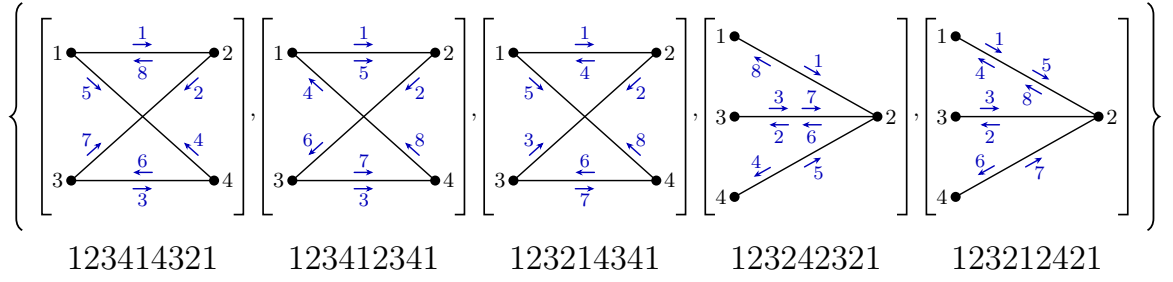
\begin{figure}[ht]
    \centering
    \scalebox{.9}{\begingroup
\providecommand{\fW}{\mathfrak{W}}
\providecommand{\WalkFamilyLabel}{\fW_2(1)}
\providecommand{\WalkFamilyAccent}{blue!73!black}

\newcommand{\WalkFamilyStep}[6]{%
    \pgfmathanglebetweenpoints{\pgfpointanchor{#1}{center}}%
                              {\pgfpointanchor{#2}{center}}%
    \let\WalkFamilyAngle\pgfmathresult
    \pgfmathsetmacro{\WalkFamilyLabelOffset}{#4 + (#4 > 0 ? 0.17 : -0.17)}%
    \begin{scope}[shift={($(#1)!#3!(#2)$)},rotate=\WalkFamilyAngle]
        \draw[wf traversal]
            ({-0.14*(#5)},#4) -- ({0.14*(#5)},#4);
        \node[wf step label] at (0,\WalkFamilyLabelOffset) {$#6$};
    \end{scope}%
}

\newcommand{\WalkFamilyEdges}{%
    \coordinate (wfUL) at (0, 0.94);
    \coordinate (wfUR) at (2.10, 0.94);
    \coordinate (wfLL) at (0,-0.94);
    \coordinate (wfLR) at (2.10,-0.94);
    \draw[wf edge] (wfUL)--(wfUR);
    \draw[wf edge] (wfUL)--(wfLR);
    \draw[wf edge] (wfLL)--(wfUR);
    \draw[wf edge] (wfLL)--(wfLR);
}

\newcommand{\WalkFamilyBrackets}{%
    \draw[wf bracket] (-0.30,1.44)--(-0.43,1.44)
        --(-0.43,-1.44)--(-0.30,-1.44);
    \draw[wf bracket] (2.40,1.44)--(2.53,1.44)
        --(2.53,-1.44)--(2.40,-1.44);
}

\newcommand{\WalkFamilyVerticesAndBrackets}{%
    \foreach \v in {wfUL,wfUR,wfLL,wfLR}
        \node[wf vertex] at (\v) {};
    \node[wf vertex label,anchor=east] at (-0.10, 0.94) {$1$};
    \node[wf vertex label,anchor=west] at ( 2.20, 0.94) {$2$};
    \node[wf vertex label,anchor=east] at (-0.10,-0.94) {$3$};
    \node[wf vertex label,anchor=west] at ( 2.20,-0.94) {$4$};
    \WalkFamilyBrackets
}

\newcommand{\WalkFamilyStarEdges}{%
    \coordinate (wfLone)   at (0, 1.18);
    \coordinate (wfLthree) at (0, 0);
    \coordinate (wfLfour)  at (0,-1.18);
    \coordinate (wfRtwo)   at (2.10,0);
    \draw[wf edge] (wfLone)--(wfRtwo);
    \draw[wf edge] (wfLthree)--(wfRtwo);
    \draw[wf edge] (wfLfour)--(wfRtwo);
}

\newcommand{\WalkFamilyStarVerticesAndBrackets}{%
    \foreach \v in {wfLone,wfLthree,wfLfour,wfRtwo}
        \node[wf vertex] at (\v) {};
    \node[wf vertex label,anchor=east] at (-0.10, 1.18) {$1$};
    \node[wf vertex label,anchor=east] at (-0.10, 0) {$3$};
    \node[wf vertex label,anchor=east] at (-0.10,-1.18) {$4$};
    \node[wf vertex label,anchor=west] at ( 2.20, 0) {$2$};
    \WalkFamilyBrackets
}

\begin{tikzpicture}[
    x=1cm,y=1cm,font=\normalsize,
    wf edge/.style={draw=black,line width=0.6pt,line cap=round},
    wf vertex/.style={circle,draw=none,fill=black,
        inner sep=0pt,outer sep=0pt,minimum size=4.2pt},
    wf vertex label/.style={font=\fontsize{9}{11}\selectfont,
        inner sep=0.5pt,outer sep=0pt},
    wf traversal/.style={draw=\WalkFamilyAccent,line width=0.6pt,
        -{Stealth[length=1.10mm,width=0.90mm]},line cap=round},
    wf step label/.style={text=\WalkFamilyAccent,font=\fontsize{8}{10}\selectfont,
        inner sep=0pt,outer sep=0pt},
    wf bracket/.style={draw=black,line width=0.7pt,line cap=rect},
    wf set brace/.style={draw=black,line width=0.75pt,
        decorate,decoration={brace,amplitude=5pt}},
    wf punctuation/.style={font=\normalsize,inner sep=0pt,outer sep=0pt},
    wf word label/.style={font=\fontsize{14}{16}\selectfont\ttfamily,inner sep=0pt,outer sep=0pt}
]

\path[use as bounding box] (-0.72,-2.55) rectangle (15.82,1.67);
\draw[wf set brace] (-0.72,-1.67)--(-0.72,1.67);
\draw[wf set brace] (15.82,1.67)--(15.82,-1.67);

\begin{scope}[xshift=0.00cm]
    \WalkFamilyEdges
    \WalkFamilyStep{wfUL}{wfUR}{0.50}{ 0.12}{ 1}{1}
    \WalkFamilyStep{wfUL}{wfUR}{0.50}{-0.12}{-1}{8}
    \WalkFamilyStep{wfUL}{wfLR}{0.20}{-0.12}{ 1}{5}
    \WalkFamilyStep{wfUL}{wfLR}{0.80}{ 0.12}{-1}{4}
    \WalkFamilyStep{wfLL}{wfUR}{0.20}{ 0.12}{ 1}{7}
    \WalkFamilyStep{wfLL}{wfUR}{0.80}{-0.12}{-1}{2}
    \WalkFamilyStep{wfLL}{wfLR}{0.50}{-0.12}{ 1}{3}
    \WalkFamilyStep{wfLL}{wfLR}{0.50}{ 0.12}{-1}{6}
    \WalkFamilyVerticesAndBrackets
\end{scope}
\node[wf word label] at (1.05,-2.08) {$123414321$};

\node[wf punctuation] at (2.675,0) {$,$};

\begin{scope}[xshift=3.25cm]
    \WalkFamilyEdges
    \WalkFamilyStep{wfUL}{wfUR}{0.50}{ 0.12}{ 1}{1}
    \WalkFamilyStep{wfUL}{wfUR}{0.50}{-0.12}{ 1}{5}
    \WalkFamilyStep{wfUL}{wfLR}{0.20}{-0.12}{-1}{4}
    \WalkFamilyStep{wfUL}{wfLR}{0.80}{ 0.12}{-1}{8}
    \WalkFamilyStep{wfLL}{wfUR}{0.20}{ 0.12}{-1}{6}
    \WalkFamilyStep{wfLL}{wfUR}{0.80}{-0.12}{-1}{2}
    \WalkFamilyStep{wfLL}{wfLR}{0.50}{-0.12}{ 1}{3}
    \WalkFamilyStep{wfLL}{wfLR}{0.50}{ 0.12}{ 1}{7}
    \WalkFamilyVerticesAndBrackets
\end{scope}
\node[wf word label] at (4.30,-2.08) {$123412341$};

\node[wf punctuation] at (5.925,0) {$,$};

\begin{scope}[xshift=6.50cm]
    \WalkFamilyEdges
    \WalkFamilyStep{wfUL}{wfUR}{0.50}{ 0.12}{ 1}{1}
    \WalkFamilyStep{wfUL}{wfUR}{0.50}{-0.12}{-1}{4}
    \WalkFamilyStep{wfUL}{wfLR}{0.20}{-0.12}{ 1}{5}
    \WalkFamilyStep{wfUL}{wfLR}{0.80}{ 0.12}{-1}{8}
    \WalkFamilyStep{wfLL}{wfUR}{0.20}{ 0.12}{ 1}{3}
    \WalkFamilyStep{wfLL}{wfUR}{0.80}{-0.12}{-1}{2}
    \WalkFamilyStep{wfLL}{wfLR}{0.50}{-0.12}{ 1}{7}
    \WalkFamilyStep{wfLL}{wfLR}{0.50}{ 0.12}{-1}{6}
    \WalkFamilyVerticesAndBrackets
\end{scope}
\node[wf word label] at (7.55,-2.08) {$123214341$};

\node[wf punctuation] at (9.175,0) {$,$};

\begin{scope}[xshift=9.75cm]
    \WalkFamilyStarEdges
    \WalkFamilyStep{wfLone}{wfRtwo}{0.52}{ 0.12}{ 1}{1}
    \WalkFamilyStep{wfLone}{wfRtwo}{0.22}{-0.12}{-1}{8}
    \WalkFamilyStep{wfLthree}{wfRtwo}{0.31}{-0.12}{-1}{2}
    \WalkFamilyStep{wfLthree}{wfRtwo}{0.31}{ 0.12}{ 1}{3}
    \WalkFamilyStep{wfLthree}{wfRtwo}{0.54}{-0.12}{-1}{6}
    \WalkFamilyStep{wfLthree}{wfRtwo}{0.54}{ 0.12}{ 1}{7}
    \WalkFamilyStep{wfLfour}{wfRtwo}{0.22}{ 0.12}{-1}{4}
    \WalkFamilyStep{wfLfour}{wfRtwo}{0.52}{-0.12}{ 1}{5}
    \WalkFamilyStarVerticesAndBrackets
\end{scope}
\node[wf word label] at (10.80,-2.08) {$123242321$};

\node[wf punctuation] at (12.425,0) {$,$};

\begin{scope}[xshift=13.00cm]
    \WalkFamilyStarEdges
    \WalkFamilyStep{wfLone}{wfRtwo}{0.24}{ 0.12}{ 1}{1}
    \WalkFamilyStep{wfLone}{wfRtwo}{0.24}{-0.12}{-1}{4}
    \WalkFamilyStep{wfLone}{wfRtwo}{0.57}{ 0.12}{ 1}{5}
    \WalkFamilyStep{wfLone}{wfRtwo}{0.57}{-0.12}{-1}{8}
    \WalkFamilyStep{wfLthree}{wfRtwo}{0.24}{-0.12}{-1}{2}
    \WalkFamilyStep{wfLthree}{wfRtwo}{0.24}{ 0.12}{ 1}{3}
    \WalkFamilyStep{wfLfour}{wfRtwo}{0.24}{ 0.12}{-1}{6}
    \WalkFamilyStep{wfLfour}{wfRtwo}{0.52}{-0.12}{ 1}{7}
    \WalkFamilyStarVerticesAndBrackets
\end{scope}
\node[wf word label] at (14.05,-2.08) {$123212421$};

\end{tikzpicture}%
\endgroup}
    \caption{An example of the set of equivalence classes $\fW_{m=2}(\Delta=1)$ with 8 edge traversals (i.e., word length 9) and 4 distinct vertices, giving deficiency $\Delta=2m+1-\wt(w)=4+1-4=1$. All walks start from node 1 and the numbered blue arrows dictate the direction that the walk takes. The label beneath each walk is the string representing the walk on the vertices, formally the canonical representative of the equivalence class that we define in \Cref{def:canonical representative}. These are different equivalence classes because there is no bijective relabeling map from $[4]\to[4]$ that can take one of these strings to another.}
    \label{fig:fW_2(1) example}
\end{figure}

All that remains now is to bound the number of ways we can label each shape using $X$ and $Y$. That is, the size of an individual equivalence class $\cW\in\fW_m(\Delta_j)$.

\begin{wrapfigure}[25]{r}{0pt}
    \centering
    \vspace{-1\baselineskip}
    \scalebox{0.75}{\begingroup
\providecommand{\cW}{\mathcal{W}}
\providecommand{\WalkEqAccent}{blue!73!black}

\newcommand{\WalkEqStep}[6]{%
    \pgfmathanglebetweenpoints{\pgfpointanchor{#1}{center}}%
                              {\pgfpointanchor{#2}{center}}%
    \let\WalkEqAngle\pgfmathresult
    \pgfmathsetmacro{\WalkEqLabelOffset}{#4 + (#4 > 0 ? 0.17 : -0.17)}%
    \begin{scope}[shift={($(#1)!#3!(#2)$)},rotate=\WalkEqAngle]
        \draw[walk traversal]
            ({-0.23*(#5)},#4) -- ({0.23*(#5)},#4);
        \node[walk step label] at (0,\WalkEqLabelOffset) {$#6$};
    \end{scope}%
}

\newcommand{\WalkEqSupport}[8]{%
    \draw[walk edge] (walkUL) -- (walkUR);
    \draw[walk edge] (walkUL) -- (walkLR);
    \draw[walk edge] (walkLL) -- (walkUR);
    \draw[walk edge] (walkLL) -- (walkLR);
    \WalkEqStep{walkUL}{walkUR}{0.50}{ 0.16}{ 1}{#1}
    \WalkEqStep{walkUL}{walkUR}{0.50}{-0.16}{-1}{#2}
    \WalkEqStep{walkUL}{walkLR}{0.20}{-0.15}{ 1}{#3}
    \WalkEqStep{walkUL}{walkLR}{0.80}{ 0.15}{-1}{#4}
    \WalkEqStep{walkLL}{walkUR}{0.20}{ 0.15}{ 1}{#5}
    \WalkEqStep{walkLL}{walkUR}{0.80}{-0.15}{-1}{#6}
    \WalkEqStep{walkLL}{walkLR}{0.50}{-0.16}{ 1}{#7}
    \WalkEqStep{walkLL}{walkLR}{0.50}{ 0.16}{-1}{#8}
}

\newcommand{\WalkEqStepShiftedLabel}[8]{%
    \pgfmathanglebetweenpoints{\pgfpointanchor{#1}{center}}%
                              {\pgfpointanchor{#2}{center}}%
    \let\WalkEqAngle\pgfmathresult
    \pgfmathsetmacro{\WalkEqLabelOffset}{#4 + (#4 > 0 ? 0.17 : -0.17)}%
    \begin{scope}[shift={($(#1)!#3!(#2)$)},rotate=\WalkEqAngle]
        \draw[walk traversal]
            ({-0.23*(#5)},#4) -- ({0.23*(#5)},#4);
        \coordinate (walkEqUnshiftedLabel) at (0,\WalkEqLabelOffset);
    \end{scope}%
    \node[walk step label] at ($(walkEqUnshiftedLabel)+(#7,#8)$) {$#6$};
}

\newcommand{\WalkEqSupportSecondTuned}[8]{%
    \draw[walk edge] (walkUL) -- (walkUR);
    \draw[walk edge] (walkUL) -- (walkLR);
    \draw[walk edge] (walkLL) -- (walkUR);
    \draw[walk edge] (walkLL) -- (walkLR);
    \WalkEqStep{walkUL}{walkUR}{0.50}{ 0.16}{ 1}{#1}
    \WalkEqStep{walkUL}{walkUR}{0.50}{-0.16}{-1}{#2}
    \WalkEqStep{walkUL}{walkLR}{0.20}{-0.15}{ 1}{#3}
    \WalkEqStepShiftedLabel{walkUL}{walkLR}{0.80}{ 0.15}{-1}{#4}{0.25}{-0.13}
    \WalkEqStep{walkLL}{walkUR}{0.20}{ 0.15}{ 1}{#5}
    \WalkEqStepShiftedLabel{walkLL}{walkUR}{0.80}{-0.15}{-1}{#6}{0.25}{ 0.13}
    \WalkEqStep{walkLL}{walkLR}{0.50}{-0.16}{ 1}{#7}
    \WalkEqStep{walkLL}{walkLR}{0.50}{ 0.16}{-1}{#8}
}

\begin{tikzpicture}[
    x=1cm,y=1cm,
    font=\normalsize,
    walk edge/.style={draw=black,line width=0.75pt,line cap=round},
    walk vertex/.style={circle,draw=none,fill=black,
        inner sep=0pt,outer sep=0pt,minimum size=5.7pt},
    walk unused vertex/.style={walk vertex,fill=black!40,minimum size=4.0pt},
    walk vertex label/.style={font=\large,inner sep=1pt,outer sep=0pt},
    walk unused label/.style={walk vertex label,text=black!55},
    walk set label/.style={font=\Large,inner sep=0pt},
    walk traversal/.style={draw=\WalkEqAccent,line width=0.8pt,
        -{Stealth[length=1.5mm,width=1.25mm]},line cap=round},
    walk step label/.style={text=\WalkEqAccent,font=\small,
        inner sep=0pt,outer sep=0pt},
    walk class bracket/.style={draw=black,line width=0.9pt,line cap=rect},
    walk branch/.style={draw=black!55,line width=1.0pt,
        -{Stealth[length=1.8mm,width=1.45mm]}}
]

\begin{scope}
    \coordinate (walkUL) at (-1.90,6.00);
    \coordinate (walkUR) at ( 2.80,6.00);
    \coordinate (walkLL) at (-1.90,4.12);
    \coordinate (walkLR) at ( 2.80,4.12);

    \WalkEqSupport{1}{8}{5}{4}{7}{2}{3}{6}
    \foreach \v in {walkUL,walkUR,walkLL,walkLR}
        \node[walk vertex] at (\v) {};
    \node[walk vertex label,anchor=east] at (-2.09,6.00) {$1$};
    \node[walk vertex label,anchor=west] at ( 2.99,6.00) {$2$};
    \node[walk vertex label,anchor=east] at (-2.09,4.12) {$3$};
    \node[walk vertex label,anchor=west] at ( 2.99,4.12) {$4$};

    \draw[walk class bracket] (-2.51,6.55) -- (-2.72,6.55)
        -- (-2.72,3.57) -- (-2.51,3.57);
    \draw[walk class bracket] ( 3.41,6.55) -- ( 3.62,6.55)
        -- ( 3.62,3.57) -- ( 3.41,3.57);
    \node[font=\LARGE,anchor=east,inner sep=0pt]
        at (-3.03,5.06) {$\cW = {}$};
\end{scope}

\draw[walk branch] (-1.58,3.38)
    .. controls (-2.22,3.09) and (-2.98,2.84) .. (-3.55,2.56);
\draw[walk branch] ( 2.48,3.38)
    .. controls ( 2.90,3.09) and ( 3.27,2.84) .. ( 3.75,2.56);

\begin{scope}[xshift=-3.65cm]
    \def\walkLeft{-2.13}
    \def\walkRight{2.13}
    \node[walk set label] at (\walkLeft,2.31) {$[n]$};
    \node[walk set label] at (\walkRight,2.31) {$[k_j]$};

    \coordinate (walkUL) at (\walkLeft, 0.98); %
    \coordinate (walkUR) at (\walkRight,0.98); %
    \coordinate (walkLL) at (\walkLeft,-0.66); %
    \coordinate (walkLR) at (\walkRight,-0.66);%
    \WalkEqSupport{1}{8}{5}{4}{7}{2}{3}{6}

    \foreach \lab/\yy in {1/1.80,3/0.16,5/-1.48,n/-3.02}{
        \node[walk unused vertex] at (\walkLeft,\yy) {};
        \node[walk unused label,anchor=east]
            at ({\walkLeft-0.20},\yy) {$\lab$};
    }
    \foreach \lab/\yy in {2/0.16,4/-1.48,k_j/-3.02}{
        \node[walk unused vertex] at (\walkRight,\yy) {};
        \node[walk unused label,anchor=west]
            at ({\walkRight+0.20},\yy) {$\lab$};
    }
    \foreach \lab/\yy in {2/0.98,4/-0.66}{
        \node[walk vertex] at (\walkLeft,\yy) {};
        \node[walk vertex label,anchor=east]
            at ({\walkLeft-0.20},\yy) {$\lab$};
    }
    \foreach \lab/\yy in {1/0.98,3/-0.66}{
        \node[walk vertex] at (\walkRight,\yy) {};
        \node[walk vertex label,anchor=west]
            at ({\walkRight+0.20},\yy) {$\lab$};
    }
    \node[text=black!55,inner sep=0pt] at (\walkLeft,-2.20) {$\vdots$};
    \node[text=black!55,inner sep=0pt] at (\walkRight,-2.20) {$\vdots$};
    \node[font=\large,inner sep=0pt] at (0,-3.58) {$w^{(1)}\in\cW$};
\end{scope}

\begin{scope}[xshift=3.65cm]
    \def\walkLeft{-2.13}
    \def\walkRight{2.13}
    \node[walk set label] at (\walkLeft,2.31) {$[n]$};
    \node[walk set label] at (\walkRight,2.31) {$[k_j]$};

    \coordinate (walkUL) at (\walkLeft, 0.98); %
    \coordinate (walkUR) at (\walkRight,0.16); %
    \coordinate (walkLL) at (\walkLeft,-1.48); %
    \coordinate (walkLR) at (\walkRight,-1.48);%
    \WalkEqSupportSecondTuned{7}{2}{3}{6}{1}{8}{5}{4}

    \foreach \lab/\yy in {1/1.80,3/0.16,4/-0.66,n/-3.02}{
        \node[walk unused vertex] at (\walkLeft,\yy) {};
        \node[walk unused label,anchor=east]
            at ({\walkLeft-0.20},\yy) {$\lab$};
    }
    \foreach \lab/\yy in {1/0.98,3/-0.66,k_j/-3.02}{
        \node[walk unused vertex] at (\walkRight,\yy) {};
        \node[walk unused label,anchor=west]
            at ({\walkRight+0.20},\yy) {$\lab$};
    }
    \foreach \lab/\yy in {2/0.98,5/-1.48}{
        \node[walk vertex] at (\walkLeft,\yy) {};
        \node[walk vertex label,anchor=east]
            at ({\walkLeft-0.20},\yy) {$\lab$};
    }
    \foreach \lab/\yy in {2/0.16,4/-1.48}{
        \node[walk vertex] at (\walkRight,\yy) {};
        \node[walk vertex label,anchor=west]
            at ({\walkRight+0.20},\yy) {$\lab$};
    }
    \node[text=black!55,inner sep=0pt] at (\walkLeft,-2.20) {$\vdots$};
    \node[text=black!55,inner sep=0pt] at (\walkRight,-2.20) {$\vdots$};
    \node[font=\large,inner sep=0pt] at (0,-3.58) {$w^{(2)}\in\cW$};
\end{scope}

\end{tikzpicture}
\endgroup}
    \caption{An example of an equivalence class $\cW$ of a walk on a bipartite graph, which gives the shape of the walk, and elements $w^{(1)}$ and $w^{(2)}$ of the equivalence class, which are possible realizations of labelings. }
    \label{fig:labelings-of-equiv-class}
\end{wrapfigure}
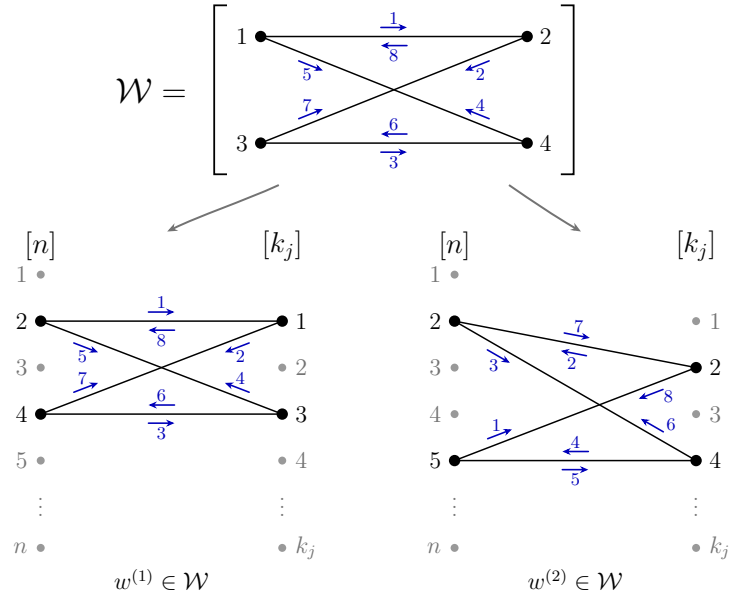
\paragraph{Using sparsity to bound labelings}
Fix deficiencies $\Delta_1,\dots,\Delta_q$ and let $\cW_1,\dots,\cW_q$ be equivalence classes of walks where $\cW_j\in\fW_m(\Delta_j)$. We must bound how many choices of $x\in X^m$, $y\in Y^m$, and $r_j\in[k_j]^{2m}$ for $j\in[q]$ there are that yield sign monomials $\Phi_j(x,y,r_j)$ with their corresponding walks in each $\cW_j$. See \Cref{fig:labelings-of-equiv-class} for explicit examples of elements in an equivalence class of $\fW_2(1)$.

There are no restrictions on how we can choose the right vertices $r_1,\dots,r_q$, so we take the naive bound of $\prod_{j=1}^qk_j^{\wt_R(j)}$. For the left vertices in each walk, they must adhere to both the shapes we have chosen and the $\eta$-sparsity of $X$ and $Y$. 

To see how the shape $\cW_j$ comes into play, consider the example where the path returns to its first vertex after four steps, as in \Cref{fig:labelings-of-equiv-class}. Consequently, when we give a label to the first vertex, we are also labeling the fifth vertex in this walk. 

For the general setting, we will consider choosing the vectors $x_1,y_1,x_2,y_2,\dots,x_m,y_m$ in order. Due to the shapes $\cW_1,\dots,\cW_q$ being fixed, whenever we choose a new vector, it's possible that some of its elements have already been fixed by the choices of previous vectors. Say that there are $f\geq 1$ free (unfixed) coordinates when we're choosing a vector. Since all of these vectors must lie in either of the $\eta$-sparse sets $X$ or $Y$, a consequence of sparseness, shown in \Cref{fact:number of completions in sparse set}, is that there are at most $\eta n^f$ ways to choose a valid vector.

Now, take $W_\sL=\sum_{j=1}^q \wt_\sL(j)$. Each vector can have at most $q$ unfixed coordinates, while at least $W_\sL/q$ of the vectors have an unfixed coordinate. Therefore, there are at most $n^{\wt_\sL}\eta^{\wt_\sL/q}$ possible ways to choose $x_1,y_1,x_2,y_2,\dots,x_m,y_m$ with $x_t\in X$ and $y_t\in Y$ that respect the fixed shapes $\cW_1,\dots,\cW_q$ of each closed walk. Including our bound on the right vertices $r_1,\dots,r_q$, we get that each choice of shapes $(\cW_j)_{j\in[q]}$ has at most
\begin{align}
    n^{W_\sL}\eta^{W_\sL/q}\prod_{j=1}^qk_j^{\wt_\sR(j)}\label{eq:overview-bound on prod of equivalence classes}
\end{align}
possible labelings. This is the only part of our bound that uses the sparsity of $X$ and $Y$.

Finally, to achieve the desired bound in \Cref{eq:overview-bound on exp max opnorm T}, all that is left to do is to combine \Cref{eq:overview-bound on prod of equivalence classes} with the bound from \cite{anderson_clt_2006}, sum over all possible deficiencies, and normalize for the expectation. After much simplification, this yields
\begin{align*}
    \E_\sigma[\trace((T_\lambda T_\lambda^\transpose)^m)]\leq 2^q n^q \eta^{m+1} (16^q \rho^{-2q})^m,
\end{align*}
as shown in \Cref{lem:UB on expectation of trace of power}. Then we can achieve \Cref{eq:overview-bound on exp max opnorm T} by summing this bound over all $n^\nu$ choices of $\lambda$, taking the $(2m)$-th root, and setting the power in our trace-power argument to $m=D\log n$. 

\dobib

\section{Preliminaries}\label{sec:preliminaries}
We introduce some basic definitions and facts that we will use throughout this work. In \Cref{subsec:prelim-lines planes and sparseness} we present the central objects of planes,  \ec and \es sets, and consequences of sparseness. \Cref{subsec:prelim-maximal planes} focuses on the notion of maximal planes that, as shown in \Cref{subsec:prelim-anchor sets}, are sparsely distributed if they come from an \ec set. We then shift to tensor codes and their related notions in \Cref{subsec:prelim-tensor codes and related} before finishing with some necessary facts about operator norms in \Cref{subsec:prelim-linear alg}.
\subsection{Lines, planes, and sparseness}\label{subsec:prelim-lines planes and sparseness}
We begin by defining the central combinatorial objects that we will work on throughout this paper. Throughout, we take integers $D\geq1$ and $n\geq2$, so that the direction and anchor of a plane are uniquely determined by it as a subset of $[n]^D$.

\begin{definition}[Axis-parallel planes]
    For a set of free coordinates $I\subseteq[D]$ and $a\in[n]^{\setcomp{I}}$, we define the \emph{axis-parallel plane} at the \emph{anchor} $a$ as
    \begin{align*}
        \plane(I,a) &=\{x\in[n]^D:\forall j\in \setcomp{I},\  x_j=a_j\}.
    \end{align*}
    So points have $I=\varnothing$, lines have $\abs{I}=1$, and the whole cube is given when $I=[D]$. 
\end{definition}
It will be useful to focus on the specific setting of 1-dimensional planes, which we call axis-parallel lines.
\begin{definition}
    For $i\in[D]$ let $I=\{i\}$ and $\lines_i(n,D)=\{\plane(I,a):a\in[n]^{\setcomp{I}}\}$. All axis-parallel lines are denoted as $\lines(n,D)=\bigcup_{i\in[D]}\lines_i(n,D)$.
\end{definition}

\begin{definition}
    Let $\lines(M)=\{\ell\in\lines(n,D):\ell\subseteq M\}$ denote the set of axis-parallel lines contained in a set $M\subseteq[n]^D$.
\end{definition}
There are two notions of sparseness that we will focus on. The first is similar to a gap property and requires that a subset of $[n]^D$ be either sparse along an axis-parallel line or completely contain it.
\begin{definition}[\ec]\label{def:eps closed}
    A set $M\subseteq[n]^D$ is \emph{\ec} if for all $\ell\in\lines(n,D)$, we have that either $\ell\subseteq M$ or $\abs{M\cap \ell}\leq\varepsilon n$.
\end{definition}
Our second definition of sparseness requires that no axis-parallel line is fully contained in our set.
\begin{definition}[\es]
    We say that an \ec set $M$ is \emph{$\varepsilon$-sparse} if no axis-parallel line is contained in $M$.
\end{definition}

One immediate consequence of sparseness is a bound on the cardinality of $M$.

\begin{lemma}\label{lem:eps closed and line free gives cardinality bound}
    For any \es set $M\subsetneq[n]^D$,  we can bound its cardinality as
    \begin{align*}
        \abs{M}\leq\varepsilon n^D.
    \end{align*}
\end{lemma}
\begin{proof}
    Fix any direction $i\in[D]$. Then any axis-parallel line $\ell\in\cL_i(n,D)$ has $\abs{M\cap \ell}\leq\varepsilon n$. Since there are $n^{D-1}$ axis-parallel lines in the $i$-th direction, we get
    \begin{align*}
        \abs{M}&\leq \sum_{\ell\in \cL_i(n,D)}\abs{M\cap \ell}\\
        &\leq n^{D-1}\cdot \varepsilon n\\
        &=\varepsilon n^D,
    \end{align*}
    as claimed.
\end{proof}

We will also make use of the following fact about sparse sets.
\begin{fact}\label{fact:number of completions in sparse set}
    For a constant $\varepsilon\in(0,1)$ and an \es set $A\subseteq[n]^D$, if we fix $I\subsetneq[D]$ and $a_I\in[n]^I$, then
    \begin{align*}
        \abs{\{a_{\setcomp{I}}\in[n]^{\setcomp{I}}:a=(a_I,a_{\setcomp{I}})\in A\}}\leq \varepsilon n^{\abs{\setcomp{I}}}.
    \end{align*}
\end{fact}
\begin{proof}
    For $\abs{\setcomp{I}}-1$ of the free coordinates to assign in $a_{\setcomp{I}}$, give arbitrary assignments, of which there are $n^{\abs{\setcomp{I}}-1}$ many. The last coordinate not assigned in $\setcomp{I}$ then only has $\varepsilon n$ many options by the $\varepsilon$-sparseness of $A$. Thus, there are at most $n^{\abs{\setcomp{I}}-1}\cdot\varepsilon n=\varepsilon n^{\abs{\setcomp{I}}}$ many assignments of $a_{\setcomp{I}}$.
\end{proof}

\subsection{Maximal planes}\label{subsec:prelim-maximal planes}
We first define the set of all planes contained in $M$.
\begin{definition}
    For a set $M\subseteq[n]^D$, we define the set of all planes in $M$ to be
    \begin{align*}
        \allplanes(M)&=\{\plane(I,a)\subseteq M : I\subseteq[D], a\in[n]^{\setcomp{I}}\}.
    \end{align*}
    In particular, we do allow $I=\varnothing$. 
\end{definition}

\begin{figure}[ht]
    \centering
    \includegraphics[width=0.4\linewidth]
    {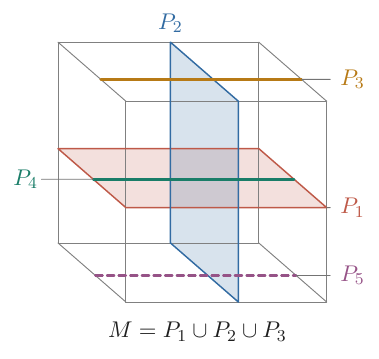}
    \caption{In this example, $M$ is composed of three planes $M=P_1\cup P_2\cup P_3$, each of which is a maximal plane. On the other hand, $P_4$ is a plane of $M$ but not maximal because it's contained in $P_1$, and $P_5$ isn't a plane of $M$ because it's not fully contained in $M$.}
    \label{fig:max planes example}
\end{figure}

Throughout our work, we will be particularly focused on the maximal planes of $M$ because they give a natural way to decompose $M$ into structured components.
\begin{definition}[Maximal planes in $M$]
    For a set $M\subseteq[n]^D$, we define the set $\maxplanes(M)$ of \emph{maximal planes} of $M$ to contain $\plane(I,a)$ for $I\subseteq[D]$ and $a\in[n]^{\setcomp{I}}$ if 
    \begin{enumerate}
        \item it is contained in $M$: $\plane(I,a)\subseteq M$
        \item and there is no other $\plane(J,b)$ such that $\plane(I,a)\subsetneq P(J,b)\subseteq M$.
    \end{enumerate}
    See \Cref{fig:max planes example} for an example illustration.
\end{definition}
As we will see later, these maximal planes interact at their intersections, so we will need to consider the set of planes generated as intersections of maximal planes.
\begin{definition}
    For a set $M\subseteq[n]^D$, we define $\maxplanesclosed(M)\subseteq\allplanes(M)$ to be the closure by intersection of $\maxplanes(M)$, not including empty intersections. Formally,
    \begin{align*}
        \maxplanesclosed(M)=\left\{
        \bigcap_{P\in\mathcal S}P:
        \varnothing\ne\mathcal S\subseteq\maxplanes(M),\
        \bigcap_{P\in\mathcal S}P\ne\emptyset
        \right\}.
    \end{align*}
\end{definition}

Using these definitions, we can show the simple fact that any $M\subseteq[n]^D$ can be covered by maximal planes.
\begin{lemma}\label{lem:M covered by maximal planes}
    For any $M\subseteq[n]^D$, we have
    \begin{align*}
        M=\bigcup_{I\subseteq[D]}\bigcup_{a\in \maxAI(M)}\plane(I,a).
    \end{align*}
\end{lemma}
\begin{proof}
    Every maximal plane is contained in $M$, which gives the inclusion $\supseteq$. Conversely, any $x\in M$ is itself a zero-dimensional plane $\{x\}=\plane(\varnothing,x)\subseteq M$, so it belongs to some element of $\allplanes(M)$. Since $\allplanes(M)$ is a finite poset under inclusion, $\{x\}$ is contained in a maximal one, which is by definition an element of $\maxplanes(M)$ containing $x$.
\end{proof}

\subsection{Anchor sets}\label{subsec:prelim-anchor sets}

Next, given a set of free coordinates $I\subseteq[D]$, we can define the set of fixed coordinates that index all maximal planes.
\begin{definition}[Indexing set of maximal planes in direction $I$]
    For a set $M\subseteq[n]^D$ and free coordinates $I\subseteq[D]$, we define the set of \emph{anchors} indexing the maximal planes as
    \begin{align*}
        \maxAI(M)&=\{a\in[n]^{\setcomp{I}} : \plane(I,a)\in\maxplanes(M)\}.
    \end{align*}
\end{definition}

Now, we will show that $M$ being \ec makes $\maxAI(M)$ \es.
\begin{lemma}\label{lem:M eps closed implies maxAI is eps closed and line free}
    Let $0<\varepsilon<1$ and let $M\subsetneq [n]^D$ be \ec. Then for any $I\subsetneq[D]$, the set $\maxAI(M)\subseteq[n]^{\setcomp{I}}$ is \es. Consequently,
    \begin{align*}
        \abs{\maxAI(M)}\leq\varepsilon n^{D-\abs{I}}.
    \end{align*}
\end{lemma}
\begin{proof}
    \begin{figure}[ht]
        \centering
        \includegraphics[width=0.6\linewidth]
    {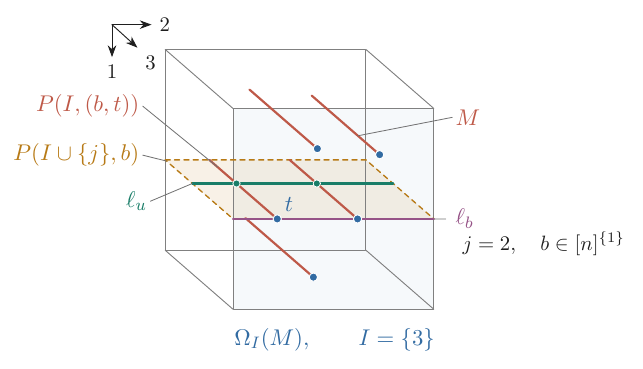}
        \caption{Visualization of the proof of \Cref{lem:M eps closed implies maxAI is eps closed and line free} for $D=3$ with $I=\{3\}$, $j=2$, and an $M$ with lines parallel to the third axis only.}
        \label{fig:proof of maxAI eps closed and line free}
    \end{figure}
    This proof is visually simple, and we recommend that the reader follow along in \Cref{fig:proof of maxAI eps closed and line free}. Suppose, for the sake of contradiction, that there exist $j\in\setcomp{I}$ and $b\in[n]^{\setcomp{I}\setminus\{j\}}$ such that the line $$\ell_b=\{s\in[n]^{\setcomp{I}}:s_j\in[n],\ s_{\setcomp{I}\setminus\{j\}}=b\}$$ in the $j$-th direction has $\abs{\maxAI(M)\cap \ell_b}>\varepsilon n$. For each $t\in[n]$ with $(b,t)\in\maxAI(M)\cap\ell_b$, it must be that $P(I,(b,t))\in\maxplanes(M)$ where $(b,t)\in[n]^{\setcomp{I}}$ is the element with the $j$-th coordinate as $t$ and the other coordinates set as $b$.

    Now, consider any $u\in[n]^I$ (the empty tuple when $I=\varnothing$) and the $j$-parallel line $$\ell_u=\{x\in[n]^D:x_I=u,\ x_{\setcomp{I}\setminus\{j\}}=b,\ x_j\in[n]\}.$$ By construction, for any $t\in[n]$ with $(b,t)\in\maxAI(M)\cap\ell_b$, we have that the point $y=(y_I=u, y_{\setcomp{I}\setminus\{j\}}=b,y_j=t)\in\ell_u$ is contained in the maximal plane $P(I,(b,t))$. Because $y\in \ell_u\cap P(I,(b,t))$ and $P(I,(b,t))\subseteq M$ and $\abs{\maxAI(M)\cap\ell_b}>\varepsilon n$, we have that $\abs{\ell_u\cap M}>\varepsilon n$. Using the fact that $M$ is \ec, we get that $\ell_u\subseteq M$.

    Our choice of $u\in[n]^I$ was arbitrary, so we can take their union to conclude that the plane
    \begin{align*}
        P(I\cup\{j\},b)=\bigcup_{u\in[n]^I}\ell_u
    \end{align*}
    is fully contained in $M$. However, we then get that $P(I\cup\{j\},b)\supsetneq P(I,(b,t))$ for any $t\in[n]$ with $(b,t)\in\maxAI(M)\cap\ell_b$, contradicting the fact that $P(I,(b,t))$ is maximal. Therefore, for all $j\in\setcomp{I}$ and $b\in[n]^{\setcomp{I}\setminus\{j\}}$, the line $\ell_b$ has $\abs{\maxAI(M)\cap\ell_b}\leq\varepsilon n<n$, meaning that $\maxAI(M)$ is \es.

    Finally, now that we have that $\maxAI(M)$ is \es in $[n]^{\setcomp{I}}$, by \Cref{lem:eps closed and line free gives cardinality bound} we can bound its cardinality as $\abs{\maxAI(M)}\leq\varepsilon n^{D-\abs{I}}$, as claimed.
\end{proof}
We can get a similar result for the indexing set $\maxAIclosed(M)=\{a\in[n]^{\setcomp{I}}:\plane(I,a)\in\maxplanesclosed(M)\}$ of $\maxplanesclosed(M)$. This is a loose bound and introduces a factor of $2^D$, but since $D$ is a constant, this is not an issue.
\begin{corollary}\label{cor:maxAIclosed is eps sparse}
    Let $M\subsetneq[n]^D$ be \ec with $0<2^D\varepsilon<1$. Then for every $I\subsetneq[D]$, the set $\maxAIclosed(M)\subseteq[n]^{\setcomp{I}}$ is $(2^D\varepsilon)$-sparse.
\end{corollary}
\begin{proof}
    Consider a direction $j\in\setcomp{I}$ and an axis-parallel line $\ell_j$ in $[n]^{\setcomp{I}}$ in the $j$-th direction. For every anchor $a\in\ell_j\cap\maxAIclosed(M)$, write $\plane(I,a)$ as an intersection of maximal planes. At least one of these planes, say $\plane(J,b)$, must fix coordinate $j$, so $I\subseteq J$, $j\notin J$, and $b=a_{\setcomp{J}}$. For a fixed such $J$, all coordinates of $b$ except $j$ are fixed by $\ell_j$. By \Cref{lem:M eps closed implies maxAI is eps closed and line free}, $\maxA{J}(M)$ is \es, meaning that there are at most $\varepsilon n$ possible values of $b_j=a_j$. There are at most $2^D$ choices of $J$, so a union bound gives at most $2^D\varepsilon n$ anchors in $\ell_j\cap\maxAIclosed(M)$.
\end{proof}

\subsection{Tensor codes and related notions}\label{subsec:prelim-tensor codes and related}

We now introduce the notions of inner-generated and extendable sets. In the next section, we will go over some of these definitions again and see how they are naturally derived in a slightly different context, but we introduce their standard definitions here as they show up in the literature. Throughout this section, we will let $\gfield$ represent a generic field, possibly finite or infinite.
Throughout, we identify a line $\ell=\plane(\{i\},a)\in\cL_i(n,D)$ with $[n]$ via its $i$-th coordinate, so that the restriction $x|_\ell$ of a word is an element of $\gfield^n$ and can be tested for membership in the component code $\cC_i$.
\begin{definition}[Codes along a line $\ell$]
    For $\ell\in\cL_i(n,D)$ and a code $\cC_i\subseteq\gfield^n$, define
\[
\cC_\ell=\{c\in\gfield^{[n]^D}:\supp(c)\subseteq \ell,\ c|_\ell\in\cC_i\}.
\]
\end{definition}

\begin{definition}[Check-product code]\label{ex:Dual cosheaf code of a tensor product code}
    Given a collection of codes $\cC_1,\dots,\cC_D$, the dual of their tensor code is the check-product code $\cC_1^\perp\boxplus\cdots\boxplus\cC_D^\perp=\left(\cC_1\otimes\cdots\otimes\cC_D\right)^\perp$. 

    An easy way of seeing the relation between the check-product code and the tensor product code is to define $\cC^{(i)}=\gfield^n\otimes\cdots\otimes\cC_i\otimes\cdots\otimes\gfield^n\subseteq\gfield^{[n]^D}$ which allows us to write
    \begin{align*}
        \bigotimes_{i=1}^D\cC_i &=\bigcap_{i=1}^D\cC^{(i)}\\
        \bigboxplus_{i=1}^D\cC_i&=\sum_{i=1}^D\cC^{(i)}.
    \end{align*}
\end{definition}

Instead of considering the entire tensor code at once, we can also consider words that satisfy local checks contained in a set of coordinates.
\begin{definition}[Words accepted by checks in $M$]\label{def:sheaf(M)}
    For a subset of coordinates $M\subseteq[n]^D$ of a tensor code $\cC=\cC_1\otimes\cdots\otimes\cC_D$, let 
    \begin{align*}
        \sheaf_{\cC}(M)=\{x\in\gfield^M:\forall i\in[D],\ \forall\ell\in\cL_i(n,D),\
\ell\subseteq M\Rightarrow x|_\ell\in\cC_i\}
    \end{align*}
    be the set of words that satisfy the local checks of the corresponding component code for each axis-parallel line contained in $M$.
\end{definition}
\begin{definition}[Extendability]
We say that $M$ is \emph{extendable} in $\cC$ if the restriction map
\[
\pi_M:\cC\to\sheaf(M)
\]
is surjective. That is, each local codeword in $\sheaf(M)$ can be extended to a global codeword of $\cC$.
\end{definition}

\begin{definition}[Inner-generated sets for the check-product code]\label{def:inner-generated}
    A set $M\subseteq[n]^D$ is \emph{inner-generated} for a code $\cC_1\boxplus\cdots\boxplus\cC_D$ if for each codeword $c\in\cC_1\boxplus\cdots\boxplus\cC_D$ with $\supp(c)\subseteq M$ we have $c\in\sum_{\ell\in\lines(M)}\cC_\ell$, where $\lines(M)$ is the set of axis-parallel lines contained in $M$.
\end{definition}

We can now define the central notion of expansion used in this paper. For a word $a\in\gfield^{[n]^D}$ we write $\abs{a}$ for its Hamming weight, and for a direction $i\in[D]$ we write
\begin{align*}
    \abs{a}_i=\abs{\{\ell\in\cL_i(n,D):a|_\ell\neq0\}}
\end{align*}
for the number of $i$-parallel lines on which $a$ does not vanish.
\begin{definition}[Product expansion]\label{def:product expansion}
    A tuple of codes $(\cC_i)_{i\in[D]}$ is \emph{$\nu$-product expanding} if every $x\in\bigboxplus_{i=1}^D\cC_i$ admits a decomposition $x=\sum_{i=1}^D a_i$ with $a_i\in\cC^{(i)}$ such that
    \begin{align*}
        \nu n\sum_{i=1}^D\abs{a_i}_i\leq\abs{x}.
    \end{align*}
    It is \emph{two-sided $\nu$-product expanding} if both $(\cC_i)_{i\in[D]}$ and $(\cC_i^\perp)_{i\in[D]}$ are $\nu$-product expanding, where duals are taken with respect to the standard coordinate pairing on $\gfield^n$.
\end{definition}
Finally, for a prime $p$ we write $\grass_{p^t}(n,k)$ for the Grassmannian of $k$-dimensional subspaces of $\F_{p^t}^n$. Sampling a random code of dimension $k$ always means sampling uniformly from $\grass_p(n,k)$.

\subsection{Miscellanea}
\subsubsection{Linear algebra facts}\label{subsec:prelim-linear alg}
We record some useful linear algebra facts that we will use at various points later on in this work. Here all spaces are finite-dimensional real or complex Hilbert spaces, and we use the Euclidean operator norm.
\begin{fact}[Equation (2.79) and Exercise 2.3.16 of \cite{tao_topics_2012}]\label{fact:opnorm^2p is bounded by trace^p}
    For a linear map $T$ and integer $p\geq 1$,
    \begin{align*}
        \opnorm{T}^{2p}\leq\trace((TT^\adjoint)^p).
    \end{align*}
\end{fact}
\begin{fact}[Equation (1.175) of \cite{watrous_theory_2018}]
    For two compatible linear maps $A,B$, we have that the operator norm is submultiplicative
    \begin{align}
        \opnorm{AB}&\leq\opnorm{A}\opnorm{B}.
    \end{align}
\end{fact}
\begin{fact}[Equation (2.2.10) of \cite{tropp_introduction_2015} or Problem 6(b) of \cite{johnson_singular_1991}]\label{fact:opnorm of direct sum is max of opnorms}
    For a family of linear maps $(A_s)_{s\in S}$,
    \begin{align*}
        \opnorm{\bigoplus_{s\in S} A_s}&=\max_{s\in S}\opnorm{A_s}.
    \end{align*}
\end{fact}
\begin{fact}[Theorem 4.2.15 and Problem 27 of \cite{johnson_matrix_1991}]\label{fact:opnorm of tensor is prod of opnorms}
    For a family of linear maps $(A_s)_{s\in S}$,
    \begin{align*}
        \opnorm{\bigotimes_{s\in S}A_s}&=\prod_{s\in S}\opnorm{A_s}.
    \end{align*}
\end{fact}

\subsubsection{A useful estimate on the distribution of primes}
Here we record a consequence of Chebyshev's estimates that bounds the value of a prime not dividing a given integer. We begin by recalling Chebyshev's estimate.

\begin{theorem}[Corollary 2.6 of \cite{montgomery_elementary_2006}]\label{thm:PNT}
    For any $x\geq 2$, the first Chebyshev function $\vartheta(x):=\sum_{\substack{p\leq x\\p\text{ prime}}}\log p$ is bounded as $\vartheta(x)=\Theta(x)$.
\end{theorem}
We now derive our required corollary.
\begin{corollary}\label{cor:PNT bounding non-dividing prime}
    There exists a universal constant $C$ such that for all $q\geq 2$ and every $m\in\N$, there exists a prime $p$ \emph{not} dividing $m$ such that
    \begin{align*}
        q< p\leq C(q+\log(2 m)).
    \end{align*}
\end{corollary}
\begin{proof}
    By \Cref{thm:PNT}, there exist universal constants $0<A\leq B$ such that for all $x\geq 2$ we have \begin{align}
        Ax\leq\vartheta(x)\leq Bx.\label{eq:Ax<=cheby(x)<=Bx}
    \end{align}
    Let $\ell=\log(2m)$ and $s=\frac{B+1}{A}(q+\ell)$. Now, for the sake of contradiction, assume that all primes $p\in(q,s]$ divide $m$. These primes are of course distinct, so $\prod_{\substack{p\in(q,s]\\p\text{ prime}}} p$ must divide $m$. Consequently,
    \begin{align*}
        \vartheta(s)-\vartheta(q)&=\sum_{\substack{p\leq s\\p\text{ prime}}}\log p - \sum_{\substack{p\leq q\\p\text{ prime}}}\log p\\
        &=\log\prod_{\substack{q<p\leq s\\p\text{ prime}}}p\\
        &\leq\log m.\numberthis\label{eq:cheby(s)-cheby(q) <= log m}
    \end{align*}
    However, from \Cref{eq:Ax<=cheby(x)<=Bx} we obtain
    \begin{align*}
        \vartheta(s)-\vartheta(q)&\geq As-Bq\\
        &=(B+1)(q+\ell)-Bq\\
        &=q+(B+1)\ell\\
        &>\log m,
    \end{align*}
    which contradicts \Cref{eq:cheby(s)-cheby(q) <= log m}. Therefore, there must exist some prime $p\in(q,s]$ not dividing $m$ such that
    \begin{align*}
        q< p\leq s=\frac{B+1}{A}(q+\log(2m)),
    \end{align*}
    as claimed.
\end{proof}
\dobib

\section{(Co)Restriction Maps and Atoms}\label{sec:duality-atoms}
In this section, we will give intuition for the duality between the properties of extendability and inner-generatedness. We will detail how exactly the tensor product code and check-product codes are dual to each other along with how surjectivity of the restriction map of the tensor code, which is the condition for extendability, is equivalent to injectivity of the corestriction map for the check-product code, which is the condition for inner-generatedness. Since many of these definitions can be defined over any arbitrary field, and we will use them over both finite fields and the reals, we will use the generic field $\gfield$, which may be finite or infinite, unless we specify otherwise.

In \Cref{subsec:atoms-tensor restrictions and extendability} we introduce the restriction map on tensor codes to subsets of coordinates and its relation to extendability, and in \Cref{subsec:atoms-check-product corestrictions and inner-generated} we define the dual notion of the corestriction map on the check-product code and its relation to inner-generatedness.\footnote{The reader familiar with tensor codes may skip these two sections.} We combine these notions in \Cref{subsec:atoms-all together} to show precisely how they are related and give an explicit formula for the corestriction maps. Finally, in \Cref{subsec:atoms-atom decomposition} we present the  notion of the atom decomposition of the check-product code.

\subsection{The tensor code, restrictions, and extendability}\label{subsec:atoms-tensor restrictions and extendability}

We begin by choosing generator matrices 
\begin{align*}
    G_i=\left[\begin{matrix}
        \vert & \vert &  & \vert\\
        g_i(1) & g_i(2) & \cdots & g_i(n)\\
        \vert & \vert &  & \vert\\
    \end{matrix}\right]\in\gfield^{k_i\times n}.
\end{align*}
for codes $$\cC_1,\dots,\cC_D\subseteq\gfield^n$$ such that $\cC_i=\Im(G_i^\transpose)$ and $\dim\cC_i=k_i$. Their tensor code is denoted as $\cC=\cC_1\otimes\cdots\otimes\cC_D\subseteq\gfield^{[n]^D}$.

For $P=\plane(I,a)\in\allplanes(M)$, define
\begin{align*}
    \cC_P=\cC_I=\bigotimes_{i\in I}\cC_i\subseteq\gfield^{[n]^I}.
\end{align*}
If $I=\varnothing$, then we define $\cC_I=\gfield$.
Between these vector spaces, we can define restriction maps.
\begin{definition}[Restriction maps]
    For any $P,Q\in\allplanes(M)$ such that $P=\plane(I,a)\subseteq Q=\plane(J,b)$, so $I\subseteq J\subseteq[D]$, define the restriction map
\[
\pi_{Q\to P}
=
\left(\bigotimes_{j\in J\setminus I}\ev_{a_j}\right)
\otimes
\left(\bigotimes_{i\in I}\mathrm{id}_{\cC_i}\right)
,
\]
where tensor factors are placed in increasing coordinate order and $\ev_{a_j}$ is the evaluation map at coordinate $a_j$. Equivalently, we can define the restriction map on pure tensors $c=\bigotimes_{j\in J}c_j$ as 
\begin{align*}
    \res_{Q\to P}:\sheaf(Q)&\to\sheaf(P)\\
    c &\mapsto \left(\prod_{j\in J\setminus I} c_j(a_j)\right)\left(\bigotimes_{i\in I} c_i\right),
\end{align*}
and extend linearly to the whole space.
\end{definition}

In particular, one can check that for planes $P\subseteq Q\subseteq R$, we have
\begin{align*}
    \res_{R\to P}=\res_{Q\to P}\circ\res_{R\to Q}.
\end{align*}
\begin{figure}[t]
    \centering
    \includegraphics[width=0.75\linewidth]
    {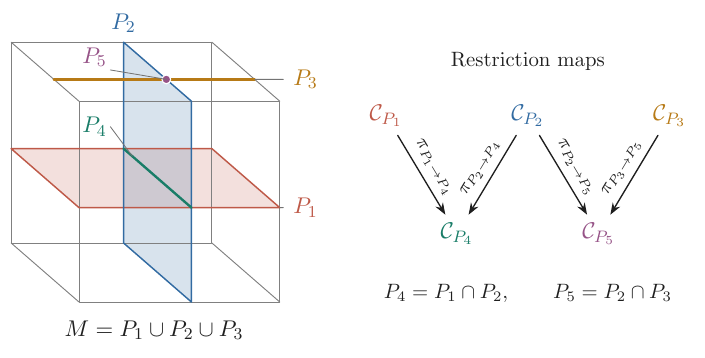}
    \caption{Example of planes in $M=P_1\cup P_2\cup P_3$ and part of the diagram that is ultimately being used for the limit over the $\cC_P$'s.}
    \label{fig:example subdiagram for limit of planes}
\end{figure}
See \Cref{fig:example subdiagram for limit of planes} for an example.
Therefore, because \Cref{lem:M covered by maximal planes} shows that $M$ is covered by its maximal planes, we have that for any  $M\subseteq[n]^D$ we can order $\allplanes(M)$ by inclusion and define the local sections on $M$ as limits:\footnote{We will not use (co)limits outside of this section and they do not play a central role in the rest of this paper, so we do not define them formally and instead simply expand out their definitions in this context.}
\begin{align*}
    \sheaf(M)&=\lim_{P\in\maxplanesclosed(M)} \cC_P\\
    &=\{(c_P)_{P\in\maxplanesclosed(M)}:c_P\in\cC_P, \text{ s.t. for all }P\subseteq Q,\ \pi_{Q\to P}(c_Q)=c_P\}\\
    &\cong\{x\in\gfield^M:\forall i\in[D],\ \forall\ell\in\cL_i(n,D),\ \ell\subseteq M\Rightarrow x|_\ell\in\cC_i\}.
\end{align*}
The last line is an isomorphism rather than an equality because a compatible family is indexed by planes and not by the points of $M$. It is a bijection because $M$ is covered by its maximal planes by \Cref{lem:M covered by maximal planes} (some of which may be zero-dimensional), and any two planes of $\maxplanesclosed(M)$ agree on their intersection whenever it is nonempty, since that intersection again lies in $\maxplanesclosed(M)$. So a compatible family glues to a single word $x\in\gfield^M$, and conversely every such $x$ restricts to a compatible family. Since every line contained in $M$ lies in a maximal plane, the constraints on $x$ are exactly the local checks of \Cref{def:sheaf(M)}, and taking $M=[n]^D$ gives $\cC=\sheaf([n]^D)$.
Finally, to interpret extendability in this context, we consider the global restriction map
\begin{align*}
    \res_M:\cC\to\sheaf(M).
\end{align*}
\begin{fact}\label{claim:extendable iff res surjective}
    A set $M\subseteq[n]^D$ is extendable if and only if $\res_M$ is surjective.
\end{fact}
\begin{proof}
    The condition that a codeword $c_M\in\sheaf(M)$ can extend to a global word $c\in\cC$ such that $c|_M=c_M$ is exactly the condition that $\pi_M$ is surjective, which allows us to ``lift'' $c_M$ to $c$ (possibly non-uniquely). 
\end{proof}

\subsection{The check-product code, corestrictions, and inner-generatedness}\label{subsec:atoms-check-product corestrictions and inner-generated}
In this subsection, we will essentially take the dual of the previous subsection and derive the corresponding notions. Let $\cB_i=\cC_i^\perp$ for $i\in[D]$, and define
\begin{align*}
    \cB=\cB_1\boxplus\cdots\boxplus\cB_D\subseteq\gfield^{[n]^D}.
\end{align*}
Similarly, for $P=\plane(I,a)\in\allplanes(M)$, let
\begin{align*}
    \cB_P=\cB_I=\bigboxplus_{i\in I}\cB_i\subseteq\gfield^{[n]^I}.
\end{align*}
For a point $P$ we have $I=\varnothing$, and we take the empty check-product to be $\cB_\varnothing=0$, matching the empty tensor product $\cC_\varnothing=\gfield$.
As we saw in \Cref{ex:Dual cosheaf code of a tensor product code}, $\cC^\perp=\cB$, and similarly for any plane $P=\plane(I,a)\in\allplanes(M)$, it is easy to verify that
\begin{align*}
    \cC_P^\perp&=\left(\bigotimes_{i\in I}\cC_i\right)^\perp\\
    &=\bigboxplus_{i\in I}\cC_i^\perp\\
    &=\bigboxplus_{i\in I}\cB_i\\
    &=\cB_P.
\end{align*}
Therefore, to dualize the construction of tensor codes above, we assign to each $P\in\allplanes(M)$ the space
\begin{align*}
    \cosheaf(P)&=\sheaf(P)^*\\
    &=\cC_P^*\\
    &=\Hom_\gfield(\cC_P,\gfield)\\
    &\cong \gfield^P/\cC_P^\perp\\
    &=\gfield^P/\cB_P.
\end{align*}
Between these local spaces, we can define the dual of the restriction map, the corestriction map.
\begin{definition}[Corestriction map]\label{def:cores}
    For $P,Q\in\allplanes(M)$ with $P\subseteq Q$, restriction goes from $Q$ to $P$, so its dual goes from $P$ to $Q$. We define the corestriction map
    \begin{align*}
        \cores_{P\to Q}=\res^*_{Q\to P}:\cC_P^*&\to\cC_Q^*\\
        \varphi &\mapsto\varphi\circ\res_{Q\to P}.
    \end{align*}
\end{definition}

Now, it is clear that for planes $P\subseteq Q\subseteq R$, their corestriction maps are compatible via
\begin{align*}
    \cores_{P\to R}=\cores_{Q\to R}\circ\cores_{P\to Q}
\end{align*}
exactly because $\res_{R\to P}=\res_{Q\to P}\circ\res_{R\to Q}$. 
See \Cref{fig:example subdiagram for colimit of planes} for an example.
\begin{figure}[t]
    \centering
    \includegraphics[width=0.75\linewidth]
    {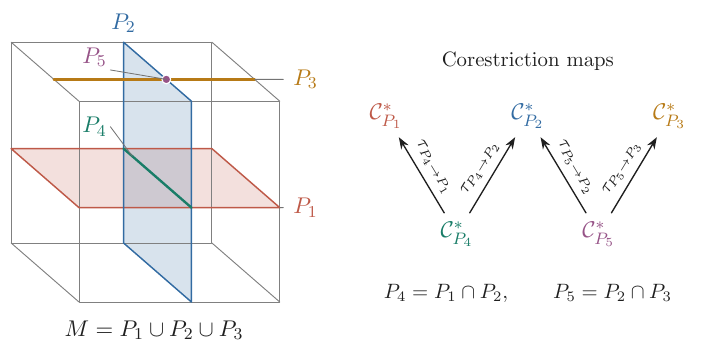}
    \caption{Example of planes in $M=P_1\cup P_2\cup P_3$ and part of the diagram that is ultimately being used for the colimit over the $\cC_P^*$'s.}
    \label{fig:example subdiagram for colimit of planes}
\end{figure}
Therefore, for any $M\subseteq[n]^D$ we can define
\begin{align*}
    \cosheaf(M)&=\sheaf(M)^*\\
    &=\left(\lim_{P\in\maxplanesclosed(M)} \cC_P\right)^*\\
    &\cong\colim_{P\in\maxplanesclosed(M)} \cC_P^*\numberthis\label{eq:colim of duals def of G(M)}\\
    &\cong\colim_{P\in\maxplanesclosed(M)}\gfield^P/\cB_P\\
    &\cong\frac{\gfield^M}{\sum_{\ell\in \cL(M)}\cB_\ell},
\end{align*}
where we are including each $\cB_\ell$ into $\gfield^M$ by extending by 0's outside of the line $\ell$. For the whole space, this gives $\cosheaf([n]^D)=\frac{\gfield^{[n]^D}}{\cB}$. We also get the global corestriction map
\begin{align*}
    \cores_M:\cosheaf(M)\cong\frac{\gfield^M}{\sum_{\ell\in \cL(M)}\cB_\ell}\to\frac{\gfield^{[n]^D}}{\cB}\cong\cosheaf([n]^D).
\end{align*}
We now show that the kernel of $\cores_M$ is exactly the obstruction to $M$ being inner-generated.
\begin{fact}\label{claim:inner generated iff cores injective}
    A set $M\subseteq[n]^D$ is inner-generated in $\cB$ iff $\cores_M$ is injective iff $\ker(\cores_M)=0$.
\end{fact}
\begin{proof}
    To be explicit, we define the inclusion $\iota_M:\gfield^M\to\gfield^{[n]^D}$ that extends an element of $\gfield^M$ by 0's to a full element of $\gfield^{[n]^D}$. Thus, an element $[z]\in\frac{\gfield^M}{\sum_{\ell\in \cL(M)}\cB_\ell}$ is in $\ker(\cores_M)$ iff $[\iota_M(z)]=0\in \frac{\gfield^{[n]^D}}{\cB}$, meaning that $\iota_M(z)\in \cB$. Equivalently, this means that we need $z\in\gfield^M\cap\cB$, so
    \begin{align*}
        \ker(\cores_M)\cong\frac{\gfield^M\cap \cB}{\sum_{\ell\in \cL(M)}\cB_\ell},
    \end{align*}
    where we are implicitly applying $\iota_M$ to $\gfield^M$.
    The definition of $M$ being inner-generated from \Cref{def:inner-generated} is exactly when $\gfield^M\cap \cB=\sum_{\ell\in \cL(M)}\cB_\ell$, which is exactly $\ker(\cores_M)=0$, as claimed.
\end{proof}

\subsection{Putting it all together}\label{subsec:atoms-all together}
Using the fact that a linear map is injective if and only if its dual is surjective, we get the following lemma.
\begin{lemma}\label{lem:equivalence of extendability inner-generated cores and res}
    For any set $M\subseteq[n]^D$, the following are equivalent:
    \begin{enumerate}
        \item $M$ is extendable in $\cC=\cC_1\otimes\cdots\otimes\cC_D$,
        \item $M$ is inner-generated in $\cB=\cB_1\boxplus\cdots\boxplus\cB_D$,
        \item $\res_M$ is surjective, and
        \item $\cores_M$ is injective.
    \end{enumerate}
\end{lemma}
\begin{proof}
    We just use \Cref{claim:extendable iff res surjective,claim:inner generated iff cores injective} together with the fact that a linear map is injective if and only if its dual is surjective.
\end{proof}
In the rest of this paper, our goal will be to show that $\cores_M$ is injective with high probability for a selection of random linear codes. In this vein, we show how the injectivity of $\cores_M$ depends on the generator matrices $(G_i)_{i\in[D]}$.

\begin{definition}[Dual message space]
    For $i\in[D]$, let $U_i$ be a $k_i$-dimensional vector space over $\gfield$.
    This is the \emph{dual message space} of $\cC_i$ since we can consider $\cC_i\cong U_i^*$ by taking the columns $g_i(1),\dots,g_i(n)$ of the generator matrix $G_i$ and writing down a message $\mu_i\in U_i^*$ as the codeword
\begin{align*}
    c_{\mu_i}=(\mu_i(g_i(1)),\mu_i(g_i(2)),\dots,\mu_i(g_i(n)))\in\cC_i\subseteq\gfield^n.
\end{align*}

For any plane $P=\plane(I,a)$, define $$U_P=U_I=\bigotimes_{i\in I}U_i. $$ We take $U_\varnothing=\gfield$ for a point.
\end{definition}

By definition of the $U_i$'s, we have that for a plane  $P=\plane(I,a)$,
\begin{align*}
    \cC_P^*&=\left(\bigotimes_{i\in I} \cC_i\right)^*\\
    &\cong\bigotimes_{i\in I}\cC_i^*\\
    &\cong\bigotimes_{i\in I}U_i\\
    &=U_I=U_P.
\end{align*}
Thus, we can make the computation of $\cores_{P\to Q}$ explicit.
\begin{fact}
     For planes $P=\plane(I,a)\subseteq Q=\plane(J,b)$ with $I\subseteq J$
    \begin{align*}
    \cores_{P\to Q}=\res^*_{Q\to P}:U_I&\to U_J\\
    u &\mapsto u\otimes\left(\bigotimes_{j\in J\setminus I} g_j(a_j)\right),\numberthis\label{eq:cores explicit def}
\end{align*}
where we order the tensor products by the ordering of $[D]$. 
\end{fact}
From this, we can derive a nice structural result on $\cosheaf(M)$.

\begin{lemma}\label{lem:cosheaf(M) as quotient of direct sum}
    We have that $\cosheaf(M)\cong U_M/ R_M$ where $U_M=\bigoplus_{P\in\maxplanesclosed(M)} U_P$ and
    \begin{align*}
    R_M=\Span\{e_P(u)-e_Q(\cores_{P\to Q}(u)) : P,Q\in \maxplanesclosed(M),\ P\subseteq Q,\ u\in U_P\}
    \end{align*}
     where $e_P:U_P\to U_M$ is the natural inclusion into the direct sum of vector spaces and $e_Q$ is analogous.
\end{lemma}
\begin{proof}
    We start from \Cref{eq:colim of duals def of G(M)} and use the definition of a colimit to compute
    \begin{align*}
        \cosheaf(M)&\cong\colim_{P\in\maxplanesclosed(M)} \cC_P^*\\
        &=\colim_{P\in\maxplanesclosed(M)} U_P\\%
        &\cong U_M/R_M,
    \end{align*}
    where the last step is just the definition of the colimit on vector spaces.
\end{proof}

Later on we will need to treat the entries of the generator matrices as free parameters and allow arbitrary matrices $G_i$, including ones of deficient row rank that do not generate a $k_i$-dimensional code. Since \Cref{eq:cores explicit def} only inserts columns of the $G_i$, it continues to make sense in this generality.
\begin{definition}[Polynomially parametrized insertion maps]\label{def:polynomial insertion maps}
    Let $\bX=(X^{(i)}_{r,j})_{i\in[D],r\in[k_i],j\in[n]}$ be independent indeterminates, let $X_i$ be the $k_i\times n$ matrix over $\Z[\bX]$ with entries $X^{(i)}_{r,j}$, and write $x_i(j)$ for its $j$-th column. We take $U_i=\Z[\bX]^{k_i}$ and $U_P=U_I=\bigotimes_{i\in I}U_i$ for a plane $P=\plane(I,a)$, with $U_\varnothing=\Z[\bX]$ for a point. 
    For planes $P=\plane(I,a)\subseteq Q=\plane(J,b)$, the corestriction map $\cores_{P\to Q}^{\bX}$ is a $\Z[\bX]$-linear map $U_P\to U_Q$ defined by the same formula as in \Cref{eq:cores explicit def} but using $X_i$ in place of $G_i$.
    
    A \emph{specialization} of these maps substitutes matrices $G_i\in\gfield^{k_i\times n}$ over a field $\gfield$ for the $X_i$, and recovers the corestriction maps of \Cref{eq:cores explicit def} for those $G_i$.
\end{definition}
No rank requirement is involved here, and compatibility is immediate from the formula: for $P\subseteq Q\subseteq R$, both $\cores_{Q\to R}^{\bX}\circ\cores_{P\to Q}^{\bX}$ and $\cores_{P\to R}^{\bX}$ insert the columns $x_m(a_m)$ for the coordinates $m$ freed between $P$ and $R$, since the anchors of $P\subseteq Q$ agree on the coordinates that $Q$ fixes. This is a polynomial identity, so it survives every specialization.

Taking these maps with the indeterminates $\bX$ over all planes of $M$ gives the diagram we will actually compute with later on. 
\begin{definition}\label{def:formal insertion diagram}
    For $M\subseteq[n]^D$, keep $U_M=\bigoplus_{P\in\maxplanesclosed(M)}U_P$ from \Cref{lem:cosheaf(M) as quotient of direct sum} and set
    \begin{align*}
        \cE_M=\bigoplus_{\substack{P,Q\in\maxplanesclosed(M)\\P\subsetneq Q}}U_P,
    \end{align*}
    and let $e_{P,Q}:U_P\to\cE_M$ and $e_P:U_P\to U_M$ denote the inclusions of the corresponding summands. The maps
    \begin{align}
        \alpha_M^{\bX}:\cE_M&\to U_M &  \beta_M^{\bX}:U_M&\to U_{[D]}\notag\\
        e_{P,Q}(u)&\mapsto e_P(u)-e_Q(\cores_{P\to Q}^{\bX}(u)) & (u_P)_P&\mapsto\sum_{P\in\maxplanesclosed(M)}\cores_{P\to[n]^D}^{\bX}(u_P),\label{eq:formal alpha and beta}
    \end{align}
    are defined on generators and extend linearly. Compatibility gives $\beta_M^{\bX}\circ\alpha_M^{\bX}=0$, so $\beta_M^{\bX}$ descends to the quotient map
    \begin{align*}
        \overline\beta_M^{\bX}:\Coker\alpha_M^{\bX}\to U_{[D]}.
    \end{align*}
\end{definition}
The $P=Q$ terms in the spanning set of $R_M$ contribute $e_P(u)-e_P(u)=0$, so $\Im\alpha_M=R_M$. Hence whenever every $G_i$ has full row rank, \Cref{lem:cosheaf(M) as quotient of direct sum} gives $\Coker\alpha_M\cong\cosheaf(M)$ and $\overline\beta_M=\cores_M$. For arbitrary $G_i$ there is no such interpretation in terms of codes, but the diagram and its cokernel are still well defined.

The specializations of this diagram that we will need are the ones witnessing injectivity over $\R$.
\begin{definition}[Real sign witness]\label{def:real sign witness}
    Fix $M\subseteq[n]^D$ and dimensions $(k_i)_{i\in[D]}$. A \emph{real sign witness} is a tuple of sign matrices $\sigma=(\sigma^{(i)})_{i\in[D]}$ with $\sigma^{(i)}\in\{-1,1\}^{k_i\times n}$ at whose specialization over $\gfield=\R$ the map $\overline\beta_M^{\bX}$ of \Cref{def:formal insertion diagram} is injective. We emphasize that the $\sigma^{(i)}$ are \emph{not} required to have full row rank.
\end{definition}

\subsection{Atom decomposition}\label{subsec:atoms-atom decomposition}

Using these definitions so far, we give an even more precise decomposition of $\cosheaf(M)$ into what we call ``atoms''. For the remainder of this section, we will take our base field $\gfield=\R$ and assign the standard inner product to all $U_i$. 

\begin{definition}\label{def:atom}
    For $P\in\maxplanesclosed(M)$, define the vector space
    \begin{align*}
        U_{\subsetneq P}=\sum_{\substack{P'\in\maxplanesclosed(M)\\P'\subsetneq P}}\cores_{P'\to P}(U_{P'})\subseteq U_P.
    \end{align*}
    The \emph{atom} $A_P$ of $P$ is then defined as the orthogonal complement
    \begin{align*}
        A_P=U_{\subsetneq P}^\perp\subseteq U_P,
    \end{align*}
    which is naturally isomorphic to the quotient
    \begin{align*}
        A_P\cong U_P/U_{\subsetneq P}.
    \end{align*}
\end{definition}

Unlike in \Cref{lem:cosheaf(M) as quotient of direct sum} where we describe $\cosheaf(M)$ exactly as the quotient of a direct sum, the direct sum of these atoms covers $\cosheaf(M)$; that is, there is a natural surjection from this direct sum onto $\cosheaf(M)$. 
\begin{definition}
    For a plane $P\in\maxplanesclosed(M)$, let $\overline{e}_P:U_{P}\to\left(\bigoplus_{P'\in\maxplanesclosed(M)}U_{P'}\right)/R_M$ be the natural inclusion $e_{P}$ composed with quotienting by $R_M$.
\end{definition}

\begin{lemma}
\label{lem:direct sum of atoms spans cosheaf}
Let $M\subseteq[n]^D$.
Then the natural map
\begin{align*}
    \Psi:\bigoplus_{P\in\maxplanesclosed(M)}A_P
    &\to \cosheaf(M)\\
    (a_P)_{P\in\maxplanesclosed(M)} &\mapsto\sum_{P\in\maxplanesclosed(M)}\overline{e}_P(a_P)
\end{align*}
is always surjective.
\end{lemma}

\begin{proof}
If $M=\varnothing$, then $\Psi:0\to0$ is surjective, so assume $M\ne\varnothing$. Point planes are also included: for $P=\{x\}$, we have $U_P=\R$ and $U_{\subsetneq P}=0$, so $A_P=\R$.
Ordering $\maxplanesclosed(M)$ by inclusion gives a finite poset, so we can find a labeling $\maxplanesclosed(M)=\{P_1,\dots,P_s\}$ such that $P_a\subsetneq P_b$ implies $a<b$. Define the subset $\planes_i(M)=\{P_1,\dots,P_i\}$ and the colimit of the subdiagram using these planes
    \begin{align*}
        \cosheaf_i(M)=\colim_{P\in\planes_i(M)} U_P,
    \end{align*}
    so $\cosheaf_s(M)=\cosheaf(M)$. 
We will now induct on $i\in[s]$ by showing that, at each step, 
\begin{align*}
    \Psi_i:\bigoplus_{j=1}^i A_{P_j}&\to \cosheaf_i(M)\\
    (a_j)_{j\in[i]}&\mapsto \sum_{j=1}^i \overline e_{P_j}^{(i)}(a_j)
\end{align*}
is surjective, where $\overline e_{P_j}^{(i)}:U_{P_j}\to\cosheaf_i(M)$ denotes the canonical map from $U_{P_j}$ into the partial colimit.
By definition, $A_{P_i}=U_{\subsetneq P_i}^\perp$, 
so over $\R$ we have
\begin{align}
    U_{P_i}=U_{\subsetneq P_i}\oplus A_{P_i}.\label{eq:VPi decomposed}
\end{align}

For the base case of $i=1$, there are no proper subplanes of $P_1$ by our labeling, so $U_{\subsetneq P_1}=0\subseteq U_{P_1}$ and 
    \begin{align*}
        \cosheaf_1(M)&=U_{P_1}\\
        &=U_{\subsetneq P_1}\oplus A_{P_1}\\
        &=0\oplus A_{P_1}\\
        &=A_{P_1},
    \end{align*}
    as required.

For $i>1$, our inductive hypothesis is that $\Psi_{i-1}$ is surjective.
Let $\eta_i:\cosheaf_{i-1}(M)\to \cosheaf_i(M)$ be the canonical map induced by the inclusion of diagrams $\planes_{i-1}(M)\subseteq\planes_i(M)$.
For every $j<i$, the canonical maps satisfy
\begin{align*}
    \overline e_{P_j}^{(i)}=\eta_i\circ \overline e_{P_j}^{(i-1)}.
\end{align*}
Therefore, by our inductive hypothesis,
\begin{align*}
    \sum_{j=1}^{i-1}\overline e_{P_j}^{(i)}(U_{P_j})&=\eta_i\left(\sum_{j=1}^{i-1}\overline e_{P_j}^{(i-1)}(U_{P_j})\right)\\
    &\subseteq\eta_i(\Im(\Psi_{i-1}))\\
    &\subseteq\Im(\Psi_i).\numberthis\label{eq:inductive hypothesis on atoms expanded}
\end{align*}
Now it remains to analyze the image of $U_{P_i}$ in $\cosheaf_i(M)$.
We first show that the image of $U_{\subsetneq P_i}$ in $\cosheaf_i(M)$ is
contained in $\Im(\Psi_i)$. 
Take any $u\in U_{\subsetneq P_i}$. By the definition of
$U_{\subsetneq P_i}$, we can write
\begin{align*}
    u=\sum_{\substack{P'\in\maxplanesclosed(M)\\P'\subsetneq P_i}}
    \cores_{P'\to P_i}(u_{P'})
\end{align*}
for some vectors $u_{P'}\in U_{P'}$. In the colimit
$\cosheaf_i(M)$, the colimit relations (from $R_M$) give
\begin{align*}
    \overline e_{P_i}^{(i)}\bigl(\cores_{P'\to P_i}(u_{P'})\bigr)=\overline e_{P'}^{(i)}(u_{P'}).
\end{align*}
Therefore, using \Cref{eq:inductive hypothesis on atoms expanded} we get
\begin{align*}
    \overline e_{P_i}^{(i)}(u)&=\sum_{\substack{P'\in\maxplanesclosed(M)\\P'\subsetneq P_i}}\overline e_{P'}^{(i)}(u_{P'})\\
    &\in\Im(\Psi_i),
\end{align*}
showing that $\overline{e}_{P_i}^{(i)}(U_{\subsetneq P_i})\subseteq\operatorname{im}(\Psi_i)$. The surjectivity of $\Psi_i$ now follows from \Cref{eq:VPi decomposed} and the fact that, by definition, $\overline{e}_{P_i}^{(i)}(A_{P_i})\subseteq\Im(\Psi_i)$. Finally, taking $i=s$ gives us that $\Psi$ is surjective. This proof uses only the insertion maps and colimit relations, so it also applies to the formal cokernel of \Cref{def:formal insertion diagram} after any real specialization.

\end{proof}

\dobib

\section{Two-Sided Product Expansion Over Random Codes}\label{sec:two-sided-product-expansion}
In this section, we prove our main result in \Cref{thm:main two sided prod expanding} by reducing it to a single statement: that over one prime field, the corestriction maps of a random tuple of codes and of its dual are simultaneously injective for every \ec set. That statement is \Cref{thm:cores_M injective whp}, which we prove in the next section. 

Before stating the theorem, we fix the parameters at which the real input will be applied. Let $0<\rho<1/2$, and define $C_\rho=3\left(1+\rho^{-1/2}\right)$ and $A_{D,\rho}=2(4C_\rho)^D\left(1+4/\rho\right)^D$ so that we can set
\begin{align}
    \varepsilon_0(D,\rho)&=\frac{1}{2^{D+2}A_{D,\rho}^2},&
    N^*_{D,\varepsilon}&=\left\lceil\max\left\{16\left(\frac{192D}{e}\right)^{24},(2^D\varepsilon)^{-2}\right\}\right\rceil.\label{eq:real witness parameters}
\end{align}
We will see later in \Cref{sec:injectivity of cores via K} that these parameters are sufficient to guarantee the existence of a real witness for every \ec set, which is the key to proving injectivity of the corestriction maps.

With the parameters fixed, our main result follows.
\begin{restatable}{theorem}{mainTwoSidedProd}\label{thm:main two sided prod expanding}
    Let $D\geq1$, let $0<\rho<1/2$, let $0<\varepsilon\leq\varepsilon_0(D,\rho)$, let $n\geq N^*_{D,\varepsilon}$, and let $k_1,\dots,k_D$ be integers with $\rho n\leq k_i\leq(1-\rho)n$. Then there is a prime
    \begin{align*}
        p=O_D\left(2^{n^D}n^D\log n\right)
    \end{align*}
    such that independent uniformly random codes $\cC_i\in\grass_p(n,k_i)$ are two-sided $\gamma(\varepsilon,D)$-product expanding with probability at least
    \begin{align*}
        1-\frac{1}{\log(n+1)},
        \qquad\text{where}\qquad
        \gamma(\varepsilon,D)=\frac{\varepsilon^D}{D\left(2^D+1\right)^D}.
    \end{align*}
\end{restatable}

The one ingredient we need besides \Cref{thm:cores_M injective whp} is the implication from inner-generatedness to product expansion, which we take from \cite{kalachev_maximally_2025}.
\begin{lemma}[Lemma 19 of \cite{kalachev_maximally_2025}]\label{lem:KP25 lem 19}
    Let $\cC_1,\dots,\cC_D\subsetneq\gfield^n$ be codes over any field, and let $0<\varepsilon\leq1$. If every \ec subset of $[n]^D$ is inner-generated in $\bigboxplus_i\cC_i$, then $(\cC_i)_{i\in[D]}$ is $\gamma(\varepsilon,D)$-product expanding.
\end{lemma}
Although \cite{kalachev_maximally_2025} works in characteristic two throughout, the proof of their Lemma 19 never uses the characteristic of the base field and immediately generalizes to any finite field.

\begin{proof}[Proof of \Cref{thm:main two sided prod expanding}]
    Let $p$ be the prime given by \Cref{thm:cores_M injective whp} and apply that theorem with $t=1$. Then with probability at least $1-1/\log(n+1)$, every \ec set $M$ has injective corestriction maps for both $(\cC_i)_{i\in[D]}$ and $(\cC_i^\perp)_{i\in[D]}$. 

    Conditioned on this event, \Cref{claim:inner generated iff cores injective} to $(\cC_i)_{i\in[D]}$, whose dual codes are the $\cC_i^\perp$, makes every \ec set inner-generated in $\bigboxplus_i\cC_i^\perp$. Applying the same fact to $(\cC_i^\perp)_{i\in[D]}$ makes every \ec set inner-generated in $\bigboxplus_i\cC_i$. So \Cref{lem:KP25 lem 19} applies to each of the two tuples in turn and gives $\gamma(\varepsilon,D)$-product expansion of both, which is exactly two-sided expansion. 
\end{proof}

\dobib

\section{Injectivity of the Corestriction Map Via the Plane-Interaction Matrix}\label{sec:injectivity of cores via K}
As we saw in the previous section, our remaining goal is to show the simultaneous injectivity of $\cores_M$ over a single prime field for all \ec sets $M$ and both $(\cC_i)_{i\in[D]}$ and $(\cC_i^\perp)_{i\in[D]}$. In particular, we show the following.

\begin{theorem}\label{thm:cores_M injective whp}
    Let $D\geq1$, $0<\rho<1/2$, $0<\varepsilon\leq\varepsilon_0(D,\rho)$, and $n\geq N^*_{D,\varepsilon}$, with the parameters from \Cref{eq:real witness parameters}, and let $k_1,\dots,k_D$ be integers with $\rho n\leq k_i\leq(1-\rho)n$. Write $k_i^+=k_i$ and $k_i^-=n-k_i$ for the dimensions of $(\cC_i)_{i\in[D]}$ and $(\cC_i^\perp)_{i\in[D]}$, respectively. Then there is a prime
    \begin{align*}
        p=O_D\left(2^{n^D}n^D\log n\right)
    \end{align*}
    such that the following holds: for every integer $t\geq1$, if $\cC_i\sim\grasspt(n,k_i)$ independently for $i\in[D]$, then, with probability at least
    \begin{align*}
        1-\frac{p^{1-t}}{\log(n+1)},
    \end{align*}
    the corestriction map $\cores_M$ is injective for every \ec set $M\subseteq[n]^D$, simultaneously for $(\cC_i)_{i\in[D]}$ and $(\cC_i^\perp)_{i\in[D]}$.
\end{theorem}
Recall from \Cref{def:real sign witness} that a witness is not required to have full row rank. Our analysis of the \emph{plane-interaction matrix} over the reals is what we use to show the existence of such witnesses. The rest of our argument involves going back to a finite field from the reals, which we give at the end of this section. We begin by defining the random model over the reals that we will use.

\paragraph{Setup} 
Our random model over $\R$ samples matrices with independent Rademacher entries and then normalizes the columns to have norm one.
\begin{definition}[Random model]
    For fixed integers $1\leq k_i\leq n$ for $i\in[D]$, take $G_i\in\R^{k_i\times n}$ to have columns $(g_i(j))_{j\in[n]}\in\R^{k_i}$ where
    \begin{align*}
        g_i(j)=k_i^{-1/2}(\sigma^{(i)}_{1,j},\sigma^{(i)}_{2,j},\dots,\sigma^{(i)}_{k_i,j})\in\R^{k_i}
    \end{align*}
    and each $\sigma^{(i)}_{r,j}\sim\{-1,1\}$ is sampled identically and independently. By construction, $\norm{g_i(j)}=1$. The row-space code $\cC_i$ has dimension $\rank G_i$, which may be smaller than $k_i$.
\end{definition}

We set the formal variables $\bX$ of \Cref{def:formal insertion diagram} to $G=(G_i)_{i\in[D]}$, keeping $U_i=\R^{k_i}$ and $U_\varnothing=\R$ regardless of the ranks of the matrices. For the real argument up to \Cref{subsec:injectivity cores-reals to finite fields}, we use the notation
\begin{align*}
    \cosheaf(M):=\Coker\alpha_M^G,
    \qquad
    \cores_M:=\overline\beta_M^G:\cosheaf(M)\to U_{[D]}.
\end{align*}
These agree with the cosheaf and corestriction map of the row-space codes when every $G_i$ has full row rank. We do not condition the signs on this event. The atoms are defined from these insertion maps as in \Cref{def:atom}. 

With this setup, our first ingredient is the following bound on the probability that $\cores_M$ is injective over the reals.

\begin{restatable}{lemma}{coresMinjectivewhpoverR}\label{lem:cores_M injective whp over reals}
    For constants $D\in\N$, $\rho\in(0,1)$, and $\varepsilon\in(0,2^{-D})$, let
    \[
        n\geq N_{D,\varepsilon}:=\max\left\{16\left(\frac{192D}{e}\right)^{24}, (2^D\varepsilon)^{-2}\right\}.
    \]
    Fix an \ec set $M\subseteq[n]^D$, and let $\rho n\leq k_i\leq n$ for every $i\in[D]$. Then
    \begin{align*}
        \Pr_\sigma[\cores_M\text{ injective}]\geq 1- C_{D,\rho,\varepsilon},
    \end{align*}
    where
    \begin{align*}
        C_{D,\rho,\varepsilon}&=2\left(4 C_\rho\right)^D\left(1+\frac{4}{\rho}\right)^D\sqrt{2^D\varepsilon} \qquad\text{and}\qquad C_\rho:=3\left(1+\rho^{-1/2}\right).
    \end{align*}
\end{restatable}
For the parameters in \Cref{eq:real witness parameters}, this lower bound is at least $1/2$.

Our second ingredient transfers this result over the reals to a finite field. We first state the transfer for a single fixed set $M$, where the resulting prime is allowed to depend on $M$. Then, the proof of \Cref{thm:cores_M injective whp} at the end of this section upgrades it to one prime for every \ec set and both $(\cC_i)_{i\in[D]}$ and $(\cC_i^\perp)_{i\in[D]}$ simultaneously.

\begin{restatable}{lemma}{coresMinjectiveoverrealsimpliesoverF}\label{lem:cores_M injective over reals gives injective over finite field}
    Fix constants $D\in\N$ and $\varepsilon,\rho\in(0,1)$, let $k_i\geq\rho n$ for every $i\in[D]$, and let $\emptyset\neq M\subsetneq[n]^D$ be \ec. Suppose that there is a real sign witness for $(M,(k_i)_{i\in[D]})$. Then, for all sufficiently large $n$, there is a prime
    \begin{align*}
        p=O_{D,\varepsilon}(n^D\log n)
    \end{align*}
    such that, for independent uniformly random $k_i$-dimensional codes $\cC_i\subseteq\F_p^n$,
    \begin{align*}
        \Pr[\cores_M\text{ is injective}]\geq 1-\frac{2}{\log(n+1)}.
    \end{align*}
\end{restatable}

To obtain the real sign witnesses for \Cref{thm:cores_M injective whp}, we introduce the plane-interaction matrix $K$ which will be the main tool we use to prove the injectivity of $\cores_M$ over $\R$. In particular, the construction of the plane-interaction matrix will allow us to make explicit the obstructions to the injectivity of $\cores_M$.

In \Cref{subsec:injectivity cores-plane interaction observations} we observe how the corestriction maps of intersecting planes interact to motivate our derivation of the definition of the plane-interaction matrix over $\R$ in \Cref{subsec:injectivity cores-constructing K}. Then, in \Cref{subsec:injectivity cores-operator norm < 1 implies cores injective}, we prove \Cref{lem:cores_M injective whp over reals} by showing that if the operator norm of $K$ is less than one, then $\cores_M$ is injective. Finally, in \Cref{subsec:injectivity cores-reals to finite fields} we prove the fixed-set transfer in \Cref{lem:cores_M injective over reals gives injective over finite field} and the simultaneous statement in \Cref{thm:cores_M injective whp}.

\subsection{Interactions between the corestriction maps of planes}\label{subsec:injectivity cores-plane interaction observations}

We begin by analyzing the interactions between different planes under our random Rademacher model to motivate the definition of the plane-interaction matrix. We first make a quick notational definition for succinctness.
\begin{definition}\label{def:random rademacher model}
    For a plane $P\in\maxplanesclosed(M)$, we write its global corestriction map $\cores_{P\to[n]^D}: U_P\to U_{[D]}$ as simply $\cores_P$.
\end{definition}

Now we show that the inner product of the global corestriction maps of incident planes vanishes when restricted to atoms.
\begin{lemma}\label{lem:cores of elts of intersecting atoms are orthogonal}
    For distinct intersecting $P,Q\in\maxplanesclosed(M)$, we have that for all $x\in A_P$ and $y\in A_Q$:
    \begin{align*}
        \left\langle\cores_P x,\cores_Q y\right\rangle = 0.
    \end{align*}
\end{lemma}
The main technical tool required for this is the fact that we can work with the inner product of the adjoint of the corestriction maps instead of the corestriction maps themselves.

\begin{lemma}\label{lem:inner prod of cores is inner prod of adjoints to intersection}
    For two intersecting planes $P=\plane(I,a),Q=\plane(J,b)\in\maxplanesclosed(M)$, let $K=I\cap J$ and define the plane at their intersection $R=\plane(K,c)$ where
    \begin{align*}
        c_i=
        \begin{cases}
            a_i=b_i & i\in \setcomp{(I\cup J)}\\
            a_i & i\in J\setminus I\\
            b_i & i\in I\setminus J.
        \end{cases}
    \end{align*}
    Then we have that for all $x\in U_P$ and $y\in U_Q$
    \begin{align*}
        \langle\cores_P x,\cores_Q y\rangle=\langle\gamma_{P\to R} x,\gamma_{Q\to R}y\rangle
    \end{align*}
    where $\gamma_{P\to R}$ and $\gamma_{Q\to R}$ are the adjoints of $\cores_{R\to P}$ and $\cores_{R\to Q}$, respectively, given on pure tensors by\footnote{We use the standard Euclidean inner products on the real spaces $U_P$, $U_Q$, and $U_R$.}
    \begin{align*}
        \gamma_{P\to R}=\cores_{R\to P}^\adjoint : U_P&\to U_R &\gamma_{Q\to R}=\cores_{R\to Q}^\adjoint:U_Q&\to U_R\\
        x&\mapsto \left(\prod_{i\in I\setminus J}\langle x_i,g_i(c_i)\rangle\right)\bigotimes_{k\in K} x_k & y &\mapsto \left(\prod_{j\in J\setminus I}\langle y_j,g_j(c_j)\rangle\right)\bigotimes_{k\in K} y_k.
    \end{align*}
\end{lemma}

We now prove \Cref{lem:cores of elts of intersecting atoms are orthogonal}.

\begin{proof}[Proof of \Cref{lem:cores of elts of intersecting atoms are orthogonal}]
    Let $R=P\cap Q=\plane(I\cap J,c)$, using the same notation as in \Cref{lem:inner prod of cores is inner prod of adjoints to intersection}. Using the assumption on $P$ and $Q$ that they have non-empty intersection, at least one of $R\subsetneq P$ or $R\subsetneq Q$ must hold.

    \paragraph{Case $R\subsetneq P$:}
    In this case, by \Cref{def:atom} we have that $\cores_{R\to P}(U_R)\subseteq U_{\subsetneq P}$. Since $A_P=U_{\subsetneq P}^\perp$, we have that for any $x\in A_P$ and $z\in U_R$
    \begin{align*}
        0&=\langle x,\cores_{R\to P} z\rangle\\
        &=\langle\gamma_{P\to R} x,z\rangle,
    \end{align*}
    by the definition of adjoints. However, since this holds for any $z\in U_R$, it must be that $\gamma_{P\to R} x=0$. Consequently, using \Cref{lem:inner prod of cores is inner prod of adjoints to intersection} we get
    \begin{align*}
        \langle\cores_P x,\cores_Q y\rangle &= \langle \gamma_{P\to R} x,\gamma_{Q\to R} y\rangle\\
        &=\langle 0,\gamma_{Q\to R} y\rangle\\
        &=0,
    \end{align*}
    as claimed.

    \paragraph{Case $R\subsetneq Q$:}
    Symmetrically in this case we can derive that $\gamma_{Q\to R} y=0$, so 
    \begin{align*}
        \langle\cores_P x,\cores_Q y\rangle &= \langle \gamma_{P\to R} x,\gamma_{Q\to R} y\rangle\\
        &=\langle \gamma_{P\to R} x,0\rangle\\
        &=0,
    \end{align*}
    as required.
\end{proof}

On the other hand, if instead of distinct planes $P$ and $Q$ we have that $P=Q$, then the corestriction map is an isometry when restricted to atoms.

\begin{lemma}\label{cor:cores is isometry on atoms}
    If instead $P=Q$, then for all $x,y\in A_P$ we have that
    \begin{align*}
        \langle\tau_P x,\tau_P y\rangle_{U_{[D]}}=\langle x,y\rangle_{U_P}.
    \end{align*}
\end{lemma}
\begin{proof}
    We directly apply \Cref{lem:inner prod of cores is inner prod of adjoints to intersection} to compute
    \begin{align*}
        \langle\tau_P x,\tau_P y\rangle_{U_{[D]}}&=\langle\gamma_{P\to P} x,\gamma_{P\to P}y\rangle_{U_P}\\
        &=\langle x,y\rangle_{U_P}.
    \end{align*}
\end{proof}

We end by providing the proof of \Cref{lem:inner prod of cores is inner prod of adjoints to intersection}.
\begin{proof}[Proof of \Cref{lem:inner prod of cores is inner prod of adjoints to intersection}]
    To examine the result of the desired inner product $\left\langle\cores_P x,\cores_Q y\right\rangle$, we begin by taking pure tensors 
    \begin{align*}
        x = \bigotimes_{i\in I} x_i\in U_P=U_I \qquad y=\bigotimes_{j\in J}y_j\in U_Q=U_J
    \end{align*}
    and analyzing $\cores_P x$ and $\cores_Q y$. Recalling \Cref{eq:cores explicit def}, we get that
    \begin{align*}
        \cores_P x= x\otimes\left(\bigotimes_{i\in\setcomp{I}}g_i(a_i)\right) \qquad
        \cores_Q y= y\otimes\left(\bigotimes_{j\in\setcomp{J}}g_j(b_j)\right).
    \end{align*}

    This allows us to compute their inner product explicitly as
    \begin{align*}
        \left\langle\cores_P x,\cores_Q y\right\rangle&=\left\langle x\otimes\left(\bigotimes_{i\in\setcomp{I}}g_i(a_i)\right),y\otimes\left(\bigotimes_{j\in\setcomp{J}}g_j(b_j)\right)\right\rangle\\
        &=\left(\prod_{k\in K}\langle x_k,y_k\rangle\right)\left(\prod_{i\in I\setminus J} \langle x_i,g_i(c_i)\rangle\right)\left(\prod_{j\in J\setminus I}\langle g_j(c_j),y_j\rangle \right)\left(\prod_{i\in\setcomp{(I\cup J)}}\langle g_i(c_i),g_i(c_i)\rangle\right)\\
        &=\left(\prod_{k\in K}\langle x_k,y_k\rangle\right)\left(\prod_{i\in I\setminus J} \langle x_i,g_i(c_i)\rangle\right)\left(\prod_{j\in J\setminus I}\langle g_j(c_j),y_j\rangle \right)\left(\prod_{i\in\setcomp{(I\cup J)}} \norm{g_i(c_i)}^2\right)\\
        &=\left(\prod_{k\in K}\langle x_k,y_k\rangle\right)\left(\prod_{i\in I\setminus J} \langle x_i,g_i(c_i)\rangle\right)\left(\prod_{j\in J\setminus I}\langle g_j(c_j),y_j\rangle \right)\numberthis\label{eq:cores prod expansion used norm is 1}\\
        &=\left\langle \left(\prod_{i\in I\setminus J}\langle x_i,g_i(c_i)\rangle\right)\bigotimes_{k\in K} x_k, \left(\prod_{j\in J\setminus I}\langle y_j,g_j(c_j)\rangle\right)\bigotimes_{k\in K} y_k\right\rangle\\%
        &=\langle\gamma_{P\to R} x,\gamma_{Q\to R} y\rangle,\numberthis\label{eq:equality of inner prod of cores and adjoints}
    \end{align*}
    where to get \Cref{eq:cores prod expansion used norm is 1} we used the fact that the columns are normalized as $\norm{g_i(c_i)}=1$. Moreover, we have that \Cref{eq:equality of inner prod of cores and adjoints} holds for any $x\in U_P$ and $y\in U_Q$ because pure tensors span $U_P$ and $U_Q$.
\end{proof}

\subsection{Constructing the plane-interaction matrix}\label{subsec:injectivity cores-constructing K}
With these observations of the interactions between the corestriction maps of intersecting planes, we move on to defining the plane-interaction matrix $K$. We begin by introducing some useful notation.
\begin{definition}
    Recall that for $P\in\maxplanesclosed(M)$ we have $A_P:=U_{\subsetneq P}^\perp\subseteq U_P$, and let $\Pi_P:U_P\to A_P$ be the orthogonal projection map. Define
    \begin{align*}
        U=U_M&=\bigoplus_{P\in\maxplanesclosed(M)} U_P\\
        A&=\bigoplus_{P\in\maxplanesclosed(M)} A_P\\
        \Pi_A&=\bigoplus_{P\in\maxplanesclosed(M)} \Pi_P:U\to A.
    \end{align*}
    We drop the subscript $M$ when it is clear from context.
\end{definition}

Next, we examine all interactions between the corestriction maps.
\begin{definition}
    For $P\in\maxplanesclosed(M)$, define the map $\theta_P:=\cores_P|_{A_P}:A_P\to U_{[D]}$ and the global operator
    \begin{align*}
        \Theta:A&\to U_{[D]}\\
        (a_P)_{P\in\maxplanesclosed(M)}&\mapsto \sum_{P\in\maxplanesclosed(M)}\theta_P a_P.
    \end{align*}
    Moreover, we define its Gram operator
    \begin{align*}
        \Upsilon:=\Theta^{\adjoint}\Theta:A&\to A.
    \end{align*}
\end{definition}
We can immediately observe that for $P,Q\in\maxplanesclosed(M)$, the $(P,Q)$-th block of $\Upsilon$ is $\Upsilon_{P,Q}=\theta_P^\adjoint\theta_Q:A_Q\to A_P$. Consequently, $\Upsilon_{P,P}=\theta_P^\adjoint \theta_P= \Id_{A_P}$ by \Cref{cor:cores is isometry on atoms} and $\Upsilon_{P,Q}=0$ when $P\neq Q$, and $P\cap Q\neq\varnothing $ by \Cref{lem:cores of elts of intersecting atoms are orthogonal}. Since the diagonal of $\Upsilon$ is only the identity, only the off-diagonal blocks are interesting. This motivates the following definition.
\begin{definition}[Plane-interaction matrix]\label{def:comparison matrix}
    Let $H:=\Upsilon-\Id_A:A\to A$, so by \Cref{lem:cores of elts of intersecting atoms are orthogonal} and \Cref{cor:cores is isometry on atoms} for $P,Q\in\maxplanesclosed(M)$ we have
    \begin{align*}
        H_{P,Q}=
        \begin{cases}
            0 & P\cap Q\neq\varnothing\\
            \theta^\adjoint_P\theta_Q & P\cap Q=\varnothing.
        \end{cases}
    \end{align*}
    Similarly, on the entire space $U\supseteq A$, we will define the \emph{plane-interaction matrix} $K:U\to U$ via blocks as
    \begin{align*}
        K_{P,Q}=
        \begin{cases}
            0 & P\cap Q\neq\varnothing\\\
            \cores_P^\adjoint\cores_Q & P\cap Q=\varnothing.
        \end{cases}
    \end{align*}
\end{definition}

\subsection{Injectivity of \texorpdfstring{$\cores_M$}{the corestriction map} over the reals}\label{subsec:injectivity cores-operator norm < 1 implies cores injective}
In this subsection, we will prove that $\cores_M$ is injective with probability at least $1-C_{D,\rho,\varepsilon}$ under our real Rademacher model.
\coresMinjectivewhpoverR*

This follows from two facts. First, that if $K$ has operator norm less than one, then $\cores_M$ is injective.
\begin{restatable}{lemma}{Kopnormsmallimpliescoresinjective}\label{thm:K or H have op norm < 1 implies cores injective}
    If $\opnorm{K}<1$, then $\cores_M$ is injective.
\end{restatable}

Second, that the operator norm of $K$ is less than one with probability at least $1-C_{D,\rho,\varepsilon}$. Proving this theorem is the main focus of \Cref{sec:expected op norm of K}.

\begin{restatable}{theorem}{probOpNormKlessonelarge}\label{thm:prob op norm K < 1 large}
For constants $D\in\N$, $\rho\in(0,1)$, and $\varepsilon\in(0,2^{-D})$, let
    \[
        n\geq N_{D,\varepsilon}:=\max\left\{16\left(\frac{192D}{e}\right)^{24}, (2^D\varepsilon)^{-2}\right\}.
    \]
    Fix an \ec set $M\subseteq[n]^D$, and let $\rho n\leq k_i\leq n$ for every $i\in[D]$. Then
    \begin{align*}
        \Pr_\sigma[\opnorm{K}<1]\geq 1- C_{D,\rho,\varepsilon},
    \end{align*}
    where
    \begin{align*}
        C_{D,\rho,\varepsilon}&=2\left(4 C_\rho\right)^D\left(1+\frac{4}{\rho}\right)^D\sqrt{2^D\varepsilon} \qquad\text{and}\qquad C_\rho:=3\left(1+\rho^{-1/2}\right).
    \end{align*}
\end{restatable}

The proof of \Cref{lem:cores_M injective whp over reals} now follows directly.
\begin{proof}[Proof of \Cref{lem:cores_M injective whp over reals}]
    Apply \Cref{thm:K or H have op norm < 1 implies cores injective} to \Cref{thm:prob op norm K < 1 large}.
\end{proof}

\subsubsection{Injectivity of \texorpdfstring{$\cores_M$}{the corestriction map} from the operator norm of \texorpdfstring{$K$}{K}}
The rest of this subsection is dedicated to proving \Cref{thm:K or H have op norm < 1 implies cores injective}. The only tool we need is that the operator norm of $K$ bounds the operator norm of $H$.

\begin{claim}\label{claim:opnorm of K geq H}
    We have $\opnorm{H}\leq\opnorm{K}$.
\end{claim}
\begin{proof}
    Let $\iota_A:A\hookrightarrow U$ be the inclusion. By construction, we have that $H=\Pi_A K \iota_A$, so we can compute
    \begin{align*}
        \opnorm{H}&=\opnorm{\Pi_A K \iota_A}\\
        &\leq\opnorm{\Pi_A}\opnorm{K}\opnorm{\iota_A}\\
        &\leq\opnorm{K},
    \end{align*}
    since both $\Pi_A$ and $\iota_A$ have operator norm at most $1$.
\end{proof}

We can now condition the injectivity of $\cores_M$ on the operator norm of $H$ and $K$.
\begin{proof}[Proof of \Cref{thm:K or H have op norm < 1 implies cores injective}]
    By \Cref{claim:opnorm of K geq H}, we have that $\opnorm{H}\leq\opnorm{K}<1$. Since $H$ is self-adjoint, for any $a\in A$ we can lower bound its quadratic form as
    \begin{align}
        \langle a,Ha\rangle\geq-\opnorm{H}\norm{a}^2.\label{eq:quadratic form of H LB by opnorm}
    \end{align}
    Therefore, $\Theta$ is injective because $1-\opnorm{H}>0$ and 
    \begin{align*}
        \norm{\Theta a}^2&=\langle \Theta a,\Theta a\rangle\\
        &=\langle a,\Upsilon a\rangle\\
        &=\langle a,(\Id_A+H)a\rangle\\
        &=\norm{a}^2+\langle a,Ha\rangle\\
        &\geq\norm{a}^2-\opnorm{H}\norm{a}^2 &\text{By \Cref{eq:quadratic form of H LB by opnorm}}\\
        &=(1-\opnorm{H})\norm{a}^2.
    \end{align*}
    To conclude that $\cores_M$ is injective, we consider any $z\in\Ker(\cores_M)\subseteq\cosheaf(M)$. The surjectivity of $\Psi$ gives an $a\in A$ such that $\Psi(a)=z$. Finally, compatibility of the maps in the colimit defining $\cosheaf(M)$ gives 
    \begin{align*}
        \Theta(a)&=\cores_M(\Psi(a))\\
        &=\cores_M(z)\\
        &=0.
    \end{align*}
    This forces $a=0$ since $\Theta$ is injective, so $z=\Psi(a)=\Psi(0)=0$, meaning that $\cores_M$ is injective, as desired.
\end{proof}
This leads to a nice observation that $\cores_M$ is injective if all maximal planes pairwise intersect, showing that the barrier to $\cores_M$ being injective is exactly pairs of maximal planes with empty intersection.
\begin{corollary}
    If all planes in $\maxplanes(M)$ have pairwise nonempty intersections, then $K=0$ so $\cores_M$ is injective.
\end{corollary}

\subsection{From the reals to finite fields}\label{subsec:injectivity cores-reals to finite fields}
We finish this section by proving that we can transfer our result about the injectivity of $\cores_M$ over the reals to a finite field of large enough characteristic.

To simplify analysis, we note that we can remove the normalization of the sign matrices without affecting the injectivity of $\cores_M$ since they're just a rescaling. In particular, we can write $\sigma^{(i)}=\sqrt{k_i}G_i$ and, for $I\subseteq[D]$, set
\begin{align*}
    S_I=\left(\prod_{i\in I}\sqrt{k_i}\right)\Id_{U_I},
    \qquad S_\varnothing=\Id_\R.
\end{align*}
Then for $P=\plane(I,a)\subseteq Q=\plane(J,b)$, we can relate the corestriction map with the unnormalized signs $\cores_{P\to Q}^{\sigma}$ to the corestriction map with the normalized signs $\cores_{P\to Q}^{G}$ as
\begin{align*}
    \cores_{P\to Q}^{\sigma}=S_J\cores_{P\to Q}^{G}S_I^{-1}.
\end{align*}
Thus, these maps give an isomorphism of the normalized and unnormalized diagrams. Consequently, $\overline\beta_M^G$ is injective if and only if $\overline\beta_M^\sigma$ is injective, so any injective normalized instance gives a real sign witness in the sense of \Cref{def:real sign witness}. No full-row-rank assumption is needed here. 

\coresMinjectiveoverrealsimpliesoverF*

We will prove \Cref{lem:cores_M injective over reals gives injective over finite field} in two steps. First, we will show that a real sign witness gives a non-identically zero determinant polynomial over the integers with bounded degree and a specific value when evaluated at the witness. We state this for an arbitrary set and arbitrary dimensions, since the proof of \Cref{thm:cores_M injective whp} will need to apply it to both $(\cC_i)_{i\in[D]}$ and $(\cC_i^\perp)_{i\in[D]}$.
\begin{lemma}\label{lem:injectivity of cores_M over reals gives non zero det poly over Z}
    Let $M\subseteq[n]^D$, let $(k_i)_{i\in[D]}$ be dimensions, and let $\sigma$ be a real sign witness for $(M,(k_i)_{i\in[D]})$. Put
    \begin{align*}
        N_M:=\dim U_M=\sum_{P=\plane(I,a)\in\maxplanesclosed(M)}\ \prod_{i\in I}k_i.
    \end{align*}
    There is a polynomial $f_\sigma\in\Z[\bX]$ such that
    \begin{align}
        1\leq\abs{f_\sigma(\sigma)}&\leq N_M^{N_M/2}, &
        \deg f_\sigma&\leq DN_M,
        \label{eq:determinant certificate bounds}
    \end{align}
    and such that, over any field, every specialization $x$ of $\bX$ with $f_\sigma(x)\neq0$ makes the induced map $\overline\beta_M^{x}$ of \Cref{def:formal insertion diagram} injective. If, in addition, every component of $x$ has full row rank, then $\overline\beta_M^{x}$ is the corestriction map $\cores_M$ of the corresponding row-space codes, so $\cores_M$ is injective.
\end{lemma}
\begin{proof}
    Throughout the proof $M$ is fixed, so we drop it from the notation and write $\alpha^{\bX},\beta^{\bX}$ for the maps $\alpha_M^{\bX},\beta_M^{\bX}$ of \Cref{def:formal insertion diagram}, and $\alpha^{x},\beta^{x}$ for their specializations at a tuple of matrices $x$. Recall that $\beta^{\bX}\alpha^{\bX}=0$ is a polynomial identity and so holds after every specialization.

    We work first at the witness $\sigma$ and set $a:=\rank\alpha^{\sigma}$ and $b:=\rank\beta^{\sigma}$. Because $\ker\overline\beta^{\sigma}=\ker\beta^{\sigma}/\Im\alpha^{\sigma}$, the injectivity of $\overline\beta^{\sigma}$ says exactly that $\ker\beta^{\sigma}\subseteq\Im\alpha^{\sigma}$, while $\beta^{\sigma}\alpha^{\sigma}=0$ gives the reverse inclusion. Hence $\Im\alpha^{\sigma}=\ker\beta^{\sigma}$, and the rank-nullity theorem gives $a+b=N_M$.

    Choose $a$ independent columns of $\alpha^{\sigma}$ and $b$ independent rows of $\beta^{\sigma}$, and let $\widetilde{\alpha}^{\bX}$ and $\widetilde{\beta}^{\bX}$ be the submatrices of $\alpha^{\bX}$ and $\beta^{\bX}$ on those indices. The columns of $\widetilde{\alpha}^{\sigma}$ span $\Im\alpha^{\sigma}=\ker\beta^{\sigma}$, while the columns of $(\widetilde{\beta}^{\sigma})^\transpose$ span the row space of $\beta^{\sigma}$, which is the orthogonal complement of $\ker\beta^{\sigma}$ in $\R^{N_M}$. Consequently, the square matrix
    \begin{align*}
        \left[\widetilde{\alpha}^{\sigma}\ \middle|\ (\widetilde{\beta}^{\sigma})^\transpose\right]
    \end{align*}
    is invertible, so the determinant polynomial
    \begin{align*}
        f_\sigma(\bX):=
        \det\left[\widetilde{\alpha}^{\bX}\ \middle|\ (\widetilde{\beta}^{\bX})^\transpose\right]
        \in\Z[\bX]
    \end{align*}
    satisfies $f_\sigma(\sigma)\neq0$. This is the only place where the real inner product is used.

    Now let $x$ be any specialization over any field with $f_\sigma(x)\neq0$. Then the selected columns of $\alpha^{x}$ and rows of $\beta^{x}$ are still independent, so
    \begin{align*}
        \rank\alpha^{x}\geq a
        \qquad\text{and}\qquad
        \rank\beta^{x}\geq b.
    \end{align*}
    On the other hand $\beta^{x}\alpha^{x}=0$ forces $\Im\alpha^{x}\subseteq\ker\beta^{x}$, so $\rank\alpha^{x}+\rank\beta^{x}\leq N_M$. Since $a+b=N_M$, we can conclude that $\Im\alpha^{x}=\ker\beta^{x}$ and $\overline\beta^{x}$ is injective. Moreover, if every component of $x$ has full row rank, then \Cref{lem:cosheaf(M) as quotient of direct sum} identifies $\Coker\alpha^{x}$ with $\cosheaf(M)$ and $\overline\beta^{x}$ with $\cores_M$, so $\cores_M$ is injective.

    It remains to prove the bounds in \eqref{eq:determinant certificate bounds}. By \Cref{eq:cores explicit def}, each corestriction map inserts at most $D$ columns, so every entry of $\alpha^{\bX}$ and of $\beta^{\bX}$ is either zero or a signed monomial of degree at most $D$. Expanding the $N_M\times N_M$ determinant gives $\deg f_\sigma\leq DN_M$. At the sign instance $\sigma$, the matrix defining $f_\sigma$ has entries in $\{0,\pm1\}$, so each of its $N_M$ columns has Euclidean norm at most $\sqrt{N_M}$, and Hadamard's inequality gives $\abs{f_\sigma(\sigma)}\leq N_M^{N_M/2}$. Finally, $f_\sigma(\sigma)$ is a nonzero integer, so $\abs{f_\sigma(\sigma)}\geq1$.
\end{proof}
Second, we will show that we can upper bound the characteristic of the finite field. We need to ensure that, when modding by this characteristic, the determinant polynomial does not become zero.
\begin{lemma}\label{lem:bound on characteristic ensuring det poly not zero}
    Let $d_\sigma:=\deg f_\sigma$ and $m_\sigma:=\abs{f_\sigma(\sigma)}$. There is a prime $p\nmid m_\sigma$ satisfying
    \begin{align}
        \max\{2D+1,d_\sigma\log(n+1)\}<p,
        \qquad
        p=O_{D,\varepsilon}(n^D\log n).
        \label{eq:good characteristic bound}
    \end{align}
    In particular, $f_\sigma$ modulo $p$ is not identically zero.
\end{lemma}
\begin{proof}
    From \Cref{cor:PNT bounding non-dividing prime}, there exists a universal constant $C$ such that for $q=\max\{2D+1,d_\sigma\log(n+1)\}$ there exists a prime $p\nmid m_\sigma$ satisfying
    \begin{align}
        q< p\leq C(q+\log(2m_\sigma)).\label{eq:PNT consequence}
    \end{align}
    The lower bound of \Cref{eq:PNT consequence} is the first claim, so it remains only to bound its upper endpoint. We first bound $N_M$ by computing
    \begin{align*}
        N_M&=\dim\left(\bigoplus_{P\in\maxplanesclosed(M)}U_P\right)\\
        &=\sum_{I\subsetneq[D]}\abs{\maxAclosed{I}(M)}\prod_{i\in I}k_i\\
        &\leq 2^D\varepsilon n^D
        \sum_{I\subsetneq[D]}\prod_{i\in I}\frac{k_i}{n} &\text{By \Cref{cor:maxAIclosed is eps sparse,lem:eps closed and line free gives cardinality bound}}\\
        &\leq 2^D(2^D-1)\varepsilon n^D. &\text{By }k_i\leq n
    \end{align*}
    From \Cref{lem:injectivity of cores_M over reals gives non zero det poly over Z} we have that $d_\sigma\leq D N_M$ and $m_\sigma\leq  N_M^{N_M/2}$, so $\log(2m_\sigma)\leq \log 2+\frac{N_M}{2}\log N_M$. Thus,
    \begin{align*}
        C(q+\log(2m_\sigma))&\leq C\left(\max\{2D+1,d_\sigma\log(n+1)\}+\log(2m_\sigma)\right)\\
        &=O_{D,\varepsilon}(n^D\log(n)),
    \end{align*}
    as required.
    
    Since $p\nmid m_\sigma$, the reduction of $f_\sigma$ modulo $p$ evaluates nontrivially at $\sigma$ and is consequently not the zero polynomial.
\end{proof}
Now, we can prove \Cref{lem:cores_M injective over reals gives injective over finite field}. Our use of the Schwartz-Zippel lemma here is similar to that of Lemma 21 of \cite{kalachev_maximally_2025}.
\begin{proof}[Proof of \Cref{lem:cores_M injective over reals gives injective over finite field}]
    Let $\sigma$ be a real sign witness for $(M,(k_i)_{i\in[D]})$, let $f_\sigma$ be as in \Cref{lem:injectivity of cores_M over reals gives non zero det poly over Z}, and choose $p$ as in \Cref{lem:bound on characteristic ensuring det poly not zero}. Sample independent uniform matrices $G_i\in\F_p^{k_i\times n}$. Since $f_\sigma\bmod p$ is nonzero, the Schwartz--Zippel lemma and \Cref{lem:bound on characteristic ensuring det poly not zero} give
    \begin{align*}
        \Pr[f_\sigma(G_1,\ldots,G_D)=0]
        \leq\frac{d_\sigma}{p}
        <\frac{1}{\log(n+1)}.
    \end{align*}
    Let $E_{\text{full}}$ be the event that every $G_i$ has full row rank. For each $G_i$ individually, we have
    \begin{align*}
        \Pr[\rank G_i<k_i]
        \leq\sum_{j=0}^{k_i-1}p^{j-n}
        <\frac{1}{p-1}.
    \end{align*}
    Consequently, since $p>2D+1$ we can use a union bound to get
    \begin{align*}
        \Pr[E_{\text{full}}]&\geq1-\sum_{i=1}^D\Pr[\rank G_i<k_i]\\
        &>1-\frac{D}{p-1}\\
        &>\frac{1}{2}.
    \end{align*}
By \Cref{lem:injectivity of cores_M over reals gives non zero det poly over Z},
    \begin{align*}
        \Pr[\cores_M\text{ is not injective}\mid E_{\text{full}}]
        \leq\frac{\Pr[f_\sigma(G_1,\ldots,G_D)=0]}{\Pr[E_{\text{full}}]}
        <\frac{2}{\log(n+1)}.
    \end{align*}
    Conditioned on $E_{\text{full}}$, the row spaces of the $G_i$ are independent uniformly random $k_i$-dimensional codes. Every such code has exactly $\abs{\mathrm{GL}_{k_i}(\F_p)}$ ordered generator matrices, so choosing random generator matrices gives the same distribution.
\end{proof}

We can now choose a single characteristic that works for every \ec set and both $(\cC_i)_{i\in[D]}$ and $(\cC_i^\perp)_{i\in[D]}$ at once. The strategy is the product-of-polynomials method from the proof of Lemma 21 of \cite{kalachev_maximally_2025}, but applied to the witnesses of \Cref{lem:injectivity of cores_M over reals gives non zero det poly over Z} while taking care to choose a large enough characteristic.
\begin{proof}[Proof of \Cref{thm:cores_M injective whp}]
    The argument has five steps. We consider one set of dimensions for $(\cC_i)_{i\in[D]}$ and a fixed \ec $M$, then produce one real witness for each of them, then choose a single prime at which every witness survives, and finally sample the codes and take a union bound over both the codes and their duals.

    \paragraph{Step 1: A uniform dimension bound.}
    For $M=\varnothing$ the map $\cores_M$ has zero domain, and for $M=[n]^D$ it is the identity, so both are injective for every choice of codes. We may therefore restrict attention to the family $\mathscr M_\varepsilon$ of nonempty proper \ec subsets of $[n]^D$, and we write $s:=\abs{\mathscr M_\varepsilon}\leq2^{n^D}$.

    Fix $M\in\mathscr M_\varepsilon$ and $\eta\in\{+,-\}$, and let $N_M^\eta$ denote the dimension $N_M$ of \Cref{lem:injectivity of cores_M over reals gives non zero det poly over Z} for the profile $(k_i^\eta)_{i\in[D]}$. Using the fact that, for a given $I\subseteq[D]$, we have $\abs{\maxAIclosed}\leq n^{D-\abs{I}}$, and that $k_i^\eta\leq n$,
    \begin{align*}
        N_M^\eta=\sum_{P=\plane(I,a)\in\maxplanesclosed(M)}\ \prod_{i\in I}k_i^\eta
        \leq\sum_{I\subsetneq[D]}n^{D-\abs{I}}\prod_{i\in I}n=(2^D-1)n^D=:R.
    \end{align*}
    This loose count suffices, and we do not need to use any sparsity properties for the anchor sets here. Note also that $R\geq n\geq2$.

    \paragraph{Step 2: an integer certificate for each set and profile.}
    Both dimension profiles satisfy $k_i^\eta\geq\rho n$, and $n\geq N^*_{D,\varepsilon}\geq N_{D,\varepsilon}$. Moreover, \Cref{eq:real witness parameters} gives $2^D\varepsilon\leq1/(4A_{D,\rho}^2)<1$ and $C_{D,\rho,\varepsilon}=A_{D,\rho}\sqrt{2^D\varepsilon}\leq1/2$. Thus, \Cref{lem:cores_M injective whp over reals} gives an injective normalized instance with probability at least $1/2$ for this fixed $M$ and profile $(k_i^\eta)_{i\in[D]}$. Removing the normalization as above gives a real sign witness $\sigma_M^\eta$ for $(M,(k_i^\eta)_{i\in[D]})$. Applying \Cref{lem:injectivity of cores_M over reals gives non zero det poly over Z} to it produces a polynomial $f_M^\eta\in\Z[\bX]$ with
    \begin{align*}
        1\leq\abs{f_M^\eta(\sigma_M^\eta)}\leq\left(N_M^\eta\right)^{N_M^\eta/2}\leq R^{R/2}
        \qquad\text{and}\qquad
        \deg f_M^\eta\leq DN_M^\eta,
    \end{align*}
    using in the third inequality that $y\mapsto y^{y/2}$ is nondecreasing for $y\geq1$. Over any field, any specialization $x$ at which $f_M^\eta$ does not vanish makes $\overline\beta_M^{x}$ injective, and makes it equal to $\cores_M$ as soon as $x$ has full row rank. The witnesses are allowed to depend on both $M$ and $\eta$. However, we need one characteristic in which all of the $f_M^\eta$ survive.

    \paragraph{Step 3: one prime avoiding all of the certificates.}
    Set
    \begin{align*}
        m_*:=\prod_{\eta\in\{+,-\}}\ \prod_{M\in\mathscr M_\varepsilon}\abs{f_M^\eta(\sigma_M^\eta)},
        \qquad\text{and}\qquad
        L:=Dn+2DsR,
    \end{align*}
    with the convention that an empty product is $1$, so that $m_*\geq1$ is an integer. There are $2s$ factors, each at most $R^{R/2}$ by Step 2, so $\log m_*\leq sR\log R$. Applying \Cref{cor:PNT bounding non-dividing prime} with $m=m_*$ and $q=2+L\log(n+1)$ gives a prime $p\nmid m_*$ with
    \begin{align*}
        L\log(n+1)<p\leq C\bigl(2+L\log(n+1)+\log2+sR\log R\bigr)
        =O_D\left(2^{n^D}n^D\log n\right),
    \end{align*}
    where the last estimate uses $s\leq2^{n^D}$ and $R=(2^D-1)n^D$. Since $p$ divides none of the integers $f_M^\eta(\sigma_M^\eta)$, each $f_M^\eta$ reduces to a nonzero polynomial over $\F_p$, and hence over each of its extensions $\F_{p^t}$. Observe that this step fixes $p$ before any code is sampled, and that the same $p$ works for every $M$ and both dimension profiles.

    \paragraph{Step 4: full rank and simultaneous injectivity over \texorpdfstring{$\F_{p^t}$}{the extension field}.}
    Fix an integer $t\geq1$ and a profile $\eta\in\{+,-\}$, and let $X_1,\dots,X_D$ with $X_i\in\F_{p^t}^{k_i^\eta\times n}$ have independent uniformly random entries. Writing $\Delta_i(\bX)$ for the determinant of the formal matrix $(X^{(i)}_{r,j})_{r,j\in[k_i^\eta]}$, set
    \begin{align*}
        h^\eta:=\prod_{i=1}^D\Delta_i,
        \qquad
        F^\eta:=h^\eta\prod_{M\in\mathscr M_\varepsilon}f_M^\eta,
        \qquad
        d_\eta:=\sum_{i=1}^Dk_i^\eta+D\sum_{M\in\mathscr M_\varepsilon}N_M^\eta,
    \end{align*}
    where all of the $f_M^\eta$ are formed in the same variables $\bX$, so that a single sample specializes all of them at once. Each $\Delta_i$ contains the diagonal monomial $\prod_{r\in[k_i^\eta]}X^{(i)}_{r,r}$ with coefficient $1$ and so is a nonzero polynomial in $\F_{p^t}[\bX]$. Combined with Step 3 this shows that $F^\eta$ is a nonzero polynomial in $\F_{p^t}[\bX]$, and $\deg F^\eta\leq d_\eta$ by Step 2. The Schwartz--Zippel lemma therefore gives
    \begin{align*}
        \Pr\left[F^\eta(X_1,\dots,X_D)\neq0\right]\geq1-\frac{d_\eta}{p^t}.
    \end{align*}
    Let $E_\eta$ be the event that every $X_i$ has full row rank, and let $A_\eta$ be the event that $\cores_M$ is injective for every $M\in\mathscr M_\varepsilon$ in the diagram of the row spaces of $X_1,\dots,X_D$. A nonzero evaluation $\Delta_i(X_1,\dots,X_D)$ forces $X_i$ to have full row rank, so by Step 2 we have $\{F^\eta(X_1,\dots,X_D)\neq0\}\subseteq A_\eta\cap E_\eta$. Since $k_i^\eta\leq n$, full-row-rank matrices exist and $\Pr[E_\eta]>0$, so we may condition on $E_\eta$; using $\Pr[E_\eta]\leq1$ in the denominator,
    \begin{align*}
        \Pr[A_\eta\mid E_\eta]=\frac{\Pr[A_\eta\cap E_\eta]}{\Pr[E_\eta]}
        \geq\Pr\left[F^\eta(X_1,\dots,X_D)\neq0\right]
        \geq1-\frac{d_\eta}{p^t}.
    \end{align*}
    Moreover, conditioned on $E_\eta$ the row spaces of $X_1,\dots,X_D$ are independent and uniformly random in $\grasspt(n,k_i^\eta)$, because every $k_i^\eta$-dimensional code has exactly $\abs{\mathrm{GL}_{k_i^\eta}(\F_{p^t})}$ ordered generator matrices. So this conditional sampling is the same as choosing a random full rank generator matrix for the profile $(k_i^\eta)_{i\in[D]}$.

    \paragraph{Step 5: the union bound over the tuple and its dual.}
    Let $\cC_1,\dots,\cC_D$ be independent and uniformly random with $\cC_i\in\grasspt(n,k_i)$. Applying Step 4 with $\eta=+$, the maps $\cores_M$ for $(\cC_i)_{i\in[D]}$ fail to be simultaneously injective with probability at most $d_+/p^t$. Since the dual of a random $k_i$ dimensional space is a random $(n-k_i)$ dimensional space, the tuple $(\cC_i^\perp)_{i\in[D]}$ is itself uniformly random with the dimensions $(k_i^-)_{i\in[D]}$. Applying Step 4 with $\eta=-$, its global corestriction maps fail to be injective with probability at most $d_-/p^t$. A union bound gives a total failure probability of at most
    \begin{align*}
        \frac{d_++d_-}{p^t}\leq\frac{L}{p^t}<\frac{p^{1-t}}{\log(n+1)}.
    \end{align*}
    Here the first inequality is the computation
    \begin{align*}
        d_++d_-=\sum_{i=1}^D\left(k_i+(n-k_i)\right)+D\sum_{M\in\mathscr M_\varepsilon}\left(N_M^++N_M^-\right)\leq Dn+2DsR=L,
    \end{align*}
    using $N_M^\eta\leq R$ from Step 1, and the second holds because $p>L\log(n+1)$ by Step 3. Together with Step 1, this is exactly the claim.
\end{proof}

\dobib

\section{Bounding the Operator Norm of the Plane-Interaction Matrix}\label{sec:expected op norm of K}
In this section, we will take the same setup as the previous section and continue to work over $\R$ to lower bound the probability that the operator norm of the plane-interaction matrix $K$ of $M$ is less than one.
\probOpNormKlessonelarge*

Since we will do this via Markov's inequality, it suffices to bound the expected operator norm of $K$. Thus, the main result that we will derive in this section is the following upper bound. 
\begin{lemma}\label{lem:UB on expected opnorm of K}
    For constants $D\in\N$, $\rho\in(0,1)$, and $\varepsilon\in(0,2^{-D})$, let
    \[
        n\geq N_{D,\varepsilon}:=\max\left\{16\left(\frac{192D}{e}\right)^{24}, (2^D\varepsilon)^{-2}\right\}.
    \]
    Fix an \ec set $M\subseteq[n]^D$, and let $\rho n\leq k_i\leq n$ for every $i\in[D]$. Then
    \begin{align*}
        \E_\sigma[\opnorm{K}]&\leq C_{D,\rho,\varepsilon}.
    \end{align*}
\end{lemma}
A single application of Markov's inequality gives us the necessary result for the previous section.
\begin{proof}[Proof of \Cref{thm:prob op norm K < 1 large}]
    We directly compute
    \begin{align*}
        \Pr_\sigma[\opnorm{K}<1]&\geq 1-\E_\sigma[\opnorm{K}] &\text{By Markov}\\
        &\geq 1- C_{D,\rho,\varepsilon} &\text{By \Cref{lem:UB on expected opnorm of K}}
    \end{align*}
\end{proof}

Our path to proving \Cref{lem:UB on expected opnorm of K} begins in \Cref{subsec:boundexpopK-decomp of K} by rearranging $K$ into ``superblocks'' of the blocks $K_{P,Q}$ used to define it and then decomposing these blocks as sums of restricted operators whose operator norm we can cleanly express as a product of operator norms. Then, in \Cref{subsec:boundexpopK-bounding KIJE}, we bound the expected operator norm of these restricted operators. Finally, we complete several deferred proofs in \Cref{subsec:boundexpopK-deferred}.

\subsection{Decomposition of the plane-interaction matrix}\label{subsec:boundexpopK-decomp of K}
To bound the expected operator norm of $K$, we proceed via two decomposition steps. In \Cref{subsubsec:expop-directional block decomposition}, we simply change our view from the block definition of $K$ via the plane blocks $K_{P,Q}$ for planes $P$ and $Q$ to what we call \emph{directional blocks} $K_{I,J}$ defined by two directions $I,J\subseteq[D]$. There are only $4^D$ such pairs $(I,J)$, so bounding each $\opnorm{K_{I,J}}$ lets us easily bound $\opnorm{K}$ via the triangle inequality. To work on these $K_{I,J}$'s, however, we perform one more decomposition in \Cref{subsubsec:expop-decomp over mutually fixed} over subsets of $\setcomp{(I\cup J)}$. This latter decomposition is where we use both the fact that $M$ is \ec and that the entries of our generator matrices are random Rademacher variables.

\subsubsection{Directional block decomposition}\label{subsubsec:expop-directional block decomposition}
This first decomposition combines all plane blocks $K_{P,Q}$, where $P$ is an $I$-plane and $Q$ is a $J$-plane, together into one direction block $K_{I,J}$. 
\begin{definition}
    Since each plane block $K_{P,Q}:U_Q\to U_P$ and for two directions $I,J\subseteq[D]$ the planes in $\maxplanesclosed(M)$ are indexed by $\maxAIclosed(M)$ and $\maxAclosed{J}(M)$, respectively, we can define the block $K_{I,J}$ combining all blocks $K_{P,Q}$ for an $I$-plane $P$ and a $J$-plane $Q$ as an operator $K_{I,J}:U_J\otimes \R^{\maxAclosed{J}}\to U_I\otimes\R^{\maxAIclosed}$ where for $a\in\maxAIclosed$ and $b\in\maxAclosed{J}$:
    \begin{align*}
        (K_{I,J})_{a,b}=K_{\plane(I,a),\plane(J,b)}.
    \end{align*}
    Thus, $K_{I,J}$ is $\left(\abs{\maxAIclosed}\prod_{i\in I}k_i\right)\times\left(\abs{\maxAclosed{J}}\prod_{j\in J}k_j\right)$ dimensional.
\end{definition}

Later, we will prove the following bound on the expected operator norm of each $K_{I,J}$.
\begin{lemma}\label{lem:UB on opnorm of KIJ}
    For the direction block $K_{I,J}$ defined as above and $n\geq N_{D,\varepsilon}:=\max\left\{16\left(\frac{192D}{e}\right)^{24}, (2^D\varepsilon)^{-2}\right\}$, we have
    \begin{align*}
        \E_\sigma[\opnorm{K_{I,J}}]&\leq 2C_\rho^D\left(1+\frac{4}{\rho}\right)^D\sqrt{2^D\varepsilon}.
    \end{align*}
\end{lemma}

With this bound in hand, we can use the triangle inequality as mentioned previously to finish off \Cref{lem:UB on expected opnorm of K}.

\begin{proof}[Proof of \Cref{lem:UB on expected opnorm of K}]
    If $M=\varnothing$ or $M=[n]^D$, then $K=0$, so assume $\varnothing\neq M\subsetneq[n]^D$. Since both $I,J\subseteq[D]$, there are at most $4^D$ possible choices of $(I,J)$. For each $K_{I,J}$, we can pad it with zeros, which doesn't change the operator norm, to be the same dimension as $K$. Thus, we can compute
    \begin{align*}
        \E_\sigma\left[\opnorm{K}\right]&=\E_\sigma\left[\opnorm{\sum_{I,J\subseteq[D]}K_{I,J}}\right]\\
        &\leq \sum_{I,J\subseteq[D]}\E_\sigma\left[\opnorm{K_{I,J}}\right] &\text{By triangle ineq}\\
        &\leq\sum_{I,J\subseteq[D]} 2C_\rho^D\left(1+\frac{4}{\rho}\right)^D\sqrt{2^D\varepsilon} &\text{By \Cref{lem:UB on opnorm of KIJ}}\\
        &=4^D 2C_\rho^D\left(1+\frac{4}{\rho}\right)^D\sqrt{2^D\varepsilon}\\
        &=C_{D,\rho,\varepsilon},
    \end{align*}
    as claimed.
\end{proof}

\subsubsection{Decomposition over sets of mutually fixed coordinates}\label{subsubsec:expop-decomp over mutually fixed}
Next, to achieve \Cref{lem:UB on opnorm of KIJ}, we write out explicitly the action of $K_{I,J}$ and observe a clean way to write it as a sum of matrices. This will require some new notation.

\begin{definition}
    With $i\in[D]$ and $G_i\in\R^{k_i\times n}$ as the generator matrix of the component code $\cC_i$, we define the $n\times n$ real matrices
    \begin{align*}
        B_i=G_i^\transpose G_i \qquad D_i=B_i-\Id_{\R^n},
    \end{align*}
    where we can see that $D_i(a,a)=0$ and $D_i(a,b)=\frac{1}{k_i}\sum_{r=1}^{k_i}\sigma^{(i)}_{r,a}\sigma^{(i)}_{r,b}$ for $a\neq b$. Moreover, for a set $S\subseteq[D]$ we let
    \begin{align*}
        G_S=\bigotimes_{i\in S}G_i, \qquad B_S=\bigotimes_{i\in S}B_i, \qquad D_S=\bigotimes_{i\in S}D_i.
    \end{align*}
    We write $g_S(a)=\bigotimes_{i\in S}g_i(a_i)$ for $a\in[n]^S$, with $g_\varnothing=1$ and empty tensor products of operators interpreted as the identity on $\R$.
\end{definition}

\begin{definition}
    For a subset $S\subseteq X$ of some set $X$, we define the canonical inclusion and projection maps
    \begin{align*}
        \iota_S:\R^S&\to\R^X\\
        \Pi_S:\R^X&\to\R^S.
    \end{align*}
\end{definition}
We now state the decomposition.
\begin{restatable}{lemma}{hatKIJdecomp}\label{lem:hatKIJ decomposition}
    For fixed $I,J\subseteq[D]$, define
    \begin{align*}
        C=I\cap J \qquad L=I\setminus J \qquad R=J\setminus I \qquad F=\setcomp{(I\cup J)} .
    \end{align*}
    Then, after ordering the tensor coordinates in order as $(C,L,R,F)$, we can decompose $K_{I,J}$ as
    \begin{align}
        K_{I,J}&=\left(\Id_{U_I}\otimes \Pi_{\maxAIclosed}\right)\left(\Id_{U_C}\otimes\ G_L\otimes G_R^\transpose\otimes\left(B_F-\Id_{\R^{[n]^F}}\right)\right)\left(\Id_{U_J}\otimes\iota_{\maxAclosed{J}}\right)\label{eq:hatKIJ complete expression}\\
        &=\sum_{\varnothing\neq E\subseteq F}K_{I,J}^E\label{eq:hatKIJ sum expression}
    \end{align}
    where
    \begin{align}
        K_{I,J}^E=\left(\Id_{U_I}\otimes \Pi_{\maxAIclosed}\right)\left(\Id_{U_C}\otimes\ G_L\otimes G_R^\transpose\otimes D_E\otimes\Id_{\R^{[n]^{F\setminus E}}}\right)\left(\Id_{U_J}\otimes\iota_{\maxAclosed{J}}\right).\label{eq:hatKIJE expression}
    \end{align}
\end{restatable}

We defer the proof of \Cref{lem:hatKIJ decomposition} to \Cref{subsec:boundexpopK-deferred} and provide a bound on $\E\left[\opnorm{K_{I,J}^E}\right]$.

\begin{lemma}\label{lem:UB on opnorm of KIJE}
    For a non-empty subset $E\subseteq F$ and $n\geq N_{D,\varepsilon}:=\max\left\{16\left(\frac{192D}{e}\right)^{24}, (2^D\varepsilon)^{-2}\right\}$, we have that
    \begin{align*}
        \E_\sigma\left[\opnorm{K_{I,J}^E}\right]&\leq 2C_\rho^{\abs{L}+\abs{R}}\left(\frac{4}{\rho}\right)^{\abs{E}}\sqrt{2^D\varepsilon}
    \end{align*}
\end{lemma}

Given this bound on each $\E\left[\opnorm{K_{I,J}^E}\right]$, we can use the binomial theorem to sum over all subsets $E$ to bound $\E\left[\opnorm{K_{I,J}}\right]$.

\begin{proof}[Proof of \Cref{lem:UB on opnorm of KIJ}]
    We start from \Cref{lem:hatKIJ decomposition} to compute
    \begin{align*}
        \E_\sigma\left[\opnorm{K_{I,J}}\right]&=\E_\sigma\left[\opnorm{\sum_{\varnothing\neq E\subseteq F}K_{I,J}^E}\right] &\text{By \Cref{eq:hatKIJ sum expression}}\\
        &\leq\sum_{\varnothing\neq E\subseteq F}\E_\sigma\left[\opnorm{K_{I,J}^E}\right] &\text{By triangle ineq}\\
        &\leq \sum_{\varnothing\neq E\subseteq F} 2C_\rho^{\abs{L}+\abs{R}}\left(\frac{4}{\rho}\right)^{\abs{E}}\sqrt{2^D\varepsilon} &\text{By \Cref{lem:UB on opnorm of KIJE}}\\
        &\leq 2C_\rho^{D}\sqrt{2^D\varepsilon} \sum_{\varnothing\neq E\subseteq F}\left(\frac{4}{\rho}\right)^{\abs{E}}\\
        &\leq 2C_\rho^{D}\sqrt{2^D\varepsilon}\left(1+\frac{4}{\rho}\right)^D &\text{By binomial thm}
    \end{align*}
\end{proof}

\subsection{Bounding the expected operator norm of restricted directional blocks}\label{subsec:boundexpopK-bounding KIJE}

The proof of \Cref{lem:UB on opnorm of KIJE} requires two main ingredients: First, we will upper bound $\E\left[\opnorm{K_{I,J}^E}\right]$ by the product of the expected norm of the tensor of our generator matrices and the maximum operator norm of a sparsely applied operator. Second, we upper bound the expected operator norm of such sparsely applied operators. 

We begin by using the definition of $K_{I,J}^E$ from \Cref{eq:hatKIJE expression} to derive an upper bound on its operator norm.

\begin{restatable}{lemma}{UBonNormKIJE}\label{lem:UB on op norm of KIJE}
    For a non-empty subset $E\subseteq F=\setcomp{(I\cup J)}$, let $Z=F\setminus E$ so that we can decompose $[n]^{\setcomp{I}}=[n]^R\times[n]^E\times[n]^Z$ and $[n]^{\setcomp{J}}=[n]^L\times[n]^E\times[n]^Z$. Then define
    \begin{align*}
        X_{r,z}&=\{a\in[n]^E:(r,a,z)\in\maxAIclosed\}\\
        Y_{\ell,z}&=\{b\in[n]^E:(\ell,b,z)\in\maxAclosed{J}\}.
    \end{align*}
    Then we can bound the operator norm of $K_{I,J}^E$ as
    \begin{align}
        \opnorm{K_{I,J}^E}&\leq\opnorm{G_L}\opnorm{G_R}\max_{\substack{\ell\in [n]^L\\ r\in [n]^R\\ z\in [n]^Z}}\opnorm{ \Pi_{X_{r,z}}D_E\iota_{Y_{\ell,z}}}.\label{eq:UB on op norm of KIJE}
    \end{align}
\end{restatable}
Again, we defer the proof to \Cref{subsec:boundexpopK-deferred} since it is mainly derived via careful algebraic manipulation of the definition of $K_{I,J}^E$. All that is left to obtain \Cref{lem:UB on opnorm of KIJE} then is to bound, in expectation, $\opnorm{G_L}$, $\opnorm{G_R}$, and $\max_{\ell,r,z}\opnorm{ \Pi_{X_{r,z}}D_E\iota_{Y_{\ell,z}}}$. For $G_L$ and $G_R$, we use standard results relating the operator norm of random Rademacher matrices to Gaussian matrices to get the following bounds, whose proofs we also defer.

\begin{restatable}{proposition}{UbonNormRademacherMatrix}\label{prop:UB on opnorm of rademacher matrix}
    For a normalized (each column is scaled to have norm 1) random Rademacher matrix 
    \begin{align*}
        G=\frac{1}{\sqrt{k}}(\sigma_{r,c})_{r\in[k],c\in[n]}
    \in \R^{k\times n},
    \end{align*}
    
    where $k\geq\rho n$, its expected operator norm over the random signs is at most
    \begin{align*}
        \E[\opnorm{G}]&\leq C_\rho = 3\left(1+\rho^{-1/2}\right).
    \end{align*}
\end{restatable}

Using the independence of the signs of our generator matrices, we immediately get the following corollary.

\begin{corollary}\label{cor:op norm UB on G_S}
    For a set $S\subseteq[D]$ and normalized random Rademacher matrices $(G_i)_{i\in[D]}$ with $k_i\geq\rho n$, we have that
    \begin{align*}
        \E[\opnorm{G_S}]&\leq C_\rho^{\abs{S}}.
    \end{align*}
\end{corollary}
\begin{proof}
    Apply \Cref{fact:opnorm of tensor is prod of opnorms} to \Cref{prop:UB on opnorm of rademacher matrix}.
\end{proof}

To deal with $\max_{\ell,r,z}\opnorm{ \Pi_{X_{r,z}}D_E\iota_{Y_{\ell,z}}}$, we state a general proposition about the expected maximum operator norm of such sparsely applied operators with the maximum taken over bounded indexing sets. The proof of this proposition, which comprises all of \Cref{sec:bounding max opnorm with wigner}, crucially utilizes both the fact that $M$ is \ec and properties of random Rademacher matrices.

\begin{proposition}[Verbose version of \Cref{prop:UB on op norm of max T_lambda}]
    \label{prop:UB on op norm of max T_lambda-verbose}
    Take fixed integers $\nu\geq 0$ and $q\geq 1$ such that $\nu+q\leq D$. Let $\Lambda$ be a fixed indexing set such that $\abs{\Lambda}\leq n^\nu$, and take normalized random Rademacher matrices $(G_i)_{i\in[q]}$ with $k_i\geq\rho n$ and $D_i=G_i^\transpose G_i-\Id_{\R^{[n]}}$.
    For all indices $\lambda\in\Lambda$, take $X_\lambda,Y_\lambda\subseteq[n]^q$ to be $\eta$-sparse for some constant $\eta\in(0,1)$, fixed independently of the signs. Letting
    \begin{align*}
        T_\lambda&= \Pi_{X_\lambda} \left(\bigotimes_{j=1}^q D_j\right)\iota_{Y_\lambda}:\R^{Y_\lambda}\to \R^{X_\lambda}, 
    \end{align*}
    we have that for $n\geq N_{D,q,\eta}:=\max\left\{16\left(\frac{192D}{e}\right)^{24}, \eta^{-2/q}\right\}$,
    \begin{align*}
        \E[\max_{\lambda\in\Lambda}\opnorm{T_\lambda}]&\leq 2\left(\frac{4}{\rho}\right)^q\sqrt{\eta}.
    \end{align*}
    If $\Lambda=\varnothing$, then this expectation is just 0. 
\end{proposition}

Having bounded each of the components of \Cref{eq:UB on op norm of KIJE}, we can now prove \Cref{lem:UB on opnorm of KIJE}.

\begin{proof}[Proof of \Cref{lem:UB on opnorm of KIJE}]
    From \Cref{lem:UB on op norm of KIJE} we have that 
    \begin{align*}
        \opnorm{K_{I,J}^E}&\leq\opnorm{G_L}\opnorm{G_R}\max_{\substack{\ell\in [n]^L\\ r\in [n]^R\\ z\in [n]^Z}}\opnorm{ \Pi_{X_{r,z}}D_E\iota_{Y_{\ell,z}}}
    \end{align*}
    where $\varnothing\neq E\subseteq F=\setcomp{(I\cup J)}$ and $Z=F\setminus E$ so that we can decompose $[n]^{\setcomp{I}}=[n]^R\times[n]^E\times[n]^Z$ and $[n]^{\setcomp{J}}=[n]^L\times[n]^E\times[n]^Z$ to define
    \begin{align*}
        X_{r,z}&=\{a\in[n]^E:(r,a,z)\in\maxAIclosed\}\\
        Y_{\ell,z}&=\{b\in[n]^E:(\ell,b,z)\in\maxAclosed{J}\}.
    \end{align*}
    To apply \Cref{prop:UB on op norm of max T_lambda-verbose} with $q=\abs{E}$ and $\nu=\abs{L}+\abs{R}+\abs{Z}$, we note that number of triples $(\ell,r,z)$ is at most $n^{\abs{L}+\abs{R}+\abs{Z}}$ and that $\abs{L}+\abs{R}+\abs{Z}+\abs{E}=D-\abs{C}\leq D$.\footnote{In this application of \Cref{prop:UB on op norm of max T_lambda-verbose}, we are considering the $q$ generator matrices used the proposition to be exactly $(G_i)_{i\in E}$ and not necessarily the first $q$ generator matrices of $(G_i)_{i\in[D]}$.}Furthermore, since $\maxAIclosed$ and $\maxAclosed{J}$ are $2^D\varepsilon$-sparse by \Cref{cor:maxAIclosed is eps sparse}, $X_{r,z}$ and $Y_{\ell,z}$ are as well, so we can take $\eta=2^D\varepsilon$. Our choice $n\geq N_{D,\varepsilon}$ also gives $n\geq N_{D,q,\eta}$, since $\eta^{-2/q}\leq\max\{1,\eta^{-2}\}$ for $q\geq1$. The three factors in \Cref{eq:UB on op norm of KIJE} depend on disjoint collections of signs indexed by $L$, $R$, and $E$, respectively, since the sets $X_{r,z}$ and $Y_{\ell,z}$ depend only on the fixed set $M$. Thus, independence and \Cref{cor:op norm UB on G_S} give
    \begin{align*}
        \E_\sigma\left[\opnorm{K_{I,J}^E}\right]&\leq\E_\sigma\left[\opnorm{G_L}\right]\E_\sigma\left[\opnorm{G_R}\right]\E_\sigma\left[\max_{\ell,r,z}\opnorm{ \Pi_{X_{r,z}}D_E\iota_{Y_{\ell,z}}}\right]\\
        &\leq C_\rho^{\abs{L}+\abs{R}}\cdot 2\left(\frac{4}{\rho}\right)^{\abs{E}}\sqrt{2^D\varepsilon},
    \end{align*}
    as claimed.
\end{proof}

\subsection{Deferred proofs}\label{subsec:boundexpopK-deferred}

Finally, we present the deferred proofs of \Cref{lem:hatKIJ decomposition,lem:UB on op norm of KIJE} and \Cref{prop:UB on opnorm of rademacher matrix}. For convenience, we restate the relevant lemma and propositions.

\hatKIJdecomp*
\begin{proof}[Proof of \Cref{lem:hatKIJ decomposition}]
    We begin by formalizing what we mean by ordering the tensor coordinates in order as $(C,L,R,F)$. Since $I=C\sqcup L$, $\setcomp{I}=R\sqcup F$, $J=C\sqcup R$ and $\setcomp{J}=L\sqcup F$, we can decompose and order the ambient domain and codomain of $K_{I,J}$ as
    \begin{align}
        U_J\otimes\R^{[n]^{\setcomp{J}}}&= U_{C\sqcup R}\otimes \R^{[n]^{L\sqcup F}}\cong U_C\otimes\R^{[n]^L}\otimes U_R \otimes \R^{[n]^F}\label{eq:domain of hatKIJ}\\
        U_I\otimes \R^{[n]^{\setcomp{I}}}&= U_{C\sqcup L}\otimes \R^{[n]^{R\sqcup F}}\cong U_C\otimes U_L\otimes \R^{[n]^R}\otimes \R^{[n]^F},\label{eq:codomain of hatKIJ}
    \end{align}
    respectively. Recall that by convention, $U_\varnothing=\R$ and $\R^{[n]^\varnothing}=\R$. Moreover, ordering the tensor coordinates is equivalent to multiplying by a unitary matrix and doesn't change the operator norm.

    To show \Cref{eq:hatKIJ complete expression}, we will let $\widehat{K}_{I,J}$ be defined as the right hand side of \Cref{eq:hatKIJ complete expression} and prove that $K_{I,J}$ and $\widehat{K}_{I,J}$ agree on all of their plane-interaction blocks. Take
    \begin{align*}
        a&=(a_R,a_F)\in\maxAIclosed\subset[n]^R\times[n]^F\\
        b&=(b_L,b_F)\in\maxAclosed{J}\subseteq[n]^L\times[n]^F
    \end{align*}
    and consider the planes $P=\plane(I,a)$ and $Q=\plane(J,b)$, so $(K_{I,J})_{a,b}=K_{P,Q}$ and $(\widehat{K}_{I,J})_{a,b}=\widehat{K}_{P,Q}$. For a pure tensor $u_C\otimes u_R\in U_C\otimes U_R\cong U_J$, we apply the first part $\Id_{U_J}\otimes\iota_{\maxAclosed{J}}$ of $\widehat{K}_{I,J}$ and use \Cref{eq:domain of hatKIJ} to perform the identification
    \begin{align}
        (u_C\otimes u_R)\otimes e_b \mapsto u_C\otimes e_{b_L}\otimes u_R\otimes e_{b_F}. \label{eq:first operator of KIJ}
    \end{align}
    Applying the second operator of $\widehat{K}_{I,J}$ to \Cref{eq:first operator of KIJ} gives
    \begin{align*}
        &\left(\Id_{U_C}\otimes\ G_L\otimes G_R^\transpose\otimes\left(B_F-\Id_{\R^{[n]^F}}\right)\right)\left(u_C\otimes e_{b_L}\otimes u_R\otimes e_{b_F}\right)\\
        =&\ u_C\otimes g_L(b_L)\otimes G_R^\transpose u_R\otimes (B_F-\Id_{\R^{[n]^F}})e_{b_F}\\
        =&\ u_C\otimes g_L(b_L)\otimes \left(\sum_{r\in[n]^R}(G_R^\transpose u_R)(r)e_r\right)\otimes\left(\sum_{f\in[n]^F}\left(B_F(f,b_F)-\Ind\{f=b_F\}\right)e_f\right).
    \end{align*}
    Thus, after applying the third operator, we get
    \begin{align*}
        (\widehat{K}_{I,J})_{a,b}(u_C\otimes u_R)&=\widehat{K}_{P,Q}(u_C\otimes u_R)\\
        &=\left((G_R^\transpose u_R)(a_R)\right)\left(B_F(a_F,b_F)-\Ind\{a_F=b_F\}\right)\left(u_C\otimes g_L(b_L)\right).\numberthis\label{eq:hat KPQ expanded}
    \end{align*}
    We now observe what $K_{P,Q}$ does to $u_C\otimes u_R$ by computing
    \begin{align*}
        \cores_Q(u_C\otimes u_R)&=u_C\otimes g_L(b_L)\otimes u_R\otimes g_F(b_F)\\
        \cores_P^\transpose\cores_Q(u_C\otimes u_R)&=\langle g_R(a_R),u_R\rangle\langle g_F(a_F),g_F(b_F)\rangle (u_C\otimes g_L(b_L))\\
        &=\left((G_R^\transpose u_R)(a_R)\right)\left(B_F(a_F,b_F)\right) \left(u_C\otimes g_L(b_L)\right). 
    \end{align*}
    Recall though that if $P$ and $Q$ intersect, then $K_{P,Q}=0$ instead of $\cores_P^\transpose\cores_Q$. Since $P\cap Q\neq\varnothing$ if and only if $a_F=b_F$, then 
    \begin{align}
        B_F(a_F,b_F)-\Ind\{a_F=b_F\}=
        \begin{cases}
            B_F(a_F,b_F) & P\cap Q=\varnothing\\
            0 & P\cap Q\neq\varnothing
        \end{cases}\label{eq:B_F - Ind case decomposition}
    \end{align}
    Therefore, combining the expression for $\widehat{K}_{P,Q}$ from \Cref{eq:hat KPQ expanded} with \Cref{eq:B_F - Ind case decomposition} shows us that
    \begin{align*}
        \widehat{K}_{P,Q}&=
        \begin{cases}
            \cores_P^\transpose\cores_Q & P\cap Q=\varnothing\\
            0 & P\cap Q\neq\varnothing
        \end{cases}
        \\
        &=K_{P,Q}.
    \end{align*}
    Since this holds for any $a\in\maxAIclosed$ and $b\in\maxAclosed{J}$, we can conclude that $K_{I,J}=\widehat{K}_{I,J}$, as claimed in \Cref{eq:hatKIJ complete expression}.

    Finally, to derive \Cref{eq:hatKIJ sum expression}, we expand $B_F-\Id_{\R^{[n]^F}}$ as 
    \begin{align*}
        B_F-\Id_{\R^{[n]^F}}&=\bigotimes_{f\in F}B_f - \Id_{\R^{[n]^F}}\\
        &=\bigotimes_{f\in F}(D_f+\Id_{\R^{[n]^{\{f\}}}}) - \Id_{\R^{[n]^F}}\\
        &=\left[\sum_{E\subseteq F}\left(\bigotimes_{e\in E} D_e \otimes \bigotimes_{f\in F\setminus E} \Id_{\R^{[n]^{\{f\}}}}\right)\right]-\Id_{\R^{[n]^F}}\\
        &=\sum_{\varnothing\neq E\subseteq F}\left(D_E \otimes \Id_{\R^{[n]^{F\setminus E}}}\right).\numberthis\label{eq:BF minus Id RF as sum}
    \end{align*}
    Plugging \Cref{eq:BF minus Id RF as sum} into \Cref{eq:hatKIJ complete expression} then yields
    \begin{align*}
         K_{I,J}&=\left(\Id_{U_I}\otimes\Pi_{\maxAIclosed}\right)\left(\Id_{U_C}\otimes G_L\otimes G_R^\transpose\otimes\left(B_F-\Id_{\R^{[n]^F}}\right)\right)\left(\Id_{U_J}\otimes\iota_{\maxAclosed{J}}\right)\\
         &=\sum_{\varnothing\neq E\subseteq F}\left(\Id_{U_I}\otimes\Pi_{\maxAIclosed}\right)\left(\Id_{U_C}\otimes G_L\otimes G_R^\transpose\otimes D_E\otimes\Id_{\R^{[n]^{F\setminus E}}}\right)\left(\Id_{U_J}\otimes\iota_{\maxAclosed{J}}\right)\\
         &=\sum_{\varnothing\neq E\subseteq F} K_{I,J}^E,
    \end{align*}
    as claimed in \Cref{eq:hatKIJ sum expression}.
\end{proof}

\UBonNormKIJE*
\begin{proof}[Proof of \Cref{lem:UB on op norm of KIJE}]
    By definition of $X_{r,z}$ and $Y_{\ell,z}$, we can decompose $\maxAIclosed$ and $\maxAclosed{J}$ as
    \begin{align*}
        \maxAclosed{J} &= \bigsqcup_{\substack{\ell\in[n]^L \\ z\in[n]^Z}} \{(\ell,b,z):b\in Y_{\ell,z}\}\\
        \maxAIclosed &= \bigsqcup_{\substack{r\in[n]^R \\ z\in[n]^Z}} \{(r,a,z):a\in X_{r,z}\}.
    \end{align*}
    Therefore, we can decompose the domain $U_J\otimes\R^{\maxAclosed{J}}$ and the codomain $U_I\otimes\R^{\maxAIclosed}$ of $K_{I,J}^E$ as
    \begin{align}
        U_J\otimes\R^{\maxAclosed{J}}&\cong U_C\otimes U_R\otimes\R^{\maxAclosed{J}} \cong\bigoplus_{\substack{\ell\in[n]^L \\ z\in[n]^Z}} U_C\otimes U_R\otimes\R^{Y_{\ell,z}}\label{eq:UJ otimes RnAj decomposed}\\
        U_I\otimes\R^{\maxAIclosed}&\cong U_C\otimes U_L\otimes\R^{\maxAIclosed}\cong\bigoplus_{\substack{r\in[n]^R \\ z\in[n]^Z}} U_C\otimes U_L\otimes\R^{X_{r,z}}\label{eq:UI otimes RnAI decomposed}.
    \end{align}

    Using these decompositions, we'll write $K_{I,J}^E$ as the product of three operators that will make computing its norm in the form of \Cref{eq:UB on op norm of KIJE} immediate. The first operator applies $G_R^\transpose:U_R\to \R^{[n]^R}$ in every summand of the decomposition from \Cref{eq:UJ otimes RnAj decomposed}:
    \begin{align*}
        \Psi_R:=\bigoplus_{\substack{\ell\in[n]^L \\ z\in[n]^Z}} \left(\Id_{U_C}\otimes G_R^\transpose\otimes\Id_{\R^{Y_{\ell,z}}}\right):U_J\otimes\R^{\maxAclosed{J}}\to\bigoplus_{\substack{\ell\in[n]^L \\ r\in[n]^R \\ z\in[n]^Z}} U_C\otimes \R^{Y_{\ell,z}}.
    \end{align*}
    Here, to get the codomain as above, we use the identification
    \begin{align*}
        \R^{[n]^R}\otimes\R^{Y_{\ell,z}} &\cong \bigoplus_{r\in[n]^R} \R^{Y_{\ell,z}}\\
        e_r\otimes e_b &\mapsto[e_b]_r,
    \end{align*}
    where $[e_b]_r$ denotes placing $e_b$ in the $r$-th summand. 

    Next, we define our second operator in parts by taking $\ell\in[n]^L$, $r\in[n]^R$ and $z\in[n]^Z$ and letting
    \begin{align*}
        T_{\ell,r,z}:=\Pi_{X_{r,z}} D_E \iota_{Y_{\ell,z}}:\R^{Y_{\ell,z}}\to\R^{X_{r,z}},
    \end{align*}
    where $\iota_{Y_{\ell,z}}:\R^{Y_{\ell,z}}\to\R^{[n]^E}$ is inclusion by 0's (so padding in the coordinates $[n]^E\setminus Y_{\ell,z}$) and $\Pi_{X_{r,z}}:\R^{[n]^E}\to\R^{X_{r,z}}$ is projection. The direct sum of all $T_{\ell,r,z}$'s gives us our second operator
    \begin{align}
        T_E:=\bigoplus_{\substack{\ell\in[n]^L \\ r\in[n]^R \\ z\in[n]^Z}} \left(\Id_{U_C}\otimes T_{\ell,r,z}\right):\bigoplus_{\substack{\ell\in[n]^L \\ r\in[n]^R \\ z\in[n]^Z}} U_C\otimes \R^{Y_{\ell,z}}\to\bigoplus_{\substack{\ell\in[n]^L \\ r\in[n]^R \\ z\in[n]^Z}} U_C\otimes \R^{X_{r,z}}\cong\bigoplus_{\substack{r\in[n]^R \\ z\in[n]^Z}} U_C\otimes  \R^{[n]^L} \otimes \R^{X_{r,z}}.\label{eq:T_E def}
    \end{align}
    
    Finally, our third operator is similar to the first and applies $G_L:\R^{[n]^L}\to U_L $ to each summand in the codomain of $T_E$ for all pairs $(r,z)$:
    \begin{align*}
        \Psi_L:=\bigoplus_{\substack{r\in[n]^R \\ z\in[n]^Z}} \left(\Id_{U_C}\otimes G_L \otimes \Id_{\R^{X_{r,z}}}\right):\bigoplus_{\substack{r\in[n]^R \\ z\in[n]^Z}} U_C\otimes  \R^{[n]^L} \otimes \R^{X_{r,z}}\to U_I\otimes \R^{\maxAIclosed},
    \end{align*}
    where we get the desired codomain by the decomposition in \Cref{eq:UI otimes RnAI decomposed}. All together, our proposed factorization of $K_{I,J}^E$ looks like
    \begin{align}
        U_J\otimes\R^{\maxAclosed{J}}
        \xrightarrow{\Psi_R}
        \bigoplus_{\substack{\ell\in[n]^L \\ r\in[n]^R \\ z\in[n]^Z}} U_C\otimes \R^{Y_{\ell,z}}
        \xrightarrow[]{T_E}
        \bigoplus_{\substack{\ell\in[n]^L \\ r\in[n]^R \\ z\in[n]^Z}} U_C\otimes \R^{X_{r,z}}
        \xrightarrow[]{\Psi_L}
        U_I\otimes\R^{\maxAIclosed}.\label{eq:proposed KIJE decomposition}
    \end{align}

    All that remains is to verify that this factorization is correct. To do so, we observe what \Cref{eq:proposed KIJE decomposition} does to pure tensors. Take $u_C\in U_C$, $u_R\in U_R$, and $e_{\ell,b,z}\in\R^{\maxAclosed{J}}$, the indicator vector at $(\ell,b,z)\in\maxAclosed{J}$. We then denote $[v]_{\ell,r,z}$ to be the vector $v$ placed in the $(\ell,r,z)$ summand of the domain or codomain of $T_E$ as defined in \Cref{eq:T_E def} to form the pure tensor $u_C\otimes u_R\otimes e_{\ell,b,z}\in U_J\otimes\R^{\maxAclosed{J}}$ (as shown by the decomposition in \Cref{eq:UJ otimes RnAj decomposed}) to compute
    \begin{align*}
        \Psi_R(u_C\otimes u_R\otimes e_{\ell,b,z}) &= \sum_{r\in[n]^R}(G_R^\transpose u_R)(r)[u_C\otimes e_b]_{\ell,r,z}\\
        T_E \Psi_R(u_C\otimes u_R\otimes e_{\ell,b,z}) &= \sum_{\substack{r\in[n]^R \\ a\in X_{r,z}}} (G_R^\transpose u_R)(r) D_E(a,b) [u_C\otimes e_a]_{\ell,r,z}\\
        \Psi_L T_E \Psi_R(u_C\otimes u_R\otimes e_{\ell,b,z}) &= \sum_{\substack{r\in[n]^R \\ a\in X_{r,z}}} (G_R^\transpose u_R)(r) D_E(a,b) (u_C\otimes G_Le_\ell\otimes e_{r,a,z}).\numberthis\label{eq:Psi_L T_E Psi_R expanded}
    \end{align*}
    From here, it is straightforward to see that \Cref{eq:Psi_L T_E Psi_R expanded} is the same operation as $K_{I,J}^E$ as defined in \Cref{eq:hatKIJE expression}. Since pure tensors span the domain, we can conclude that $K_{I,J}^E=\Psi_L T_E \Psi_R$. A quick computation on this factorization yields \Cref{eq:UB on op norm of KIJE}:
    \begin{align*}
        \opnorm{K_{I,J}^E}&=\opnorm{\Psi_L T_E \Psi_R}\\
        &\leq\opnorm{\Psi_L}\opnorm{T_E}\opnorm{\Psi_R} &\text{By submultiplicativity}\\
        &\leq \opnorm{G_L}\opnorm{\bigoplus_{\ell,r,z}\Id_{U_C}\otimes T_{\ell,r,z}}\opnorm{G_R^\transpose}\\
        &\leq\opnorm{G_L}\left(\max_{\ell,r,z}\opnorm{T_{\ell,r,z}}\right)\opnorm{G_R}, &\text{By \Cref{fact:opnorm of direct sum is max of opnorms,fact:opnorm of tensor is prod of opnorms}}
    \end{align*}
    as desired.
\end{proof}

We will need two standard facts about Gaussian matrices to prove \Cref{prop:UB on opnorm of rademacher matrix}. The first allows us to symmetrize the norm of a mean-zero matrix using standard Gaussian random variables.
\begin{lemma}[Lemma 6.7.4 of \cite{vershynin_high-dimensional_2018}]\label{lem:symmetrization with gaussians}
    Let $X_1,\dots,X_N$ be independent mean-zero random vectors in a normed space. Let $g_1,\dots,g_N\sim\cN(0,1)$ be independent Gaussian random variables which are also independent of $X_i$. Then 
    \begin{align*}
        \E\norm{\sum_{i=1}^N X_i}\leq 3\E\norm{\sum_{i=1}^N g_i X_i}. 
    \end{align*}
\end{lemma}
The second bounds the norm of a random Gaussian matrix.
\begin{lemma}[Theorem 7.3.1 of \cite{vershynin_high-dimensional_2018}]\label{thm:norms of gaussian random matrices}
    Let $A$ be a $k\times n$ matrix with independent $\cN(0,1)$ entries. Then
    \begin{align*}
        \E\opnorm{A}&\leq\sqrt{k}+\sqrt{n}.
    \end{align*}
\end{lemma}
We can now prove \Cref{prop:UB on opnorm of rademacher matrix}.
\UbonNormRademacherMatrix*
\begin{proof}[Proof of \Cref{prop:UB on opnorm of rademacher matrix}]
    For a row $r\in[k]$ and column $c\in[n]$, let $X_{r,c}\in\R^{k\times n}$ be the matrix of all 0's except with a 1 in the $(r,c)$ position. Taking $\widetilde{G}$ to be the un-normalized version of $G$, we have
    \begin{align*}
        \widetilde{G}:=\sqrt{k}G=\sum_{r\in[k],c\in[n]}\sigma_{r,c}X_{r,c}.
    \end{align*}
    Applying \Cref{lem:symmetrization with gaussians} to the independent mean-zero matrices $\sigma_{r,c}X_{r,c}$, we have that
    \begin{align}
        \E_{\sigma}\left[\opnorm{\widetilde{G}}\right]\leq3\E_{\sigma,\gamma}\left[\opnorm{\sum_{r\in[k],c\in[n]} \gamma_{r,c}\sigma_{r,c} X_{r,c}}\right]\label{eq:versh 6.7.4}
    \end{align}
    where the $\gamma_{r,c}\sim\cN(0,1)$ are jointly independent and independent of the signs. Thus, the products $\gamma_{r,c}\sigma_{r,c}$ are jointly independent standard Gaussians, so we can use \Cref{thm:norms of gaussian random matrices} to get
    \begin{align}
        \E_{\sigma,\gamma}\left[\opnorm{\sum_{r\in[k],c\in[n]} \gamma_{r,c}\sigma_{r,c} X_{r,c}}\right]\leq\sqrt{k}+\sqrt{n}. \label{eq: versh 7.3.1}
    \end{align}
    Putting \Cref{eq:versh 6.7.4,eq: versh 7.3.1} together yields
    \begin{align*}
        \E_{\sigma}[\opnorm{G}]&=\frac{1}{\sqrt{k}}\E_\sigma\left[\opnorm{\widetilde{G}}\right]\\
        &\leq\frac{1}{\sqrt{k}}3\left(\sqrt{k}+\sqrt{n}\right)\\
        &\leq 3\left(1+\rho^{-1/2}\right),
    \end{align*}
    as desired.
\end{proof}

\dobib

\section{Bounding the Operator Norm of Sparse Tensor Gram Matrices}\label{sec:bounding max opnorm with wigner}

In this section, we will prove the standalone proposition regarding the expected max operator norm over a sufficiently small family of sparse operators formed by the tensor of operators with random Rademacher entries. Since the main result of this section should be thought of as an independent result on the operator norm of such operators, we will redefine the objects we're working on. 
\paragraph{Setup} Throughout this section, we will use the following notation. Let $D>0$ be a fixed integer, $q\geq1$ and $\nu\geq0$ be fixed integers such that $q+\nu\leq D$, $\rho\in(0,1)$ be constant, and for any $n\in\N$ and $i\in[D]$ let $k_i\in[n]$ be such that $k_i\geq\rho n$. With these parameters, we define $q$ random matrices $G_i\in\R^{k_i\times n}$ as $(G_i)_{r,a}=k_i^{-1/2}\sigma^{(i)}_{r,a}$ for $r\in[k_i]$ and $a\in[n]$ where all $\sigma^{(i)}_{r,a}\sim\{-1,1\}$ are sampled independently. For $i\in[q]$, we let $D_i=G_i^\transpose G_i-\Id_{\R^{[n]}}$. Lastly, for a set $X\subseteq[n]^q$, we let $\Pi_X:\R^{[n]^q}\to\R^{X}$ be the standard projection map and $\iota_X:\R^X\to\R^{[n]^q}$ be the standard inclusion map by 0's.

\begin{restatable}[\Cref{prop:UB on op norm of max T_lambda-verbose} restated]{proposition}{maxopUB}
\label{prop:UB on op norm of max T_lambda}
    Take fixed integers $\nu\geq 0$ and $q\geq 1$ such that $\nu+q\leq D$. Let $\Lambda$ be a fixed indexing set such that $\abs{\Lambda}\leq n^\nu$.
    For all indices $\lambda\in\Lambda$, take $X_\lambda,Y_\lambda\subseteq[n]^q$ to be $\eta$-sparse for some constant $\eta\in(0,1)$, fixed independently of the signs. Letting
    \begin{align*}
        T_\lambda&= \Pi_{X_\lambda} \left(\bigotimes_{j=1}^q D_j\right)\iota_{Y_\lambda}:\R^{Y_\lambda}\to \R^{X_\lambda}, 
    \end{align*}
    we have that for $n\geq N_{D,q,\eta}:=\max\left\{16\left(\frac{192D}{e}\right)^{24}, \eta^{-2/q}\right\}$,
    \begin{align*}
        \E[\max_{\lambda\in\Lambda}\opnorm{T_\lambda}]&\leq 2\left(\frac{4}{\rho}\right)^q\sqrt{\eta}.
    \end{align*}
    If $\Lambda=\varnothing$, then this expectation is just 0. 
\end{restatable}

\subsection{Derivation of the sign monomials}
Our proof of \Cref{prop:UB on op norm of max T_lambda} will rely on the trace power method, so we begin by deriving an explicit expression for the trace of the power of these $T_\lambda$'s, which we will eventually decompose into monomials of the signs $\sigma^{(j)}_{r,a}$.
\begin{lemma}\label{lem:trace of power of TT^transpose}
    For $X,Y\subseteq[n]^q$ and $D_1,\dots,D_q$ as above and 
    \begin{align*}
        T:=T_{X,Y}:=\Pi_{X}\left(\bigotimes_{j=1}^q D_j\right)\iota_Y : \R^Y\to \R^X,
    \end{align*}
    we have that for all integers $p\geq 1$:
    \begin{align*}
        \trace\left((TT^\transpose)^p\right)=\sum_{\substack{x_1,\dots,x_p\in X\\y_1,\dots,y_p\in Y}}\prod_{j=1}^q\prod_{t=1}^p D_j(x_{t,j},y_{t,j}) D_j(y_{t,j},x_{t+1,j})
    \end{align*}
    where we define $x_{p+1}:=x_1$. 
\end{lemma}
\begin{proof}
    For $x\in X$ and $y\in Y$ we have $T(x,y)=\prod_{j=1}^qD_j(x_j,y_j)$. Using the fact that each $D_j$ is symmetric, we get that for $x,x'\in X$,
    \begin{align*}
        (TT^\transpose)(x,x')&=\sum_{y\in Y} T(x,y)T(x',y)\\
        &=\sum_{y\in Y}\prod_{j=1}^q D_j(x_j,y_j)D_j(y_j,x'_j).\numberthis\label{eq:TT^transpose expanded}
    \end{align*}
    Thus, taking the trace of $(TT^\transpose)^p$ yields
    \begin{align*}
        \trace\left((TT^\transpose)^p\right)&=\sum_{x_1,\dots,x_p\in X}\prod_{t=1}^p(TT^\transpose)(x_t,x_{t+1})\\
        &=\sum_{x_1,\dots,x_p\in X}\prod_{t=1}^p\sum_{y\in Y}\prod_{j=1}^q D_j (x_{t,j},y_j)D_j(y_j,x_{t+1,j}) &\text{By \Cref{eq:TT^transpose expanded}}\\
        &=\sum_{\substack{x_1,\dots,x_p\in X\\y_1,\dots,y_p\in Y}}\prod_{j=1}^q\prod_{t=1}^p D_j(x_{t,j},y_{t,j}) D_j(y_{t,j},x_{t+1,j}),
    \end{align*}
    as claimed.
\end{proof}

Expanding out the result of \Cref{lem:trace of power of TT^transpose} to contain the individual signs $\sigma_{r,a}^{(j)}$ will let us eventually count how many terms in this sum of products could possibly be 1. 
\begin{corollary}\label{cor:tr(TT^transpose) as sum of prod of monomial signs}
    With the setup of \Cref{lem:trace of power of TT^transpose}, 
    \begin{align*}
        \trace\left((TT^\transpose)^p\right)=\sum_{\substack{x_1,\dots,x_p\in X\\y_1,\dots,y_p\in Y}}\sum_{\substack{r_j\in[k_j]^{2p}\\j\in[q]}}\left(\prod_{j=1}^q k_j^{-2p}\right)\left(\prod_{j=1}^q\Phi_j(x,y,r_j)\right),
    \end{align*}
    where we define the monomials of signs
    \begin{align}
        \Phi_j(x,y,r_j)=\prod_{t=1}^p\left(\sigma^{(j)}_{r_{j,2t-1},x_{t,j}}\cdot\sigma^{(j)}_{r_{j,2t-1},y_{t,j}}\cdot\sigma^{(j)}_{r_{j,2t},y_{t,j}}\cdot\sigma^{(j)}_{r_{j,2t},x_{t+1,j}}\cdot\Ind\{x_{t,j}\neq y_{t,j}\text{ and } y_{t,j}\neq x_{t+1,j}\}\right).\label{def:Phi_j sign monomial}
    \end{align}
\end{corollary}
\begin{proof}
    We use the definition of each $D_j$ to directly compute
    \begin{align*}
        &\ \prod_{t=1}^p D_j(x_{t,j},y_{t,j}) D_j(y_{t,j},x_{t+1,j})\\
        =&\ \prod_{t=1}^p\left[\left(\frac{1}{k_j}\sum_{r=1}^{k_j}\sigma^{(j)}_{r,x_{t,j}}\cdot\sigma^{(j)}_{r,y_{t,j}}\cdot\Ind\{x_{t,j}\neq y_{t,j}\}\right)\left(\frac{1}{k_j}\sum_{r=1}^{k_j}\sigma^{(j)}_{r,y_{t,j}}\cdot\sigma^{(j)}_{r,x_{t+1,j}}\cdot\Ind\{y_{t,j}\neq x_{t+1,j}\}\right)\right]\\
        =&\ k_j^{-2p}\sum_{r_{j,1},\dots,r_{j,2p}\in[k_j]}\prod_{t=1}^p\left(\sigma^{(j)}_{r_{j,2t-1},x_{t,j}}\cdot\sigma^{(j)}_{r_{j,2t-1},y_{t,j}}\cdot\sigma^{(j)}_{r_{j,2t},y_{t,j}}\cdot\sigma^{(j)}_{r_{j,2t},x_{t+1,j}}\cdot\Ind\{x_{t,j}\neq y_{t,j}\text{ and } y_{t,j}\neq x_{t+1,j}\}\right)\\
        =&\ k_j^{-2p}\sum_{r_{j,1},\dots,r_{j,2p}\in[k_j]}\Phi_j(x,y,r_j).\numberthis\label{eq:Phi_j derivation}
    \end{align*}
    Plugging in \Cref{eq:Phi_j derivation} to the result of \Cref{lem:trace of power of TT^transpose} gives us the desired claim.
\end{proof}
\subsection{Graph view of the sign monomials}
With the decomposition of $\trace\left((TT^\transpose)\right)$ into our sign monomials from \Cref{def:Phi_j sign monomial}, the next natural step is to bound how many of these monomials could possibly have non-zero expectation. Our main tool to achieve this is to view the list of indices $(x,y,r_j)$ as a walk along a graph, specifically a weak Wigner word from \cite{wigner_characteristic_1955}, and use the tools from \cite{furedi_eigenvalues_1981,anderson_clt_2006} to bound how many monomials could have expectation 1. Overall, the idea is to remove all monomials with expectation 0, split the remaining monomials into equivalence classes, bound the size of each equivalence class using the sparsity of our indexing sets, and finally bound the number of equivalence classes using bounds on the number of equivalence classes of weak Wigner words from \cite{anderson_clt_2006}.

To motivate this connection to weak Wigner words, we begin by repackaging the information contained in $(x,y,r_j)$ into a single tuple which we will interpret as a walk along a bipartite graph.
\begin{definition}\label{def:w_j from x y and r_j}
    For a fixed $j\in[q]$ and $x\in X^p$, $y\in Y^p$, and $r_j\in[k_j]^{2p}$, define the string $w_j=(w_{j,0},\dots,w_{j,4p})$ as
    \begin{align*}
        w_{j,4t-4}&=(\sL,x_{t,j}) &w_{j,4t-3}&=(\sR,r_{j,2t-1})\\
        w_{j,4t-2}&=(\sL,y_{t,j})  &w_{j,4t-1}&=(\sR,r_{j,2t})\\
        w_{j,4t}&=(\sL, x_{t+1,j})
    \end{align*}
    for all $t\in[p]$ where $\sL$ and $\sR$ are formal symbols indicating ``left'' and ``right''. We will view this as a walk on the bipartite graph $\Gamma_{w_j}=(\sV=\sV_\sL\sqcup\sV_\sR,\sE_j)$ where $\sV_\sL=\{w_{j,4t-4},w_{j,4t-2},w_{j,4t}:t\in[p]\}\subseteq \{\sL\}\times[n]$ and $\sV_\sR=\{w_{j,4t-3},w_{j,4t-1}:t\in[p]\}\subseteq\{\sR\}\times[k_j]$, and the distinct undirected edges are $\sE_j=\{\{w_{j,t},w_{j,t+1}\}:0\leq t<4p\}$.
\end{definition}
\begin{remark}\label{rmk:xyr same as w1-wq}
    It's useful to observe that, via this construction, there is a one-to-one correspondence between the $q$-tuple $(w_1,\dots,w_q)$ and the list $(x,y,r)$ for $r\in\prod_{j=1}^q[k_j]^{2p}$. That is, these are just different ways of indexing the same $q\times 4p$ elements, which we will more clearly see in the proof of \Cref{lem:UB on size of product of equivalence classes}.
\end{remark}
Next, we make a simple observation about when $\Phi_j$ is non-zero.
\begin{lemma}\label{lem:Phi non-zero iff non-left-backtracking and closed walk}
    With the setup from \Cref{def:w_j from x y and r_j}, $\Phi_j(x,y,r_j)$ is non-zero if and only if the walk $w_j$ is a non-left-backtracking closed walk.\footnote{By non-left-backtracking, we mean that at each visit to a right vertex, the preceding and following left vertices differ. Backtracking at left vertices is allowed.}
\end{lemma}
\begin{proof}
From the definition of $\Phi_j(x,y,r_j)$ in \Cref{def:Phi_j sign monomial}, we see that $\Phi_j(x,y,r_j)$ is non-zero if and only if $x_{t,j}\neq y_{t,j}$ and $y_{t,j}\neq x_{t+1,j}$ for all $t\in[p]$, meaning that $w_{j,4t-4}\neq w_{j,4t-2}\neq w_{j,4t}$. This gives us that the walk is non-left-backtracking.
To see that the walk is closed, we recall from \Cref{lem:trace of power of TT^transpose} that we defined $x_{p+1}=x_1$, giving us $w_{j,0}=(\sL,x_{1,j})=(\sL,x_{p+1,j})=w_{j,4p}$
\end{proof}
In this setting, we can express $\Phi_j(x,y,r_j)$ directly by the edges of $\Gamma_{w_j}$.
\begin{lemma}
    For a non-left-backtracking closed walk $w_j$ from $(x,y,r_j)$, we can compute $\Phi_j(x,y,r_j)$ from $w_j$ as
    \begin{align}
        \Phi_j(x,y,r_j)=\prod_{\substack{r\in[k_j]\\ a\in[n]}}\left(\sigma^{(j)}_{r,a}\right)^{m_j(r,a)},\label{eq:Phi(x y r_j) as product of signs with multiplicity powers}
    \end{align}
    where $m_j(r,a)$ is the number of times $w_j$ traverses the edge $\{(\sL,a),(\sR,r)\}$ in either direction.
\end{lemma}
\begin{proof}
    Since $w_j$ is a non-left-backtracking closed walk, the indicators  $\Ind\{x_{t,j}\neq y_{t,j}\text{ and } y_{t,j}\neq x_{t+1,j}\}$ from \Cref{def:Phi_j sign monomial} are always 1, so $\Phi_j(x,y,r_j)$ is always just a product of signs. The expression in \Cref{eq:Phi(x y r_j) as product of signs with multiplicity powers} then exactly recovers the signs used in $\Phi_j(x,y,r_j)$ because each occurrence is counted by a traversal of the corresponding edge.
\end{proof}
An immediate consequence of this formulation of $\Phi_j$ is an easy expression for its expectation.
\begin{corollary}\label{cor:expectation of Phi_j is 1 iff exponents are even}
    For a non-left-backtracking closed walk $w_j$, we have
    \begin{align*}
        \E_\sigma[\Phi_j(x,y,r_j)]=
        \begin{cases}
            1 & m_j(r,a)\text{ is even for all $r\in[k_j], a\in[n]$}\\
            0 & \text{otherwise.}
        \end{cases}
    \end{align*}
\end{corollary}
\begin{proof}
    If at least one of the $m_j(r,a)$'s is odd, then since all of the signs $\sigma^{(j)}_{r,a}$ are independent Rademacher variables, the whole product is randomized as a Rademacher variable and has expectation 0. If all the exponents are even, then $\Phi_j(x,y,r_j)=1$ deterministically. 
\end{proof}
Taking the expectation over all of $(TT^\transpose)^p$ then gives us a nice expression.
\begin{corollary}\label{cor:Expectation of trace (TT^transpose)^p in terms of even non-left-backtracking cycles}
    Randomizing over the signs, the expectation of $\trace\left((TT^\transpose)^p\right)$ is
    \begin{align*}
        \E_\sigma\left[\trace\left((TT^\transpose)^p\right)\right]=\left(\prod_{j=1}^q k_j^{-2p}\right)\cdot\abs{\scW}
    \end{align*}
    where
    \begin{align*}
        \scW(X,Y)=\left\{(x,y,r)\in X^p\times Y^p\times \prod_{j=1}^q[k_j]^{2p} : \forall j\in[q],\ w_j\text{ is even and non-left-backtracking}\right\}
    \end{align*}
\end{corollary}
Therefore, to bound the expectation of $\trace((TT^\transpose)$, we want to bound the number of monomials $\Phi_j(x,y,r_j)$ where all exponents are even. Equivalently, we want to bound the number of possible non-left-backtracking closed walks $w_j$ that traverse each edge of $\Gamma_{w_j}$ an even number of times---the cardinality of $\scW$. This brings us to weak Wigner words.

\subsection{Weak Wigner words}
One can naturally view a walk over a graph as just a list enumerating the labels of the vertices visited along that walk. That intuition forms the basis for the definition of words and sentences in the context of Wigner words.

\begin{definition}[Words and sentences][\cite{wigner_characteristic_1955,anderson_clt_2006}]
A \emph{word} $w$ is a non-empty finite sequence of letters from some alphabet, where we view letters as length-1 words. We call $w$ \emph{closed} if the first and last letters are the same. A \emph{sentence} $a$ is a non-empty sequence of words, so we may view words as one-word sentences. 
\end{definition}
There are many properties of words and sentences that we will work with.
\begin{definition}
    For a word $w$ and sentence $a$, we define:
    \begin{itemize}
        \item The \emph{length} $\length(w)$ is the number of letters in $w$,
        \item The \emph{support} $\supp(a)$ is the set of letters appearing in $a$, and
        \item The \emph{weight} $\wt(a)$ is the cardinality of $\supp(a)$.
    \end{itemize}
\end{definition}
Given a word, we can translate back and forth between its associated graph.
\begin{definition}[Graph of a word]
    Given a word $w=\alpha_1\cdots\alpha_p$, we define its \emph{graph} $\Gamma_w=(\sV,\sE)$ to have vertices $\sV=\supp(w)$ and distinct undirected edges $\sE=\{\{\alpha_i,\alpha_{i+1}\}:i\in[p-1]\}$.
\end{definition}

We can now finally define weak Wigner words.
\begin{definition}[weak Wigner word]
    We say that a word $w$ is a \emph{weak Wigner word} if $w$ is closed and traverses each edge of $\Gamma_w$  at least twice.
\end{definition}

\subsubsection{Equivalence and canonical representatives}
Given that words define walks on graphs, it is natural to define an equivalence relation between them encoding the shape of the walk.

\begin{definition}[Equivalence]
    Two words $w$ and $w'$ are \emph{equivalent}, denoted as $w\sim w'$, if $\length(w)=\length(w')$ and there exists a bijection $\psi:\supp(w)\to\supp(w')$ such that $\psi(w_i)=w_i'$ for all $i\in[\length(w)]$. 
\end{definition}

\begin{remark}
    It is worthwhile to note that $w\sim w'$ isn't the same as the graphs $\Gamma_w$ and $\Gamma_{w'}$ being isomorphic. For example, $w=\{12321\}$ and $w'=\{12131\}$ have isomorphic graphs but aren't equivalent. On the other hand, if $w\sim w'$ then $\Gamma_w$ and $\Gamma_{w'}$ are clearly isomorphic where the isomorphism is given by the relabeling of their letters, $\psi$. 
\end{remark}

With an equivalence relation, we can talk about the equivalence class $\cW$ of a word $w$. Intuitively, an equivalence class encodes the shape of a walk, but not the actual values of the vertices. Since we will need to reference the shape of the walk of an equivalence class, we define the canonical representative of an equivalence class---this is not necessarily a word generated as in \Cref{def:w_j from x y and r_j}, but it clearly encodes the shape of the path of the elements of $\cW$.
\begin{definition}[Canonical representative]\label{def:canonical representative}
    Let $\cW$ be an equivalence class and $w=(w_1,\dots,w_m)\in\cW$. Define $\psi_w:\supp(w)\to[\wt(w)]$ by ordering the distinct letters by their first appearance in $w$. The canonical representative of $\cW$ is the word $\overline{w}=(\psi_w(w_1),\dots,\psi_w(w_m))$.
\end{definition}
\begin{example}
    For the word $w=\{421314214\}$, the canonical representative of its equivalence class is $\overline{w}=\{123431231\}$.
\end{example}

It is easy to check that this notion of canonical representative is well-defined.
\begin{fact}
    If $w\sim w'$, then their canonical representatives satisfy $\overline{w}=\overline{w}'$.
\end{fact}
\begin{proof}
    Equivalence preserves the equality of letters at any two positions. Thus first appearances occur at the same positions in $w$ and $w'$, and every repeated letter receives the same canonical label in both words.
\end{proof}

\subsection{Weak Wigner words in the bipartite setting}
Since the words we will work on are generated as in \Cref{def:w_j from x y and r_j}  and represent walks on a bipartite graph, we will introduce some additional definitions and consequences in this context. We call a word $w$ \emph{bipartite} if $\Gamma_w$ is a bipartite graph.

\begin{definition}[Properties of bipartite walks]
    Given a bipartite walk $w$ on $\Gamma_w=(\sV=\sV_\sL\sqcup\sV_\sR,\sE)$, we define
    \begin{itemize}
        \item The left weight $\wt_\sL(w)=\abs{\sV_\sL}$, and
        \item The right weight $\wt_\sR(w)=\abs{\sV_\sR}$.
    \end{itemize}
\end{definition}
Moreover, we can measure the loss of vertices from the maximum possible weight of a weak Wigner word from \Cref{def:w_j from x y and r_j}. Such words have length $4p+1$ and traverse each edge at least twice, so their connected graphs have at most $2p+1$ vertices.
\begin{definition}[Deficiency]
    For a weak Wigner word $w$ on $\Gamma_w=(\sV=\sV_\sL\sqcup\sV_\sR,\sE)$ constructed as in \Cref{def:w_j from x y and r_j}, we define its \emph{deficiency} to be
    \begin{align*}
        \Delta(w)=2p+1-\wt(w).
    \end{align*}
\end{definition}

Lastly, we define a specific notion of equivalence on the words from \Cref{def:w_j from x y and r_j} and show that it agrees with the standard notion of equivalence.
\begin{definition}[Bipartite equivalence]
    Given two walks $w$ on $\Gamma_w=(\sV=\sV_\sL\sqcup\sV_\sR,\sE)$ and $w'$ on $\Gamma_{w'}=(\sV'=\sV'_\sL\sqcup\sV'_\sR,\sE')$ of equal length, we say that $w$ is \emph{bipartite-equivalent} to $w'$ if there exist bijections $\psi_\sL:\sV_\sL\to\sV'_\sL$ and $\psi_\sR:\sV_\sR\to\sV'_\sR$ such that $\psi_{\sL}(w_i)=w'_i$ for all even $i$, and $\psi_\sR(w_i)=w'_i$ for all odd $i$. This is denoted as $w\bisim w'$. 
\end{definition}

\begin{fact}\label{fact:bipartite equiv and equiv are the same for bipartite words}
    For two walks $w$ on $\Gamma_w=(\sV=\sV_\sL\sqcup\sV_\sR,\sE)$ and $w'$ on $\Gamma_{w'}=(\sV'=\sV'_\sL\sqcup\sV'_\sR,\sE')$ of equal length produced as in \Cref{def:w_j from x y and r_j}, 
    \begin{align*}
        w\bisim w' \Longleftrightarrow w\sim w'.
    \end{align*}
\end{fact}
\begin{proof}
    The forward direction is immediate since we just define $\psi:\sV\to\sV'$ to act on $\sV_\sL$ via $\psi_\sL$ and on $\sV_\sR$ via $\psi_\sR$. The reverse direction comes from noticing that in \Cref{def:w_j from x y and r_j}, it must be that each even index of $w$ is an element of $\sV_\sL$ and every odd index is an element of $\sV_\sR$, and similarly for $w'$. Thus, for an even index $i$, we have that $\psi(w_i)=\psi((\sL,x))=(\sL,x')=w_i'$. Consequently, on even indices $\psi$ restricts to a bijection from $\sV_\sL$ to $\sV'_\sL$, which allows us to define $\psi_\sL$. Similarly, on odd indices $\psi$ restricts to a bijection from $\sV_\sR$ to $\sV'_\sR$, allowing us to define $\psi_\sR$. This gives us the necessary bipartite-equivalence of $w$ and $w'$. 
\end{proof}

\subsubsection{Equivalence classes from a specified deficiency}
Since all words in an equivalence class have the same weight and length, which are the only two parameters used in the deficiency of a word, they have the same deficiency as well. Thus, we can parametrize a set of equivalence classes by only their deficiency. 

\begin{definition}\label{def:admissible word-class realizations}
    For integers $p\geq 1$ and $\Delta\geq 0$, we let $\fW_p(\Delta)$ be the set of equivalence classes of even non-left-backtracking closed walks with length $4p+1$, deficiency $\Delta$, and left and right positions as in \Cref{def:w_j from x y and r_j}. For a tuple of equivalence classes $(\cW_1,\dots,\cW_q)$, define its admissible part by
    \begin{align*}
        \scW_{X,Y}(\cW_1,\dots,\cW_q)=\left\{(x,y,r)\in\scW(X,Y):w_j(x,y,r_j)\in\cW_j\text{ for all }j\in[q]\right\}.
    \end{align*}
\end{definition}
With this notation, we can cleanly decompose $\scW$. 

\begin{lemma}\label{lem:scW disjoint decomposition}
    \begin{align*}
        \scW(X,Y)=\bigsqcup_{(\Delta_1,\dots,\Delta_q)\in\{0,\dots,2p\}^q}\bigsqcup_{(\cW_1,\dots,\cW_q)\in\fW_p(\Delta_1)\times\cdots\times\fW_p(\Delta_q)}\scW_{X,Y}(\cW_1,\dots,\cW_q).
    \end{align*}
    Here, each $\cW_j$ is one of the equivalence classes of $\fW_p(\Delta_j)$. 
\end{lemma}
\begin{proof}
    Each $(x,y,r)\in\scW(X,Y)$ determines a unique equivalence class $\cW_j$ for each $w_j(x,y,r_j)$ and a unique deficiency $\Delta_j\in\{0,\dots,2p\}$. Thus it belongs to exactly one of the admissible parts, while every member of an admissible part belongs to $\scW(X,Y)$ by definition.
\end{proof}

\subsection{Bounding the number of equivalence classes}
The first step in bounding the number of non-left-backtracking closed walks with even multiplicity on each edge is to count the number of equivalence classes of such closed walks.
\begin{proposition}[Proposition 7.9 of \cite{anderson_clt_2006}]\label{prop: AZ04 prop 7.9}
    For all positive integers $N$ and $W$ satisfying $N\geq 2W-2$, there are at most
    \begin{align*}
        2^N N^{3(N-2W+2)}
    \end{align*}
    equivalence classes of weak Wigner words of length $N+1$ and weight $W$. 
\end{proposition}

\begin{corollary}\label{cor: prop 7.9 for bipartite equiv}
    For all positive integers $N$ and $W$ satisfying $N\geq 2W-2$, there are at most
    \begin{align*}
        2^N N^{3(N-2W+2)}
    \end{align*}
    bipartite-equivalence classes of weak Wigner words produced as in \Cref{def:w_j from x y and r_j} of length $N+1$ and weight $W$. 
\end{corollary}
\begin{proof}
    The parity argument in \Cref{fact:bipartite equiv and equiv are the same for bipartite words} applies to arbitrary left and right labels, so ordinary and bipartite equivalence agree for these words. We then apply \Cref{prop: AZ04 prop 7.9}.
\end{proof}

\begin{fact}
    For a non-left-backtracking closed walk $w$ on $\Gamma_w=(\sV=\sV_\sL\sqcup\sV_\sR,\sE)$ from \Cref{def:w_j from x y and r_j} such that every edge is traversed an even number of times, we have
    \begin{align}
        \wt(w)&\leq 2p+1\label{eq:wt <= 2p+1}\\
        \Delta(w)&\geq 0\label{eq:Delta(w) >=0}\\
        \Delta(w)&\leq 2p\\
        \wt_\sR(w)&\leq p\label{eq:wt_sR <= p}\\
        \wt_\sL(w)&\geq p+1-\Delta(w)\label{eq:wt_sL >= p+1-Delta}.
    \end{align}
\end{fact}
\begin{proof}
    Let $e$ be the number of distinct undirected edges traversed by $w$. By construction, the walk makes $4p$ edge traversals. Since each edge is traversed at least twice, $e\leq 2p$. The graph traced by the walk is connected, so $\wt(w)\leq e+1\leq 2p+1$, giving \Cref{eq:wt <= 2p+1}. Consequently, $\Delta(w)=2p+1-\wt(w)\geq 0$, yielding \Cref{eq:Delta(w) >=0}. Since $w$ is nonempty, $\wt(w)\geq 1$ also gives $\Delta(w)\leq 2p$. Each right vertex is entered from and exited to distinct left vertices by the non-left-backtracking condition, so it has at least two distinct neighbors. Counting each distinct edge at its right endpoint gives $2\wt_\sR(w)\leq e\leq 2p$, so $\wt_\sR(w)\leq p$. Finally, $\wt_\sL(w)=2p+1-\Delta(w)-\wt_\sR(w)\geq p+1-\Delta(w)$, giving the final two inequalities.
\end{proof}

This gives us a convenient rewording of Proposition 7.9 of \cite{anderson_clt_2006}. 
\begin{corollary}
    For an integer $p\geq 1$ and deficiency $\Delta\geq 0$, the number of equivalence classes  of even non-left-backtracking closed walks from \Cref{def:w_j from x y and r_j} is at most
    \begin{align*}
        \abs{\fW_p(\Delta)}\leq 16^p(4p)^{6\Delta}.
    \end{align*}
\end{corollary}
\begin{proof}
    If $\Delta(w)>2p$, then $\wt(w)=2p+1-\Delta(w)\leq0$, which would give us an empty word, so we may assume that $\Delta(w)\in\{0,1,\dots,2p\}$ since $\Delta(w)\geq 0$ by \Cref{eq:Delta(w) >=0}. Then we set the variables in \Cref{cor: prop 7.9 for bipartite equiv} as $N=\length(w)-1=4p$ and $W=2p+1-\Delta(w)$. The requirement of \Cref{cor: prop 7.9 for bipartite equiv} is satisfied since $N-(2W-2)=2\Delta(w)\geq 0$, so there are at most
    \begin{align*}
        2^N N^{3(N-2W+2)}&=2^{4p} (4p)^{3(4p-2(2p+1-\Delta(w))+2)}\\
        &=16^p (4p)^{6\Delta(w)}
    \end{align*}
    equivalence classes of non-left-backtracking closed walks from \Cref{def:w_j from x y and r_j} with even multiplicity on each edge.
\end{proof}

\subsection{Bounding the size of an equivalence class of a word using sparsity}
Now that we've bounded the number of equivalence classes of words arising from \Cref{def:w_j from x y and r_j} with even multiplicity on each edge, we move on to bounding the size of their admissible part, the term $\scW_{X,Y}(\cW_1,\dots,\cW_q)$ showing up in \Cref{lem:scW disjoint decomposition}. This is the key step where the sparsity of the indexing sets is used.

\begin{lemma}\label{lem:UB on size of product of equivalence classes}
    For a constant $\eta\in(0,1)$ and $\eta$-sparse subsets $X,Y\subseteq[n]^q$, take integers $\Delta_1,\dots,\Delta_q\in\{0,\dots,2p\}$. Then, for any $(\cW_1,\dots,\cW_q)\in\fW_p(\Delta_1)\times\cdots\times\fW_p(\Delta_q)$, we have that
    \begin{align}
        \abs{\scW_{X,Y}(\cW_1,\dots,\cW_q)}\leq n^{W_\sL}\eta^{W_\sL/q}\prod_{j=1}^q k_j^{\wt_\sR(\cW_j)}\label{eq:UB on cW_1 times ... times cW_q}
    \end{align}
    where $\wt_\sL(\cW_j)=\wt_{\sL}(w_j)$ for any $w_j\in\cW_j $ since all elements in the same equivalence class have the same weight, $\wt_\sR(\cW_j)$ is similar, and 
    \begin{align*}
        W_\sL=\sum_{j=1}^q \wt_\sL(\cW_j) \qquad\text{and}\qquad W_\sR=\sum_{j=1}^q \wt_\sR(\cW_j).
    \end{align*}
    
\end{lemma} 
\begin{proof}
    If $\scW_{X,Y}(\cW_1,\dots,\cW_q)$ is empty, the bound is immediate. Otherwise, consider a tuple $(x,y,r)\in\scW_{X,Y}(\cW_1,\dots,\cW_q)$ and its words $w_1,\dots,w_q$. Since $w_j$ has $\wt_\sR(w_j)\leq k_j$ distinct right vertices and there are $k_j$ possible right vertices to choose from, there are $k_j(k_j-1)\cdots(k_j-\wt_\sR(w_j)+1)\leq k_j^{\wt_\sR(w_j)}$ possible labelings of right vertices for $w_j$. Over all $j\in[q]$, this means there are at most
    \begin{align}
        \prod_{j=1}^q k_j^{\wt_\sR(\cW_j)}\label{eq:number of possible right labelings}
    \end{align}
    many assignments of all the right vertices of $w_1,\dots,w_q$. 

    For the left vertices, we must use the sparsity of $X$ and $Y$. This is where it becomes apparent why we must count admissible tuples rather than choose each $w_j$ separately: fixing one $w_j$ fixes only the $j$-th coordinates of all $x_t,y_t$, together with $r_j$. The remaining coordinates must still be chosen so that each complete $x_t$ belongs to $X$ and each complete $y_t$ belongs to $Y$. Put more formally, we will write out the entries of $x\in X^p$ and $y\in Y^p$  as a matrix. Let $Z\in[n]^{2p\times q}$ be defined as
    \begin{align*}
        Z&=\left[
        \begin{matrix}
            - & x_1 & -\\
            - & y_1 & -\\
            - & x_2 & -\\
            - & y_2 & -\\
              & \vdots & \\
            - & x_p & -\\
            - & y_p & -
        \end{matrix}
        \right]
        =\left[
        \begin{matrix}
            - & Z_1 & -\\
            - & Z_2 & -\\
            - & Z_3 & -\\
            - & Z_4 & -\\
              & \vdots & \\
            - & Z_{2p-1} & -\\
            - & Z_{2p} & -
        \end{matrix}
        \right]\\
        &=\left[
        \begin{matrix}
            z_{1,1} & z_{1,2} & \cdots & z_{1,q}\\
            z_{2,1} & z_{2,2} & \cdots & z_{2,q}\\
            \vdots & \vdots & \ddots & \vdots\\
            z_{2p,1} & z_{2p,2} & \cdots & z_{2p,q}
        \end{matrix}
        \right]\\
        &=
        \left[
        \begin{matrix}
            \vert & \vert & & \vert\\
            w_1^\sL & w_2^\sL  & \cdots & w_q^\sL\\
            \vert & \vert & & \vert
        \end{matrix}
        \right].
    \end{align*}
    where for $t\in[p]$ and $j\in[q]$ we are letting $z_{2t-1,j}=x_{t,j}$ and $z_{2t,j}=y_{t,j}$ and also removing the ``$\sL$'' label to just get elements of $[n]$. Here $w_j^\sL$ is the column vector of these $2p$ left entries of $w_j$, with the repeated last entry omitted. Now, it is clear that each row of $Z$ is an element of $X$ (if it is at an odd index) or $Y$ (if it is at an even index), and that each column of $Z$ gives the remaining information to define $w_j(x,y,r_j)$ since we have already fixed $r_j$. Consequently, once we have fixed $r\in\prod_{j=1}^q [k_j]^{2p}$, counting the members of $\scW_{X,Y}(\cW_1,\dots,\cW_q)$ with this $r$ reduces to counting the number of possible matrices $Z$ with each odd row in $X$ and each even row in $Y$, via \Cref{rmk:xyr same as w1-wq}. However, each column must also respect the shape of the path given by the equivalence classes $\cW_1,\dots,\cW_q$. To enforce this, we take the canonical representative $\overline{w}_1$ of $\cW_j$ for each $j\in[q]$, write $\overline{w}_j^\sL$ for the corresponding column of $2p$ left entries, and define
    \begin{align*}
        \overline{Z}&=
        \left[
        \begin{matrix}
            \vert & \vert & & \vert\\
            \overline{w}_1^\sL & \overline{w}_2^\sL  & \cdots & \overline{w}_q^\sL\\
            \vert & \vert & & \vert
        \end{matrix}
        \right]\\
        &=\left[
        \begin{matrix}
            \overline{z}_{1,1} & \overline{z}_{1,2} & \cdots & \overline{z}_{1,q}\\
            \overline{z}_{2,1} & \overline{z}_{2,2} & \cdots & \overline{z}_{2,q}\\
            \vdots & \vdots & \ddots & \vdots\\
            \overline{z}_{2p,1} & \overline{z}_{2p,2} & \cdots & \overline{z}_{2p,q}
        \end{matrix}
        \right]
    \end{align*}

    These canonical labels specify equality patterns---the rows of $\overline{Z}$ need not belong to $X$ or $Y$.

    To go from columns, which define each $w_j$, to rows, which are chosen by $x_t$ and $y_t$, we define for a row index $s\in[2p]$ the set
    \begin{align*}
        F_s=\{j\in[q]:\overline{z}_{s,j}\notin\{\overline{z}_{1,j},\dots,\overline{z}_{s-1,j}\}\},
    \end{align*}
    which gives the column indices in row $s$ of elements whose symbol first appears at position $s$ in their column. This allows us to rewrite $W_\sL$ as
    \begin{align*}
        W_\sL&=\sum_{j=1}^q\wt_{\sL}(\cW_j)\\
        &=\sum_{j=1}^q\sum_{s=1}^{2p}\Ind\{\overline{z}_{s,j}\notin\{\overline{z}_{1,j},\dots,\overline{z}_{s-1,j}\}\}\\
        &=\sum_{s=1}^{2p}\abs{F_s}.\numberthis\label{eq:W_L=sum F_s}
    \end{align*}
    Thus, once the rows $Z_1,\dots,Z_{s-1}$ have been fixed, the coordinates of $Z_s$ in $[q]\setminus F_s$ are also fixed because $w_j\sim\overline{w}_j$ for all $j\in[q]$, so we are only free to choose the value of the coordinates of $Z_s$ in $F_s$. If $\abs{F_s}=0$, then $Z_s$ is already fixed and there is at most one choice. If $\abs{F_s}>0$, then since both $X$ and $Y$ are $\eta$-sparse, from \Cref{fact:number of completions in sparse set} we have that there are at most $\eta n^{\abs{F_s}}$ possible values for $Z_s$. Note that in this counting argument, we neither required the non-left-backtracking requirement that $x_{t,j}\neq y_{t,j}$ and $y_{t,j}\neq x_{t+1,j}$ nor that unfixed indices in $F_s$ actually be chosen to not match the previous elements in their column. Loosening our requirements just gives us a weaker upper bound, which is acceptable for our argument.

    Finally, multiplying the number of choices for each row $s\in[2p]$ gives us at most
    \begin{align*}
        n^{\sum_{s=1}^{2p}\abs{F_s}}\eta^{\sum_{s=1}^{2p}\Ind\{\abs{F_s}>0\}}
    \end{align*}
    many assignments. To simplify this expression, we notice that trivially $\abs{F_s}\leq q$ so by  \Cref{eq:W_L=sum F_s} we have $ W_\sL\leq q\sum_{s=1}^{2p}\Ind\{\abs{F_s}>0\} $, meaning that $\sum_{s=1}^{2p}\Ind\{\abs{F_s}>0\} \geq W_\sL / q$. Along with another application of \Cref{eq:W_L=sum F_s} this gives us the upper bound
    \begin{align}
        n^{\sum_{s=1}^{2p}\abs{F_s}}\eta^{\sum_{s=1}^{2p}\Ind\{\abs{F_s}>0\}}\leq n^{W_\sL}\eta^{W_\sL / q}\label{eq:UB on number x and y}
    \end{align}
    Taking the product of the upper bound on the number of assignments of right vertices in \Cref{eq:number of possible right labelings} and the upper bound on the number of assignments of left vertices from \Cref{eq:UB on number x and y} gives us \Cref{eq:UB on cW_1 times ... times cW_q}.
\end{proof}
We take one step further and multiply by the coefficient $\prod_{j=1}^q k_j^{-2p}$.
\begin{lemma}\label{lem:bound on prod of k_j^-2p and prod of equivalence classes}
    With the assumptions of \Cref{prop:UB on op norm of max T_lambda} and $\eta$-sparse $X,Y\subseteq[n]^q$ and integers $\Delta_1,\dots,\Delta_q\in\{0,\dots,2p\}$, for any equivalence classes $(\cW_1,\dots,\cW_q)\in\fW_p(\Delta_1)\times\cdots\times\fW_p(\Delta_q)$, we have that
    \begin{align}
        \left(\prod_{j=1}^q k_j^{-2p}\right)\abs{\scW_{X,Y}(\cW_1,\dots,\cW_q)}\leq\rho^{-2pq}n^q\eta^{p+1}(n\eta^{1/q})^{-\Delta}\label{eq:product of k_j^-2p and prod equivalence classes UB}
    \end{align}
    where $\Delta=\sum_{j=1}^q \Delta_j$. 
\end{lemma}
\begin{proof}
    By \Cref{lem:UB on size of product of equivalence classes}, we have
    \begin{align*}
         \left(\prod_{j=1}^q k_j^{-2p}\right)\abs{\scW_{X,Y}(\cW_1,\dots,\cW_q)}&\leq\left(\prod_{j=1}^q k_j^{-2p}\right)n^{W_\sL}\eta^{W_\sL/q}\prod_{j=1}^q k_j^{\wt_\sR(\cW_j)}\\
         &=n^{W_\sL}\eta^{W_\sL/q}\prod_{j=1}^q k_j^{\wt_\sR(\cW_j)-2p}\numberthis\label{eq:UB on prod of equivalence classes minus 2p}
    \end{align*}
    \Cref{eq:wt_sR <= p} gives us that $\wt_\sR(\cW_j)\leq p$, so using the fact that $k_j\geq\rho n$ and $\rho\in(0,1)$ we can compute
    \begin{align}
        k_j^{\wt_\sR(\cW_j)-2p}\leq(\rho n)^{\wt_\sR(\cW_j)-2p}\leq \rho^{-2p}n^{\wt_\sR(\cW_j)-2p}. \label{eq:UB on k_j^wt_j-2p}
    \end{align}
    Therefore, over all $j\in[q]$ we get
    \begin{align*}
        \prod_{j=1}^q k_j^{\wt_\sR(\cW_j)-2p}&\leq \prod_{j=1}^q \rho^{-2p}n^{\wt_\sR(\cW_j)-2p} &\text{By \Cref{eq:UB on k_j^wt_j-2p}}\\
        &=\rho^{-2pq}n^{W_\sR-2pq}\numberthis\label{eq:UB on prod over all j of k_j^wt_j-2p}
    \end{align*}
    Moreover, by \Cref{eq:wt_sL >= p+1-Delta} we have that $W_\sL\geq q(p+1)-\Delta$ so
    \begin{align}
        \eta^{W_\sL/q}\leq\eta^{p+1-\Delta/q}\label{eq:UB on eta^W_sL/q}
    \end{align}
    since $\eta\in(0,1)$. By the definition of the deficiencies, we also have that
    \begin{align}
        W_\sL+W_\sR\leq q(2p+1)-\Delta\label{eq:W_L+W_R <= q(2p+1)-delta}
    \end{align}
    which allows us to simplify \Cref{eq:UB on prod of equivalence classes minus 2p} as
    \begin{align*}
        \left(\prod_{j=1}^q k_j^{-2p}\right)\abs{\scW_{X,Y}(\cW_1,\dots,\cW_q)}&\leq n^{W_\sL}\eta^{W_\sL/q}\prod_{j=1}^q k_j^{\wt_\sR(\cW_j)-2p}\\
        &\leq n^{W_\sL}\cdot \eta^{W_\sL/q}\cdot \rho^{-2pq}n^{W_\sR-2pq} &\text{By \Cref{eq:UB on prod over all j of k_j^wt_j-2p}}\\
        &\leq n^{W_\sL+W_\sR-2pq}\cdot \eta^{p+1-\Delta/q}\cdot\rho^{-2pq} &\text{By \Cref{eq:UB on eta^W_sL/q}}\\
        &\leq n^{q-\Delta}\cdot \eta^{p+1-\Delta/q}\cdot\rho^{-2pq} &\text{By \Cref{eq:W_L+W_R <= q(2p+1)-delta}}\\
        &=\rho^{-2pq}\cdot  n^q\cdot \eta^{p+1}(n\eta^{1/q})^{-\Delta},
    \end{align*}
    matching \Cref{eq:product of k_j^-2p and prod equivalence classes UB}.
\end{proof}

We are now ready to bound $\E_\sigma[\trace((TT^\transpose)^p)]$. 
\begin{lemma}\label{lem:UB on expectation of trace of power}
    With the setup of \Cref{prop:UB on op norm of max T_lambda}, for any $\lambda\in\Lambda$ let $X=X_\lambda,Y=Y_\lambda\subseteq[n]^q$ be $\eta$-sparse subsets and set $p=\ceil{D\log(n)}$.
    Then for all $n\geq N_{D,q,\eta}$ we have
    \begin{align*}
        \E_\sigma[\trace((T_\lambda T_\lambda^\transpose)^p)]\leq 2^q n^q \eta^{p+1} (16^q \rho^{-2q})^p.
    \end{align*}
    Note that, while we set $p$ here, we do not fully eliminate it to make cancellations in subsequent computations cleaner.
\end{lemma}
\begin{proof}
    Let $T=T_\lambda$. \Cref{cor:Expectation of trace (TT^transpose)^p in terms of even non-left-backtracking cycles} and \Cref{lem:scW disjoint decomposition} together give us the expression
    \begin{align*}
        \E_\sigma[\trace((TT^\transpose)^p)]=\sum_{(\Delta_1,\dots,\Delta_q)\in\{0,\dots,2p\}^q}\sum_{(\cW_1,\dots,\cW_q)\in\fW_p(\Delta_1)\times\cdots\times\fW_p(\Delta_q)} \left(\prod_{j=1}^q k_j^{-2p}\right)\abs{\scW_{X,Y}(\cW_1,\dots,\cW_q)}.
    \end{align*}
    Using the result of \Cref{lem:bound on prod of k_j^-2p and prod of equivalence classes} yields 
    \begin{align}
        \E_\sigma[\trace((TT^\transpose)^p)]&\leq \sum_{(\Delta_1,\dots,\Delta_q)\in\{0,\dots,2p\}^q}\sum_{(\cW_j)_j\in\prod_{j=1}^q\fW_p(\Delta_j)} \rho^{-2pq}n^q\eta^{p+1}(n\eta^{1/q})^{-\Delta}, \label{eq:expectation of TT^T^p as sum of sum of consts}
    \end{align}
    where we recall that $\Delta=\sum_{j=1}^q \Delta_j$.
    Because the summand in \Cref{eq:expectation of TT^T^p as sum of sum of consts} does not depend on $\cW_1,\dots,\cW_q$, the second sum just contributes a factor of
    \begin{align}
        \prod_{j=1}^q\abs{\fW_p(\Delta_j)}&\leq\prod_{j=1}^q 16^p(4p)^{6\Delta_j}\\
        &=16^{pq}(4p)^{6\Delta}.\label{eq:UB on prod of cardinality of sets of equivalence classes}
    \end{align}
    All that is left now is to deal with the first sum of \Cref{eq:expectation of TT^T^p as sum of sum of consts}. We first observe some useful inequalities by our choice of parameters. Because $p=\ceil{D\log(n)}$, we have $4p\leq 8D\log(n)$. Thus, using the inequality
    \begin{align*}
        \log(n)\leq\frac{n^\alpha}{\alpha e}
    \end{align*}
    for $n\geq 1$ and $\alpha>0$ we can set $\alpha=1/24$ to compute
    \begin{align*}
        (4p)^6&\leq (8D\log(n))^6\\
        &\leq \left(8D\cdot\frac{24n^{1/24}}{e}\right)^6\\
        &=\left(\frac{192D}{e}\right)^6n^{1/4}\\
        &\leq\frac{n^{1/2}}{2},\numberthis\label{eq:UB on (4p)^6}
    \end{align*}
    where the last inequality is due to our choice of $n\geq N_{D,q,\eta}\geq 16\left(\frac{192D}{e}\right)^{24}$. Consequently, using the fact that $\eta\geq n^{-q/2}$ so $\eta^{1/q}\geq n^{-1/2}$ we can use \Cref{eq:UB on (4p)^6} to get
    \begin{align}
        \frac{(4p)^6}{n\eta^{1/q}}&\leq\frac{1}{2}.\label{eq:UB on (4p)^6 / n eta ^1/q}
    \end{align}
    Finally, we put all these inequalities to use to reduce \Cref{eq:expectation of TT^T^p as sum of sum of consts} as follows:
    \begin{align*}
        \E_\sigma[\trace((TT^\transpose)^p)]&\leq \sum_{(\Delta_1,\dots,\Delta_q)\in\{0,\dots,2p\}^q}\sum_{(\cW_j)_j\in\prod_{j=1}^q\fW_p(\Delta_j)} \rho^{-2pq}n^q\eta^{p+1}(n\eta^{1/q})^{-\Delta}\\
        &\leq \sum_{(\Delta_1,\dots,\Delta_q)\in\{0,\dots,2p\}^q} \left(16^{pq}(4p)^{6\Delta}  \right)\left(\rho^{-2pq}n^q\eta^{p+1}(n\eta^{1/q})^{-\Delta}\right) &\text{By \Cref{eq:UB on prod of cardinality of sets of equivalence classes}}\\
        &= n^q\eta^{p+1}(16^q\rho^{-2q})^p\prod_{j=1}^q \sum_{\Delta_j=0}^{2p}  \left(\frac{\left(4p\right)^6}{n\eta^{1/q}}\right)^{\Delta_j}\\
        &\leq n^q\eta^{p+1}(16^q\rho^{-2q})^p\prod_{j=1}^q \sum_{\Delta_j=0}^\infty  \left(\frac{\left(4p\right)^6}{n\eta^{1/q}}\right)^{\Delta_j}\\
        &\leq n^q\eta^{p+1}(16^q\rho^{-2q})^p\prod_{j=1}^q \sum_{\Delta_j=0}^\infty \left(\frac{1}{2}\right)^{\Delta_j} &\text{By \Cref{eq:UB on (4p)^6 / n eta ^1/q}}\\
        &= 2^q n^q\eta^{p+1}(16^q\rho^{-2q})^p,
    \end{align*}
    as claimed.
\end{proof}

\subsection{Putting it all together}
We can  finally prove \Cref{prop:UB on op norm of max T_lambda}.

\maxopUB*

\begin{proof}
    As in \Cref{lem:UB on expectation of trace of power}, we let $p=\ceil{D\log(n)}$.
    By H\"older's inequality,
    \begin{align*}
        \E_\sigma\left[\max_{\lambda\in\Lambda}\opnorm{T_\lambda}\right]&=\E_\sigma\left[1\cdot\max_{\lambda\in\Lambda}\opnorm{T_\lambda}\right]\\
        &\leq (\E_\sigma[1^{2p/(2p-1)}])^{(2p-1)/2p}\cdot\left(\E_\sigma\left[\left(\max_{\lambda\in\Lambda}\opnorm{T_\lambda}\right)^{2p}\right]\right)^{1/(2p)}\\
        &=\left(\E_\sigma\left[\max_{\lambda\in\Lambda}\opnorm{T_\lambda}^{2p}\right]\right)^{1/(2p)}.\numberthis\label{eq:exp of max after holder}
    \end{align*}
    Next, we use a summation and \Cref{fact:opnorm^2p is bounded by trace^p} to bound the maximum as
    \begin{align*}
        \max_{\lambda\in\Lambda}\opnorm{T_\lambda}^{2p}&\leq\sum_{\lambda\in\Lambda}\opnorm{T_\lambda}^{2p}\\
        &\leq\sum_{\lambda\in\Lambda}\trace((T_\lambda T_\lambda^\transpose)^p).\numberthis\label{eq:max opnorm to sum traces}
    \end{align*}
    We now put \Cref{eq:exp of max after holder,eq:max opnorm to sum traces} together with the assumption that $\abs{\Lambda}\leq n^\nu$ to compute
    \begin{align*}
        \E_\sigma\left[\max_{\lambda\in\Lambda}\opnorm{T_\lambda}\right]&\leq\left(\E_\sigma\left[\sum_{\lambda\in\Lambda}\trace((T_\lambda T_\lambda^\transpose)^p)\right]\right)^{1/(2p)}\\
        &=\left(\sum_{\lambda\in\Lambda}\E_\sigma\left[\trace((T_\lambda T_\lambda^\transpose)^p)\right]\right)^{1/(2p)}\\
        &\leq\left(\sum_{\lambda\in\Lambda} 2^q n^q \eta^{p+1} (16^q \rho^{-2q})^p\right)^{1/(2p)} &\text{By \Cref{lem:UB on expectation of trace of power}}\\
        &\leq \left( n^\nu 2^q n^q \eta^{p+1} (16^q \rho^{-2q})^p\right)^{1/(2p)}\\
        &=2^{q/(2p)}n^{(q+\nu)/(2p)}\eta^{(p+1)/(2p)}\left(\frac{4}{\rho}\right)^q \numberthis\label{eq:UB on expectation of max of opnorms}
    \end{align*}
    Using our assumption that $q+\nu\leq D$, $q\leq D$, and $p\geq D\log n$ we can simplify the first two terms of \Cref{eq:UB on expectation of max of opnorms} as
    \begin{align*}
        2^{q/(2p)}n^{(q+\nu)/(2p)}&\leq 2^{\frac{D}{2D\log n}} n^{\frac{D}{2D\log n}}\\
        &=\exp\left(\frac{\log 2}{2\log n}+\frac{1}{2}\right)\\
        &< 2\numberthis\label{eq:UB on 2^powers n^powers}
    \end{align*}
    where the last inequality uses $n\geq N_{D,q,\eta}\geq16$.
    To bound the third term of \Cref{eq:UB on expectation of max of opnorms}, we note that $\frac{p+1}{2p}\geq 1/2$ and $0<\eta\leq 1$ so
    \begin{align}
        \eta^{(p+1)/(2p)}&\leq\sqrt{\eta}.\label{eq:UB on overline eta^p+1/2p}
    \end{align}
    Plugging \Cref{eq:UB on overline eta^p+1/2p,eq:UB on 2^powers n^powers} into \Cref{eq:UB on expectation of max of opnorms} yields our desired claim.
\end{proof}

\dobib

\section{Open Questions}\label{sec:open questions}

There are several immediate directions that we think would be interesting and worthwhile to explore.

\begin{enumerate}
    \item The transference of our result from the reals to finite fields in the proof of \Cref{thm:cores_M injective whp} requires our base field to have characteristic $O(2^{n^D}n^D\log n)$ instead of two as obtained by \cite{kalachev_maximally_2025}. Can this be avoided and our method be made to yield the same result over an extension field of characteristic two?

    \item Different parametrizations of the generator matrices used to define a family of collections of codes can force the component codes to have desired groups contained in their automorphism groups. Can our techniques be used to analyze such a model? And if so, what automorphism groups can be achieved? The main difficulty here is that we explicitly use the independence of the signs in our random Rademacher generator matrices in \Cref{sec:expected op norm of K,sec:bounding max opnorm with wigner}, and non-trivial parametrizations might introduce relations amongst these signs.

    \item Besides analyzing the product expansion of tensor codes, one could analyze the coboundary expansion of more general sheaf codes that could be used as the local codes on other high-dimensional expanders, such as simplicial HDXs. In the formalization of \cite{panteleev_maximally_2024}, all of the main components of our analysis generalize. Corestriction maps and extendable/inner-generated sets are still well-defined, planes become the set of top-level neighbors of an element in a poset, and we can still use our random Rademacher model for the generator matrices. However, the clean tensor decomposition of blocks of the plane-interaction matrix from \Cref{sec:expected op norm of K} is no longer possible, and any other decomposition will depend on the exact sheaf codes being studied. It is especially tempting to combine this question with the last one to ask what automorphisms we can require these codes to have.

    \item Similar to the last couple questions, can our techniques be used to analyze polynomially parametrized families of codes that result in instances having the multiplication property? This would be interesting in both the cubical and simplicial settings to achieve global codes with the multiplication property. As previously mentioned, this is relevant when trying to construct quantum codes that are useful for computation \cite{nguyen_good_2025}.
\end{enumerate}

\dobib

\printbibliography

\end{document}